\documentclass[11pt]{article}
\usepackage{fullpage}
\usepackage{helvet}
\usepackage{courier}
\usepackage{float}

\usepackage{epigraph}
\usepackage{nicefrac}
\usepackage{hyperref}
\hypersetup{colorlinks=true,linkcolor=black,citecolor=blue}

\usepackage[margin=1in]{geometry}
\usepackage{booktabs} 
\usepackage[ruled]{algorithm2e} 

\SetAlFnt{\small}
\SetAlCapFnt{\small}
\SetAlCapNameFnt{\small}
\SetAlCapHSkip{0pt}
\IncMargin{-\parindent}

\usepackage{nicefrac}
\usepackage{amsthm,amsfonts,amstext, comment, graphicx}

\usepackage[shortlabels]{enumitem}
\usepackage{bm}
\usepackage{algpseudocode}

\usepackage{hyperref}
\usepackage{color}
\usepackage{float}
\usepackage{bbm}
\floatstyle{boxed}
\newfloat{algorithm}{t}{lop}
\usepackage{natbib}
\usepackage{tikz}
\usepackage{varwidth}
\usepackage{subfigure}

\usepackage{blkarray}

\usepackage{multirow}

\usepackage{mathtools}
\usepackage{float}
\usepackage[classfont=sanserif,langfont=roman,funcfont=italic]{complexity}
\usepackage{caption}
\usepackage{subfigure}
\usepackage{tcolorbox}
\usepackage{color-edits}
\addauthor{ss}{violet} 
\addauthor{sb}{violet} 
\addauthor{mh}{blue} 
\addauthor{rp}{olive} 
\addauthor{kw}{red} 

\newtheorem{definition}{Definition}

\usepackage{url}
\renewcommand{\vec}[1]{\mathbf{#1}}

\newcommand{\explain}[1]{\tag{\textcolor{gray}{#1}}}

\hypersetup{colorlinks=true,linkcolor=blue,citecolor=blue,anchorcolor = blue}

\definecolor{ForestGreen}{rgb}{.05,.50,.05}
\definecolor{darkmagenta}{rgb}{0.55, 0.0, 0.55}

\usepackage{authblk}

\newcommand{\sbr}[1]{\left[\,#1\,\right]}

\renewcommand{\eqref}[1]{(\ref{#1})}

\newcommand{\eps}{\varepsilon}

\DeclareMathOperator*{\argmax}{argmax}

\usepackage{thmtools}
\usepackage{thm-restate}

\declaretheorem[name=Lemma]{lemma}

\title{Dueling over Multiple Pieces of Dessert\footnote{This research was supported  by the US National Science Foundation grant CCF-2238372.}}

\author{Simina Br\^anzei\footnote{Purdue University and Google Research. E-mail: \texttt{\url{simina.branzei@gmail.com}}.} \enspace and \enspace Reed Phillips\footnote{Purdue University. E-mail: \texttt{\url{phill289@purdue.edu}}.}}

\begin{document}


\maketitle

\begin{abstract}
We study the dynamics of repeated fair division between two players, Alice and Bob, where Alice partitions a cake into two subsets and Bob chooses his preferred one over $T$ rounds. Alice aims to minimize her regret relative to the Stackelberg value---the maximum utility she could achieve if she knew Bob's private valuation.

We  show that if Alice uses arbitrary measurable partitions, achieving strongly sublinear regret is impossible; she suffers a regret of $\Omega\Bigl(\frac{T}{\log^2 T}\Bigr)$ regret even against a myopic Bob. However, when Alice uses at most $k$ cuts, the learning landscape becomes tractable.  We analyze Alice's performance based on her knowledge of Bob's strategic sophistication (his regret budget). When Bob's learning rate is public, we establish a hierarchy of polynomial regret bounds determined by $k$ and Bob's regret budget. In contrast, when this learning rate is private, Alice can universally guarantee $O\Bigl(\frac{T}{\log T}\Bigr)$ regret, but any attempt to secure a polynomial rate $O(T^\beta)$ (for $\beta < 1$) leaves her vulnerable to incurring strictly linear regret against some Bob.

Finally, as a corollary of our online learning dynamics, we characterize the randomized query complexity of finding approximate Stackelberg allocations with a constant number of cuts in the  Robertson-Webb  model.
\end{abstract}

\section{Introduction}

The problem of fair division, introduced by \citet{Steinhaus48}, provides a model for allocating a heterogeneous resource among agents with differing preferences. A well known mechanism is the \emph{Cut-and-Choose} protocol: one player (Alice) partitions the resource into two pieces, and the other (Bob) selects his preferred piece. In a standard one-shot interaction where valuations are private, this protocol incentivizes Alice to divide the cake into pieces she values equally, to ensure that she receives a fair share regardless of Bob's choice.

However, if Alice possesses knowledge of Bob's valuation, the interaction transforms from a problem of equity into one of strategic optimization. Alice can act as a Stackelberg leader, constructing a partition that leaves Bob just satisfied enough to choose a specific piece, thereby securing the remainder---and the maximum possible utility---for herself. We refer to this optimal utility as Alice's \emph{Stackelberg value}. While private information typically precludes this outcome in a single instance, the dynamics change significantly in repeated settings. When the allocation occurs over multiple rounds—such as the daily assignment of computing resources or the scheduling of shifts—Alice need not know Bob's valuation \emph{a priori}. Instead, she can treat the interaction as an online learning problem, experimenting with different partitions and using Bob's choices to refine her model of his preferences.

This perspective connects repeated fair division to a growing literature on learning in Stackelberg-style
interactions~\cite{balcan2015commitment, tambe2011security,dong2018strategic,hajiaghayi2023regret,kleinberg2003value,birmpas2020optimally,gan2019imitative, zhao2023online,Donahue24Decentralized}.
\citep{Donahue24Decentralized} analyze {decentralized learning} in Stackelberg games, showing that classical benchmarks
based on Stackelberg equilibrium payoffs can be unattainable in general and may force worst-case linear regret. This 
motivates a careful study of what is learnable under what structural conditions.
Complementary work on {steering} or exploiting no-regret learners  (see, e.g., \citep{Zhang2025Steer}) identifies information-theoretic barriers:
without sufficient knowledge about the learner's algorithm or objectives, steering toward Stackelberg outcomes
can be impossible, while additional structure can restore learnability.

This sequential setting for cake cutting was explored in 
\citep{Branzei24Dueling}, which showed that Alice can efficiently learn to exploit a myopic chooser to approach her Stackelberg value. However, this analysis has largely been restricted to the case where Alice makes a single cut. While this constraint simplifies the learning landscape to a one-dimensional search, it makes sense that a proposer can have the flexibility to create allocations consisting of multiple disjoint intervals. This added flexibility implies that the single-cut Stackelberg value is only a lower bound on what is achievable. For instance, if Alice values the ends of a time interval and Bob values the center, a single cut forces Alice to cede high-value territory. However, if permitted to make two cuts, she can offer Bob the center while retaining both ends, strictly increasing her utility.

This generalization introduces a fundamental trade-off between potential utility and learning complexity. While complex partitions allow Alice to extract significantly higher value, they also exponentially expand the strategy space. In the single-cut setting, learning Bob's valuation is akin to identifying a threshold on a line. As the number of allowed cuts increases, Bob's preferences can be obscured within increasingly intricate geometric structures. In the limit, where Alice is permitted to partition the cake into arbitrary measurable sets, the space of possible strategies becomes so rich that the information gained from Bob's binary choices may be insufficient to resolve his true preferences.

In this paper, we study the limits of strategic learning in repeated fair division under varying regimes of geometric complexity. We model the interaction as a repeated game between a learner (Alice) and a responder (Bob) and analyze Alice's \emph{Stackelberg regret}—the difference between her cumulative utility and the optimal value she could achieve with full information. We consider both the \emph{$k$-cut game}, where partitions are defined by a finite number of cuts, and the general \emph{measurable-cut game}. Our goal is to determine how the complexity of the allowable partitions impacts the learning rate and to characterize the conditions under which the Stackelberg optimum is learnable.

\section{Model} \label{sec:model}
	
We consider a cake represented by the interval $[0,1]$ and two players, Alice ($A$) and Bob ($B$). Each player $i$ has an integrable private value density function $v_i : [0,1] \to \mathbb{R}_+$.

A \emph{piece} of cake is a Lebesgue-measurable subset $S \subseteq [0, 1]$.
The value of player $i$ for a piece $S$ is denoted $V_i(S) = \int_{S} v_i(x) \; \textit{dx}$. Without loss of generality, valuations are normalized such that $V_i([0, 1]) = 1$.
We assume the densities are bounded: there exist constants $\Delta \geq \delta > 0$ such that $\delta \leq v_i(x) \leq \Delta$ for all $x \in [0,1]$. Let $\mathcal{V}$ denote the set of all such density functions.
		
\paragraph{Allocations.} An allocation $Z = (Z_A, Z_B)$ is a partition of the cake into two disjoint pieces $Z_A$ and $Z_B$ such that $Z_A \cup Z_B = [0,1]$. The allocation is \emph{envy-free} if $V_i(Z_i) \geq V_i(Z_j)$ for all $i,j$. Since valuations are normalized, this is equivalent to $V_i(Z_i) \geq 1/2$.

For  $k \in \mathbb{N}^*$, a \emph{$k$-cut} is defined by a sequence of cut points $0 \leq a_1 \leq a_2 \leq \ldots \leq a_k \leq 1$, which divides the cake into at most $k+1$ intervals.
Alice's \emph{Stackelberg value}, denoted $u_A^*(k)$, is the maximum utility she can guarantee in an allocation where she partitions the cake using $k$ cuts and Bob subsequently chooses his preferred piece (breaking ties in Alice's favor). At any such allocation, Bob's utility is exactly $1/2$. The Stackelberg value for $k=2$ cuts is illustrated in Figure~\ref{fig:two-cut-Stackelberg}.

Alice's Stackelberg value is similarly defined when the cake can be partitioned into any two Lebesgue-measurable sets; we denote this value $u_A^*(\infty)$.
	
\begin{figure}[h!]
	\centering
	\includegraphics[width=0.3\textwidth]{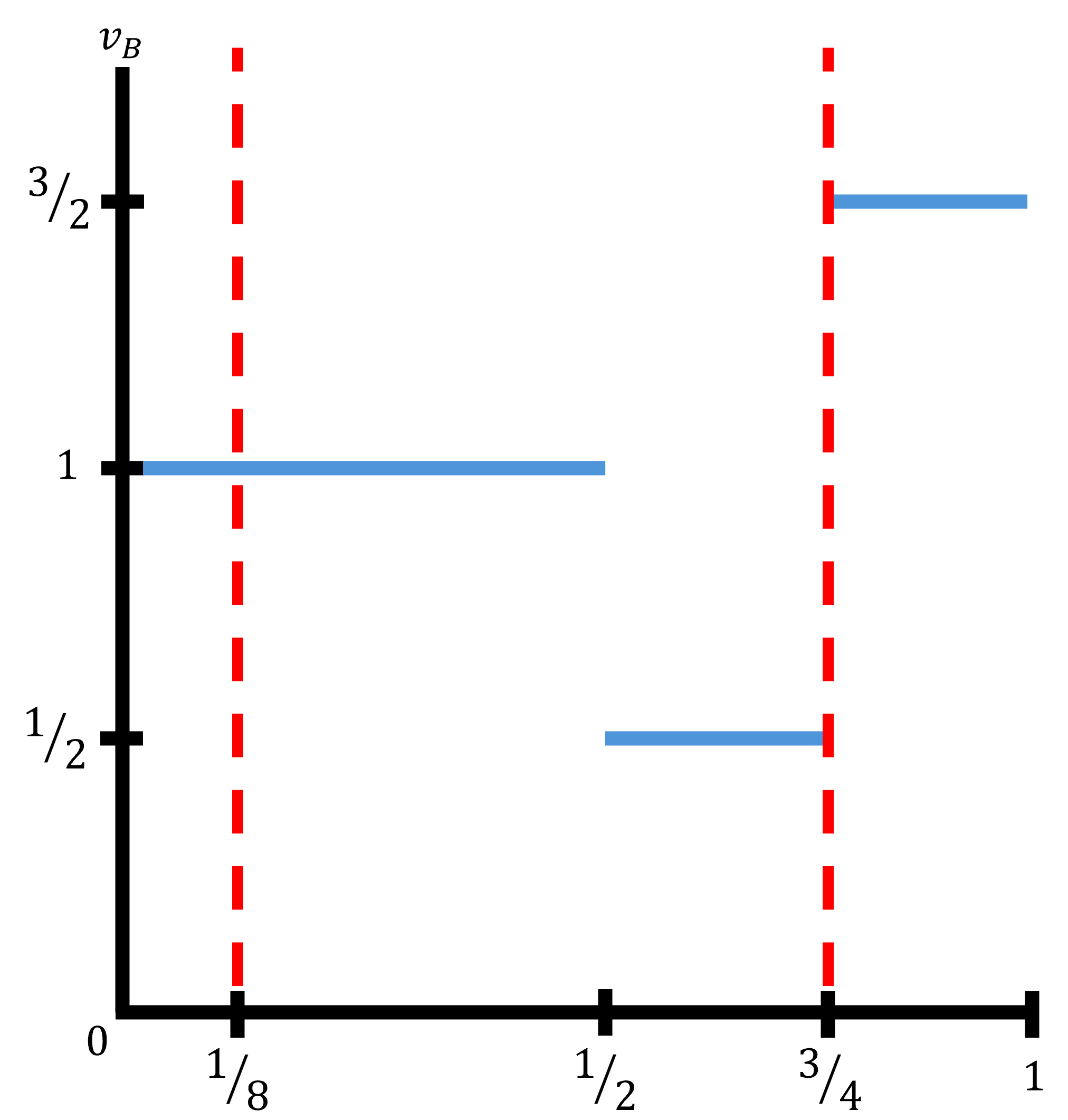}
	\caption{A Stackelberg cut for $k=2$. Bob's value density $v_B$ is plotted in blue. Alice's value density (not pictured) is assumed to be uniform: $v_A(x)=1$ $\forall x \in [0,1]$. The region between the red dashed cuts has value $1/2$ to Bob and maximizes Alice's value among all such regions; thus, $u_A^*(2) = 5/8$.}
	\label{fig:two-cut-Stackelberg}
\end{figure}
	
\subsection*{Repeated cake cutting} Let $k \in \mathbb{N}^*$ be fixed. The game proceeds over rounds $t = 1, 2, \ldots, T$. In each round:
\begin{itemize}
	\item A new cake arrives, identical to those in previous rounds.
	\item Alice cuts the cake at $k$ points of her choice and groups the resulting intervals into two pieces.
	\item Bob chooses one of the pieces and Alice receives the remainder.
\end{itemize}
We refer to this as the \emph{$k$-cut game}.
The game is similarly defined when Alice is allowed to partition the cake into any two Lebesgue-measurable sets, which we call the \emph{measurable-cut game}. An example of the first two rounds of a $2$-cut game is shown in Figure \ref{fig:two-cut-game-example}.

\begin{figure}
    \centering
    \includegraphics[width=0.5\textwidth]{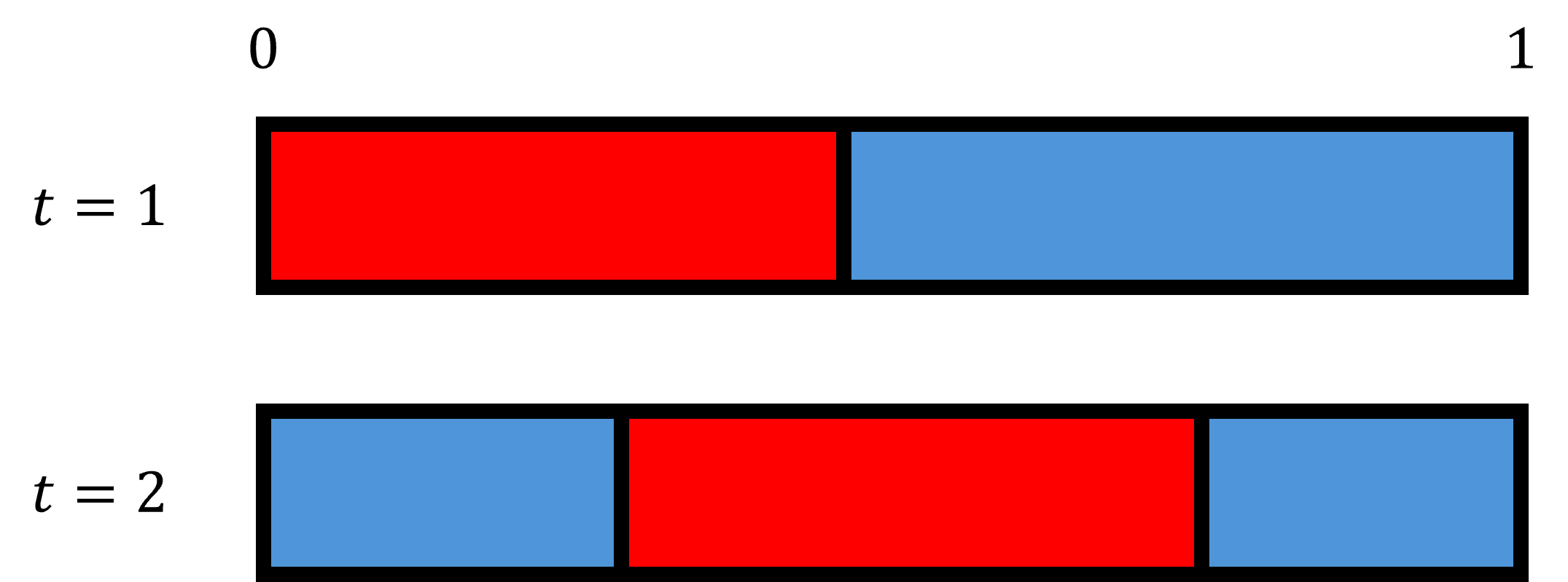}
    \caption{An example outcome of the first two rounds of a $2$-cut game. In each round, Alice partitions $[0, 1]$ into two pieces by making at most $2$ cuts. The pieces Bob chose are in blue. In this case, they are the pieces he would prefer under the value density from Figure \ref{fig:two-cut-Stackelberg}.} \label{fig:two-cut-game-example}
\end{figure}

\section{Our Results}

Our results examine how players fare in the repeated game over $T$ rounds.
Given a history $H$, Alice's Stackelberg regret is defined as
\begin{align} \label{eq:regret_stackelberg_informal}
	\text{Reg}_{A}(H) = \sum_{t=1}^{T} \left[ u_A^*(k) - u_A^t(H) \right],
\end{align}
where $u_A^*(k)$ is Alice's Stackelberg value with $k$ cuts (with $k = \infty$ for the measurable-cut game) and $u_A^t(H)$ is Alice's utility in round $t$ under history $H$.

Given a history $H$, Bob's regret is defined as
\begin{align} \label{eq:regret_Bob_informal}
	\text{Reg}_{B}(H) = \sum_{t=1}^{T} \left[ u_B^*(t) - u_B^t(H) \right],
\end{align}
where $u_B^*(t)$ is the maximum utility Bob could have obtained in round $t$ (by choosing his preferred piece) and $u_B^t(H)$ is his realized utility.

Suppose Alice uses strategy $S_A$ and Bob uses strategy $S_B$. We say $S_A$ ensures Alice's Stackelberg regret is at most $\gamma$ against $S_B$ if $\text{Reg}_{A}(H) \leq \gamma$ for all $T$-round histories $H$ consistent with $(S_A, S_B)$. Similarly, $S_B$ ensures Bob's regret is at most $\gamma$ against $S_A$ if $\text{Reg}_{B}(H) \leq \gamma$ for all such histories.

\subsection{The Measurable-Cut Game}

In the measurable-cut game---the most general setting we consider---we show that achieving strongly sublinear regret is impossible: Alice must incur a regret of order  $\Omega(T/(\log T)^2)$. Crucially, this holds even against a \emph{myopic} Bob who simply chooses his preferred piece in every round. Consequently, Alice cannot achieve regret better than $O(T^{1-\varepsilon})$ for any $\varepsilon > 0$.

\begin{restatable}{thm}{measurableMyopicLower}
    \label{thm:measurable-cut-myopic-lower-bound}
    Consider the measurable-cut game and suppose Bob is myopic. There exists $T_0 \in \mathbb{N}$ such that for every valuation and deterministic strategy of Alice, there exists a Bob valuation against which Alice's regret is $\Omega(T / (\log T)^2)$ for every  horizon $T \geq T_0$.
\end{restatable}

\subsection{The $k$-Cut Game}

Alice fares significantly better in the $k$-cut game. We start by analyzing the scenario where Bob behaves myopically. For constant $k$, our bounds are tight up to logarithmic factors.

\begin{restatable}[$k \geq 2$ cuts; Myopic Bob]{thm}{kCutMyopic}
    \label{thm:myopic_Bob}
    Let $k \geq 2$. Suppose Bob is myopic. There exists a deterministic strategy for Alice that guarantees her regret over  horizon $T$ is $O\bigl(\sqrt{Tk} \log(Tk)\bigr)$. Moreover, for every Alice valuation and deterministic Alice strategy, there exists a Bob valuation against which her regret is $\Omega\bigl(\frac{\sqrt{T}}{k^{3/2}}\bigr)$.
\end{restatable}

Next, we analyze the game when Bob is not necessarily myopic but employs a strategy with bounded regret relative to the benchmark of selecting his preferred piece in every round. We distinguish between the setting where his regret budget is public knowledge and the setting where it is private.

\subsubsection{Public Learning Rate}

We first consider the setting where Bob's regret exponent is known to both Alice and Bob.

For $k=2$ cuts, the next theorem  shows that when Alice plays against a Bob with regret bound $O(T^{\alpha})$, her optimal regret is $\widetilde{O}(T^{\frac{2+\alpha}{3}})$.

\begin{restatable}[2 cuts; Non-Myopic Bob with Public Learning Rate]{thm}{twoCutNonMyopic}\label{thm:2-cut-non-myopic}
    Let $k=2$. Suppose Bob's strategy guarantees his regret over horizon $T$ is at most $c  T^{\alpha}$, for $\alpha \in (-1/2, 1)$ and  $c > 0$. There exists a deterministic strategy for Alice, parameterized by $\alpha$ but independent of $c$, that guarantees her regret is $O(T^{\frac{2+\alpha}{3}} \log^{2/3} T)$.

    The exponent cannot be improved: for every Alice valuation and strategy, there exists a Bob valuation and an $O(T^{\alpha})$-regret strategy for Bob so that Alice's regret is  $\Omega\bigl(T^{\frac{2+\alpha}{3}}\bigr)$.
\end{restatable}

 For $k \geq 3$ cuts, we provide an upper bound of  $O\left(T^{\frac{3+\alpha}{4}}k^{3/4} (\log T)^{1/2}\right)$ and a lower bound of $\Omega\bigl(T^{\frac{2+\alpha}{3}}/k\bigr)$ for Alice's regret.

\begin{restatable}[$k \geq 3$ cuts; Non-Myopic Bob with Public Learning Rate]{thm}{kCutNonMyopic}
    \label{thm:k-cut-non-myopic-bounds}
    Let $k \geq 3$. Suppose Bob's strategy guarantees his regret over  horizon $T$ is at most $c  T^{\alpha}$, for $\alpha \in (-1, 1)$ and $c > 0$. There exists a deterministic strategy for Alice, parameterized by $\alpha$ but independent of $c$, that guarantees her regret is $O\left(T^{\frac{3+\alpha}{4}}k^{3/4} (\log T)^{1/2}\right)$.

    Moreover,  for every Alice valuation and deterministic Alice strategy, there exists a Bob valuation and an $O(T^{\alpha})$-regret strategy for Bob so that Alice's regret is $\Omega\bigl(T^{\frac{2+\alpha}{3}}/k\bigr)$.
\end{restatable}

\subsubsection{ Private Learning Rate}

When Alice does not know Bob's learning rate, she can still achieve $o(T)$ regret, regardless of the true value of $\alpha$. However, the rate is $O(T/\log{T})$ for constant $k$ and the exponent cannot be improved, as we show in the following theorem.

\begin{restatable}[$k \geq 2$ cuts; Non-Myopic Bob with Private Learning Rate]{thm}{unknownAlpha}
    \label{thm:unknown-alpha-bounds}
    Let $k \geq 2$. Suppose Bob's strategy guarantees his regret over  horizon $T$ is at most $c T^{\alpha}$, for $\alpha \in (-\infty, 1)$ and $c > 0$. There exists a deterministic strategy for Alice, independent of both $\alpha$ and $c$, that guarantees her regret is $O(Tk^{3/4}/\log T)$.
    
    The exponent cannot be improved when $k$ is constant: if for some constant $\beta < 1$ Alice's strategy guarantees her  regret is $O(T^{\beta})$ against such a Bob, then there exists a Bob valuation and strategy that guarantees Bob $O(T^{\beta})$ regret but ensures Alice's regret is $\Omega(T)$.
\end{restatable}

\subsection{Corollary for the Robertson-Webb Model}

Our analysis of the repeated game allows us to characterize the randomized query complexity of finding a Stackelberg solution with $k$ cuts---for constant $k$---in the standard Robertson-Webb (RW) query model for cake cutting.

In the Robertson-Webb model \citep{woeginger2007complexity,RW98}, the protocol interacts with the players using the following types of queries:
\begin{itemize}
	\item $\emph{\textit{Cut}}_i(\alpha)$: Player $i$ cuts the cake at a point $y$ such that $V_{i}([0,y]) = \alpha$, where $\alpha \in [0,1]$ is specified by the protocol.\footnote{Ties are resolved deterministically, typically by choosing the leftmost point satisfying the condition.} The point $y$ becomes a \emph{cut point}.
	\item $\emph{\textit{Eval}}_i(y)$: Player $i$ returns the value $V_{i}([0,y])$, where $y$ is a previously established cut point.
\end{itemize}
	
A protocol issues a sequence of cut and evaluate queries and eventually outputs an allocation demarcated by cut points discovered during its execution. 	The value of a piece $[x,y]$ can be determined with two Eval queries, $Eval_i(x)$ and $Eval_i(y)$.

Since computing an exact Stackelberg allocation with $k \geq 2$ cuts is generally infeasible with a finite number of queries, we focus on finding an approximation.

\begin{definition}
    \label{def:eps-stackelberg}
    An allocation $Z = (Z_A, Z_B)$ demarcated with $k$ cuts is \emph{$\eps$-Stackelberg} if:
    \begin{itemize} 
        \item $Z$ is envy-free; and 
        \item Alice's value is within $\eps$ of her optimal Stackelberg value: $V_A(Z_A) \geq u_A^*(k) - \eps$.
    \end{itemize}
\end{definition}

 The \emph{deterministic query complexity} of a problem is the total number of queries necessary and sufficient for a correct deterministic algorithm to find a solution.
The \emph{randomized query complexity} is the expected number of queries required to find a solution with probability at least $9/10$ for each input, where the expectation is taken over the coin tosses of the protocol.

\medskip 

We have the following bounds, which are matching when $k$ is constant.

\begin{restatable}{thm}{RWQueryModel}
    \label{thm:rw-query-model-bounds}
    Let $\varepsilon > 0$. The randomized query complexity of computing an $\eps$-Stackelberg allocation with $k \geq 2$ cuts in the RW query model is $O(k/\eps)$ and  $\Omega(1/\eps)$.
\end{restatable}

Recall that for $k=1$, the query complexity of finding an $\varepsilon$-Stackelberg allocation is $2$, since the Cut-and-Choose protocol (with Bob as the cutter) yields an exact Stackelberg allocation for Alice.

\section{Related Work} \label{sec:related_work}

\paragraph{The cake cutting model.} Introduced by \cite{Steinhaus48}, in the cake cutting model $n$ agents each have a private value density over a ``cake" represented by the interval $[0, 1]$. The goal is to divide the cake among the agents according to some fairness notion like envy-freeness. It has been extensively studied; see \cite{BT96, RW98,Moulin03,socialchoice_book,Pro13} for surveys.

\paragraph{Existence and computational complexity of envy-free allocations.} Envy-free allocations with contiguous pieces can be shown to exist for any number of players $n$ by Sperner's lemma or Brouwer's fixed-point theorem \cite{edward1999rental,Stromquist80}. The computational complexity of finding such allocations is open in many cases. Allowing for disconnected pieces, in the Robertson-Webb (RW) query model a lower bound of $\Omega(n^2)$ queries was given by \cite{Pro09} and an upper bound of $O(n^{n^{n^{n^{n}}}})$ queries was given by \cite{AM16}. Simpler algorithms are known for small numbers of agents, such as the $4$-agent algorithm of \cite{AFMPV18}. \cite{Cheze_whp_ef} gave an algorithm that, with high probability over a certain distribution of the agents' valuations, uses a polynomial number of queries.

On the other hand, there does not exist a protocol in the RW query model for finding an envy-free allocation with contiguous pieces for $n \geq 3$ players \cite{Stromquist08}. The easier task of finding $\varepsilon$-envy-free allocations was studied in works such as \cite{CDP13, PW17, BN19, HR23}. Allowing more general valuation functions, \cite{DQS12} showed that the problem is PPAD-complete. See \cite{GoldbergHS20,FHHH21,Halevi18} for surveys.

\paragraph{Strategic cake cutting.} Truthful cake cutting mechanisms have been studied in the model where agents directly reveal their preferences \cite{truth_justice_and_cake_cutting, BU2023103904, bei2022truthful, Tao22} and in the RW query model \cite{mossel_tamuz, BM15}. Equilibria with strategic agents were studied in \cite{branzei2013equilibrium, NY08}.

\paragraph{Repeated fair division.} A simplified version of our setting was one example analyzed in \cite{AumannMaschler} (page 243). They considered a cake that is uniform except for a single cherry. Alice is assumed to be indifferent to the cherry while Bob could either like or dislike it, and Alice has a prior over his type. They also restrict Alice to two specific cut locations. They analyze what range of outcomes the players can achieve in the repeated cut-and-choose game. This example was generalized by \cite{Branzei24Dueling} to what we call the $1$-cut game in our paper.

The possibility of achieving various fairness notions has been studied in this and related online settings. \cite{Igarashi24Repeated} consider fairness in the aggregate across many repeated divisions. \cite{benade2022dynamic} have a series of indivisible items arrive, each of which must be allocated immediately. Works such as \cite{walsh2011online, kash2014no, friedman2015dynamic} instead have a series of agents arrive and want to fairly divide a single divisible resource.

\paragraph{Learning in repeated Stackelberg games.} The concepts of Stackelberg games and equlibria were introduced by \cite{stackelberg1934marktform} in the context of firms entering a market. They have been applied in settings such as security games \cite{balcan2015commitment, tambe2011security}, online strategic classification \cite{dong2018strategic}, and online principal agent problems \cite{hajiaghayi2023regret}.

Our model also features the online learning task of each player gleaning information about the other player through their actions. \cite{kleinberg2003value} studied this task from the perspective of a seller posting prices for a sequence of buyers. \cite{birmpas2020optimally, gan2019imitative, zhao2023online} consider general games, but where one player starts with perfect information about their opponent and tries to use that asymmetry for their own gain. \cite{Donahue24Decentralized} consider a more cooperative setting where one player is assumed to use a sublinear-regret algorithm and the other tries to minimize their collective regrets.

\paragraph{Exploiting a no-regret learner.} Our results focus on settings where Bob's strategy has a regret bound that Alice tries to exploit. Her goal of reaching her Stackelberg value under such conditions has been studied in other games. \cite{deng2019strategizing} showed that in many bimatrix games, with knowledge of the learner's payoff function, it is possible for such an exploiter to get arbitrarily close to their Stackelberg value. \cite{Zhang2025Steer} showed that this knowledge is necessary in general, but gave sufficient conditions on the learner's algorithm that enable the exploiter to learn it on the fly. On the learner's side, \cite{Arunachaleswaran24Pareto} study which learning algorithms are Pareto optimal against all possible exploiters.

\paragraph{Corrupted stochastic bandits.} Alice's task of finding an optimal partition can be viewed as a bandit problem. In particular, our setting when Bob has a nonzero regret budget is similar to the setting of stochastic bandits with adversarial corruption introduced by \cite{Lykouris18}. Here, the learner (Alice) is faced with a multi-arm stochastic bandit, but the arm rewards go through an adversary (Bob) who can corrupt them before Alice sees them. His ability to do so is constrained by a budget $C$ of the total amount he can modify rewards by over the course of the game. The extension most similar to our setting is \cite{Kang23}, where Alice's options form a continuous space and the true rewards are Lipschitz. Our setting has two main differences from these corrupted bandits: Alice knows what the two possible payoffs are on each arm based on her own valuation, and Bob's corruption is limited not by the amount he changes Alice's payoff but by the amount he changes his own.

\paragraph{Strategic experimentation.} Continuing the bandit similarities, many works have considered the strategic experimentation setting where multiple agents play the same stochastic bandit simultaneously. \cite{bolton1999strategic} consider the perfect monitoring setting where agents see the other agents' actions and rewards. They found phenomena such as agents exploring new arms in order to encourage other agents to do the same. \cite{aoyagi1998mutual, aoyagi2011corrigendum} considered the imperfect monitoring setting where agents only see the other agents' actions and showed that, with two agents and two arms with discrete priors, the agents eventually agree which arm is better. See \cite{horner2017learning} for a survey.

\section{Preliminaries} \label{sec:prelim}
	
	In this section we formally define the notation needed for the proofs. Let $\mu(S)$ denote the Lebesgue measure of a measurable set $S$. Let $[n]$ denote the set $\{1, 2, \ldots, n\}$.
    
    Given a set of cuts $x_1, \ldots, x_k \in [0,1]$ with $0 \leq  x_1 \leq  \ldots x_k \leq  1$, 
   define $x_0 = 0$, $x_{k+1}=1$. The \emph{alternating partition} $Z = (Z_1, Z_2)$ is the division of the cake into two  pieces\footnote{These pieces overlap at the endpoints of the intervals, but as these are finitely many points they do not affect the non-atomic valuations.}, where $Z_1$ contains the intervals $[x_i, x_{i+1}]$ for even $i$, and $Z_2$ contains the intervals $[x_i, x_{i+1}]$ for odd $i$. An example is shown in Figure \ref{fig:alternating-partition}.
Let $m_A$ be Alice's midpoint of the cake and $m_B$ Bob's midpoint.

\begin{figure}[h!]
    \centering
    \includegraphics[width=0.9\textwidth]{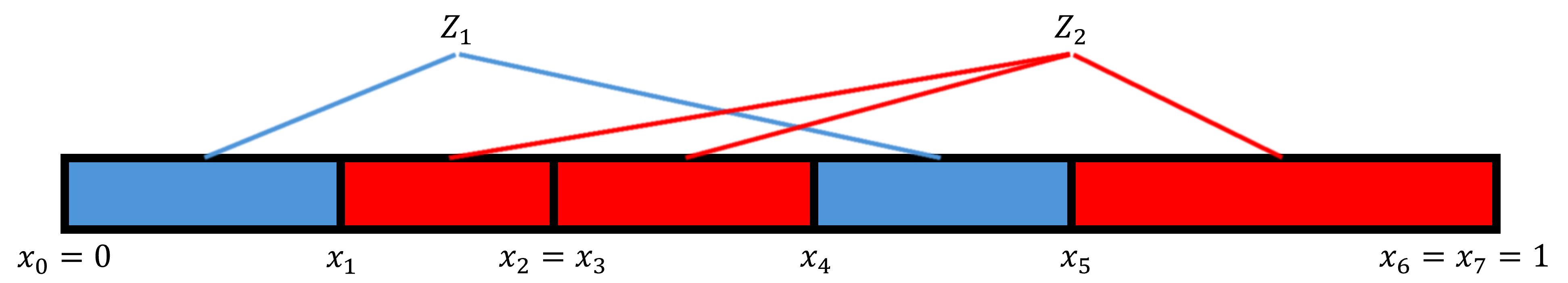}
    \caption{Illustration of an alternating partition for $k=6$. The positive-length intervals of $Z_1$ and $Z_2$ are colored blue and red, respectively. Setting some cut points equal to each other, as with $x_2=x_3$ and $x_6=x_7$, allows Alice to effectively make fewer than $k$ cuts if desired.}\label{fig:alternating-partition}
\end{figure}
    
	\subsection*{History} Recall $T$ is the number of rounds. In each round $t \in [T]$, Alice makes a sequence of cuts $0 \leq a_{t,1} \leq \ldots a_{t,k} \leq 1$, partitions the resulting intervals into two bins, and Bob chooses a bin $b_t \in \{1,2\}$. 
    
   It is without loss of generality to assume that she makes the alternating partition. Any allocation attainable with at most $k$ cuts can be simulated in this model by setting $a_{t,i} = 0$ for some of the indices $i$.

      Let $\vec{a}_t = (a_{t,1}, \ldots, a_{t,k})$ denote the cuts made by Alice  at time $t$. Let  $A_t = (a_1, \ldots, a_t)$ be the history of cuts until the end of round $t$ and $B_t = (b_1, \ldots, b_t)$ the history of choices made by Bob until the end of round $t$.
      
	A history $H = (A_T, B_T)$  denotes an entire trajectory of play.
	
	\subsection*{Strategies} 
	Let $\mathcal{A}= \left\{ a_1, \ldots, a_k \mid 0 \leq a_1 \leq \ldots \leq a_k \leq 1 \right\}$ denote  Alice's set of actions  and $\mathcal{B} = \{1,2\}$ denote Bob's set of actions in any given round.
	A pure strategy for Alice  at time $t$ is a function 
	\[ S_A^t : \mathcal{A}^{t-1} \times \mathcal{B}^{t-1} \times \mathcal{V} \times \mathbb{N} \to \mathcal{A},
	\] such that $S_A^t(A_{t-1}, B_{t-1}, v_A, T)$ is the next vector of cuts made by Alice as a function of the history $A_{t-1}$ of Alice's cuts, the history $B_{t-1}$ of Bob's choices, Alice's valuation $v_A$, and the  horizon $T$.
	
		A pure strategy for Bob at time $t$ is a function 
		\[ S_B^t:  \mathcal{A}^{t} \times \mathcal{B}^{t-1} \times \mathcal{V} \times \mathbb{N} \to \mathcal{B}\,.
		\]
		That  is, Bob observes Alice's cuts and then responds.		
	
	A pure strategy for Alice over the entire time horizon $T$ is denoted $S_A = (S_A^1, \ldots, S_A^T)$ and tells Alice what cuts to make at each time $t$. 
	A pure strategy for Bob over the entire time horizon $T$ is denoted $S_B = (S_B^1, \ldots, S_B^T)$ and tells Bob whether to choose bin $1$ or bin $2$ at each time $t$. 
	
	\subsection*{Rewards and utilities} 
	Suppose Alice has  strategy $S_A$ and Bob has  strategy $S_B$. 
	Let $u_A^t$ and $u_B^t$ be the utility (payoff) experienced by Alice  and Bob, respectively, at round $t$.  The utility of player $i \in \{A,B\}$ is denoted \[
	u_i = u_i(S_A, S_B) =  \sum_{t=1}^{T} u_i^t\,.
	\]
		
	Given a history $H$, let $u_i^t(H)$ be player $i$'s utility in round $t$ under  $H$ and let
	$u_i(H)= \sum_{t=1}^T u_i^t(H)$ be player $i$'s cumulative utility  under $H$.

We often analyze strategy profiles where Alice's valuation is uniform, i.e. $v_A(x) = 1$ for all $x \in [0,1]$. Briefly, we can do so by measuring the cake according to Alice's valuation rather than raw length. This argument is formalized in the following lemma, whose proof can be found in Appendix \ref{app:prelim}.

\begin{restatable}{lemma}{aliceValuationWarping}
    \label{lem:wlog-alice-valuation-warping}
    Let $v_A^1, v_B^1 \in \mathcal{V}$ and $(S_A^1, S_B^1)$ be a strategy pair for the repeated game in either the $k$-cut or measurable-cut model. For every $v_A^2 \in \mathcal{V}$, there exist a density $v_B^2: [0,1] \to [\delta^2/\Delta, \Delta^2/\delta]$ and strategy pair $(S_A^2, S_B^2)$ such that the sequence of Bob's choices and the realized utilities are identical under  $(v_A^1, v_B^1, S_A^1, S_B^1)$ and $(v_A^2, v_B^2, S_A^2, S_B^2)$.    
\end{restatable}

\section{The Measurable-Cut Game} \label{sec:measurable}

In this section we give a proof sketch of Theorem \ref{thm:measurable-cut-myopic-lower-bound}, which shows that in the measurable-cut game Alice cannot achieve $O(T^{\alpha})$ regret for any $\alpha < 1$ even against a myopic Bob. We restate it here for reference together with the proof sketch. The full proof can be found in Appendix \ref{app:measurable}.

\measurableMyopicLower*

If we flip the quantifiers in the statement of Theorem~\ref{thm:measurable-cut-myopic-lower-bound}, to say that for every $T \geq T_0$ there exists a Bob valuation, then the statement is immediate. For example, one could concentrate most of Bob's value in $2^T$ small high-value intervals, making it impossible for Alice to learn anything meaningful in just $T$ rounds. We provide a stronger construction, where Bob's valuation is independent of $T$.

\begin{proof}[Proof sketch of Theorem \ref{thm:measurable-cut-myopic-lower-bound}]
    We choose a function $g : \mathbb{N}^* \to [0, 1/2]$ that slowly converges to zero and use it to partition $(0, 1/2]$ into an infinite sequence of intervals $(g(2), g(1)], (g(3), g(2)], \ldots$. For concreteness we choose $g(n)=1/(2 + \log n)$, but this choice only matters for getting our specific regret bound. Let $L_i$ and $R_i$ be the left and right halves of $(g(i+1), g(i)]$, respectively. For each infinite bit vector $s \in \{0, 1\}^{\mathbb{N}^*}$, we then get a possible Bob valuation $v_B^s$ defined by:
    \begin{align}
        v_B^s(x) &= \begin{cases}
            2 & \text{if for some $i \in \mathbb{N}^*$, ($x \in L_i$ and $s_i=0$) or ($x \in R_i$ and $s_i=1$)} \\
            2/3 & \text{otherwise}
        \end{cases}
    \end{align}
    The constants are chosen so that such a Bob would be satisfied receiving only the portions of the cake with value density $2$. By Lemma \ref{lem:wlog-alice-valuation-warping}, we can assume Alice's value density is the uniform density function $v_A(x)=1$ $\forall x \in[0,1]$, so Alice's task is to learn the bits of $s$.

    Rather than reason about Alice's learning rate, though, we aim directly for her utility in each round. We show that it is possible to choose $s$ so that no partition Alice makes in a $T$-round game, for any $T$, is $O(1/\log^2 T)$-Stackelberg.

    To do so, we list out every partition Alice could make using her strategy $S_A$ across all time horizons. Because $S_A$ is deterministic and Alice only receives one bit of feedback from Bob in each round, she can only make at most $2^h$ different partitions in an $h$-round game. Index these partitions so that those made in an $h$-round game are partitions $2^h$ through $2 \cdot 2^h - 1$. We pick out Alice's less-preferred piece from each of these partitions to get a sequence of pieces $(B^j_1)_{j \geq 1}$. We further filter this sequence to those which give Alice any chance of getting a good payoff: we throw out any that came from partitions where, for all $s \in \{0, 1\}^{\mathbb{N}^*}$, against a Bob with value density $v_B^s$ Alice would get at least $1/12$ regret if she made that partition.

    By our construction of $v_B^s$, this filtering leaves only pieces $B^j_1$ that approximately satisfy the following:
    \begin{itemize}
        \item[(i)] For all $i \in \mathbb{N}^*$, either $L_i \subseteq B^j_1$ and $R_i \cap B^j_1 = \emptyset$ or vice versa (i.e., $L_i \cap B^j_1 = \emptyset$ and $R_i \subseteq B^j_1$).
        \item[(ii)] The measure $\mu(B^j_1 \cap [1/2, 1])$ is ``small".
    \end{itemize}
    Of course, approximately satisfying these conditions is less useful than actually satisfying them. We get closer to that goal by transforming each $B^j_1$ in two stages, each of which preserves Bob's preference for $B^j_1$ over $\overline{B^j_1}$ and approximately preserves $V_A(B^j_1)$. These transformations are illustrated in Figure \ref{fig:B-set-transform-unlabeled}.

    \begin{figure}
        \centering
        \includegraphics[width=0.9\textwidth]{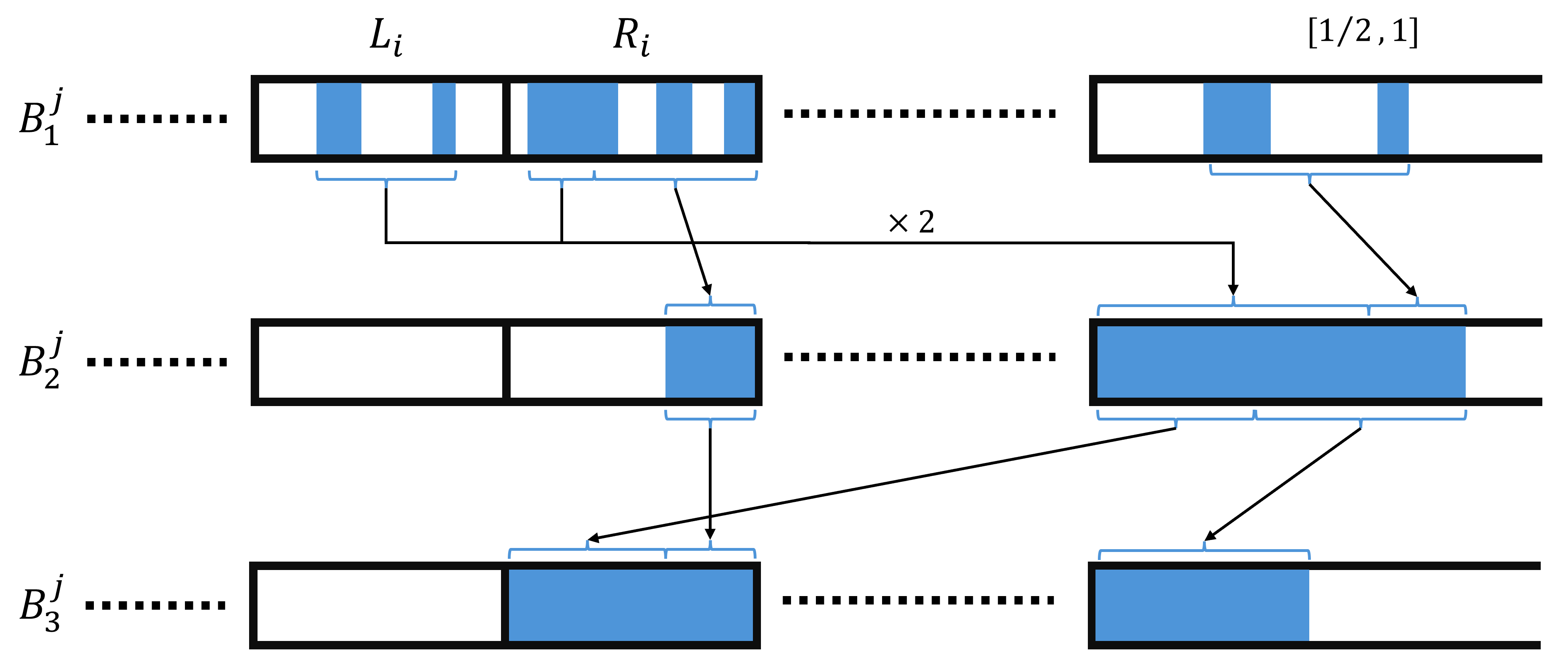}
        \caption{An illustration of the transformations in Theorem \ref{thm:measurable-cut-myopic-lower-bound}, focusing on the mass moved in and out of a single $(L_i, R_i)$ pair. Blue indicates the elements of each of the three sets. The starting $B^j_1$ is an arbitrary piece Alice's strategy could make. From $B^j_1$ to $B^j_2$, an equal amount of mass is removed from $L_i$ and $R_i$ until one of them is empty, in this case $L_i$. Twice this much mass is added to $[1/2, 1]$. From $B^j_2$ to $B^j_3$, mass is moved back into the partially filled $R_i$ to complete it. The resulting $B^j_3$ then contains exactly one set from each $(L_i, R_i)$ pair while still having similar value to $B^j_1$ for both players.} \label{fig:B-set-transform-unlabeled}
    \end{figure}
    
    We first transform $B^j_1$ into $B^j_2$ by transferring mass out of each $(L_i, R_i)$ pair into $[1/2, 1]$. We also double this transferred mass in the process to ensure Bob still prefers $B^j_2$ to $\overline{B^j_2}$. We avoid overfilling $[1/2, 1]$ by condition (ii). We avoid affecting $V_A(B^j_2)$ too much because we only transfer out just enough mass to empty out whichever of $L_i$ or $R_i$ contained less of $B^j_1$, which by condition (i) is small.

    The resulting piece $B^j_2$ now exactly satisfies the following properties:
    \begin{itemize}
        \item For all $i \in \mathbb{N}^*$, either $L_i \cap B^j_2 = \emptyset$ or $R_i \cap B^j_2 = \emptyset$ (or possibly both).
        \item The mass $B^j_2 \cap [1/2, 1]$ is large enough to complete any partially-filled $L_i$ or $R_i$.
    \end{itemize}
    This  naturally suggests the transformation to $B^j_3$: transfer mass from $[1/2, 1]$ into any partially-filled $L_i$ and $R_i$. If for some $i \in \mathbb{N}^*$ both $L_i$ and $R_i$ are empty, we arbitrarily choose one to fill. This is a pure transfer of mass, so $V_A(B^j_3)$ is unaffected and $V_B^s(B^j_3)$ can only increase from $V_B^s(B^j_2)$.

    We now have pieces $B^j_3$ whose mass in $(0, 1/2]$ is divided among the $L_i$ and $R_i$ in a binary way: for each $i \in \mathbb{N}^*$, it contains all of either $L_i$ or $R_i$ and none of the other. We can therefore almost completely describe them by bit vectors $a^j \in \{0, 1\}^{\mathbb{N}^*}$, where $a^j_i=0$ if $L_i \subset B^j_3$ and $a^j_i=1$ otherwise. For arbitrary $s^j \in \{0, 1\}^{\mathbb{N}^*}$ such that the partition $(B^j_1, \overline{B^j_1})$ is $1/12$-Stackelberg for Alice against a Bob with value density $v_B^{s^j}$, we eventually have:
    \begin{align}
        u_A^*(\infty) - V_A(\overline{B^j_3}) &\geq \sum_{i=1}^{\infty} (g(i) - g(i+1)) \left(a^j_i \oplus s^j_i\right)
    \end{align}
    Thus, to ensure Alice has high regret we must find $s^* \in \{0, 1\}^{\mathbb{N}^*}$ that has high Hamming distance from all the $a^j$. We show the existence of such a bit vector by upper-bounding the number of choices for $s^*$ that are too close to each individual $a^j$ and taking a union bound. Our choice of $g$ then ensures Alice receives $\Omega(1/\log^2 T)$ regret per round for a total of $\Omega(T/\log^2 T)$.
\end{proof}
	
\section{The $k$-Cut Game} \label{sec:k_cut_game}

In this section we analyze the $k$-cut game and give proof sketches for  Theorems \ref{thm:myopic_Bob}, \ref{thm:2-cut-non-myopic}, \ref{thm:k-cut-non-myopic-bounds}, and \ref{thm:unknown-alpha-bounds}. The full proofs can be found in Appendix \ref{app:k-cut}.

\subsection{Myopic Bob}

In this section we prove Theorem \ref{thm:myopic_Bob}, which analyzes Alice's optimal strategy when Bob is myopic. Towards this end, we prove the following: an upper bound for $k=2$ cuts (Proposition \ref{prop:two_cuts_myopic_Bob}), an upper bound for $k\geq 3$ cuts (Proposition \ref{prop:k-cut-myopic-upper-bound}), and a lower bound for $k \geq 2$ cuts (Proposition \ref{prop:k-myopic-lower-bound}).

\medskip 
\begin{restatable}[Upper bound for 2 cuts; Myopic Bob]{prop}{twoCutMyopicAlgorithm}  \label{prop:two_cuts_myopic_Bob}
Let $k = 2$ and suppose Bob is myopic. Then Alice has a deterministic strategy that guarantees her Stackelberg regret is $O\bigl(\sqrt{T \log T}\bigr)$.
\end{restatable} 

\begin{proof}[Proof sketch]
    Her strategy is to discretize the following moving-knife procedure:
    \begin{itemize}
        \item Ask Bob for his midpoint $m_B$, i.e. the point such that $V_B([0, m_B])=1/2$.
        \item Slide a knife across the cake, from $0$ to $m_B$. For each position $x$ of the knife, Bob places his own knife at a point $y(x)>x$ such that $V_B([x, y(x)])=1/2$.
        \item Alice cuts at points $(x^*, y(x^*))$ that maximize $\max \{V_A([x^*, y(x^*)]), 1-V_A([x^*, y(x^*)])\}$.
    \end{itemize}
    This procedure works because Bob is indifferent between the pieces of any Stackelberg partition and all such partitions are revealed as the knife moves.

    We now discretize the procedure. Let $\varepsilon > 0$ and let $\eta = \Delta \varepsilon$. Alice's strategy is the following:
    \begin{enumerate}[1.]
        \item For $i=0, 1, \ldots,$ let $x_i = \eta \cdot i$ and use Lemma \ref{lem:binary-search} (which essentially does binary search) to locate $y_i \in [0, 1]$ for which Bob values the pieces $[x_i, y_i]$ and $[0, x_i] \cup [y_i, 1]$ equally up to $\varepsilon$ error. Lemma \ref{lem:binary-search} will succeed if and only if $V_B([0, x_i]) < 1/2$, so once it fails we stop. Let $\widetilde{N}$ be the highest index that succeeded.
        \item Let $i^*$ be the index such that the cut at $x^*$ and $y^*$ contains the piece of greatest value to Alice, i.e.
        \begin{align}
            i^* = \argmax_{i \in \{0, \ldots, \widetilde{N}\}} \max \left(V_A([x_i, y_i]), 1-V_A([x_i, y_i])\right)
        \end{align}
        \item Set $\overline{x}$ and $\overline{y}$ to the values of $x_{i^*}$ and $y_{i^*}$, but shifted by a total of $\varepsilon \Delta/\delta$ in the directions that increase the size of the piece Alice values less.
        \item For all remaining rounds, cut at $\overline{x}$ and $\overline{y}$.
    \end{enumerate}
    In computing $(x_{i^*}, y_{i^*})$, Alice has two sources of error compared to the original procedure: the $\varepsilon$ error on every use of Lemma \ref{lem:binary-search} and the $\eta$ error from discretizing to the $x_i$ and $y_i$. These combined only get her $O(\varepsilon)$ away from a truly optimal cut. The shift from $(x_{i^*}, y_{i^*})$ to $(\overline{x}, \overline{y})$ adds another $O(\varepsilon)$ error, but importantly in the direction that incentivizes Bob to take the piece Alice likes less. As a result, Alice guarantees Bob leaves Alice her preferred piece in the remaining rounds.

    Alice's overall regret then comes from the $O(\log(1/\varepsilon)/\eta)$ rounds it took to find $(\overline{x}, \overline{y})$ and the $O(T\varepsilon)$ loss in the remaining rounds. Setting $\varepsilon = \sqrt{(\log T)/T}$, the loss  becomes $O(\sqrt{T \log T})$.
\end{proof}

We use the following primitive to let Alice learn Bob's valuation when he is myopic.

\begin{restatable}{lemma}{binarySearchLemma} \label{lem:binary-search}
    Suppose Bob is myopic. Let $\varepsilon > 0$ and $x_1, \ldots, x_{k-1} \in [0,1]$ such that $x_1 \leq x_2 \leq \ldots x_{k-1}$. 
    Then within $O(\log 1/\varepsilon)$ rounds, Alice can decide whether there exists $x_k^* \in (x_{k-1}, 1)$ such that Bob's value for either piece in the alternating partition induced by cuts $x_1, \ldots, x_{k-1}, x_k^*$ is exactly $1/2$. 
    Moreover, if $x_k^*$ exists then Alice can find  $\widetilde{x}_k \in (x_{k-1},1)$ such that $|\widetilde{x}_k - x_k^*| \leq \varepsilon$.
\end{restatable}
\begin{proof}[Proof sketch]
    The strategy is binary search. Consider the partitions generated by cut points $x_1, \ldots, x_{k-1}, z$ as $z$ moves continuously from $x_{k-1}$ to $1$. In these partitions, one piece grows while the other shrinks. Therefore, a myopic Bob will prefer one piece for small $z$ and the other for large $z$. The desired $x_k^*$ is the crossover point.
\end{proof}

\begin{restatable}[Upper bound for $k \geq 3$ cuts; Myopic Bob]{prop}{kCutMyopicUpper}
    \label{prop:k-cut-myopic-upper-bound}
    Let $k \geq 3$ and suppose Bob is myopic. Then Alice has a deterministic strategy which guarantees her Stackelberg regret is $O(\sqrt{Tk} \log (Tk))$.
\end{restatable}

\begin{proof}[Proof sketch]
    We will use the general framework of Proposition \ref{prop:two_cuts_myopic_Bob}: first learn approximate information about Bob's value density $v_B$, then use that approximation to compute a near-Stackelberg partition, then adjust the cut points of that partition to ensure Bob prefers the piece Alice likes less. However, the approximate information is much different here. Alice will find a sequence of points $x_{-\ell}, x_{-\ell+1}, \ldots, x_0, x_1, \ldots, x_r$ such that Bob values each interval $[x_i, x_{i+1}]$ similarly, effectively approximating his entire value density function.

    We set parameters $\eta = (Tk)^{-1/2}$ and $\varepsilon = \delta^2 \eta^2/2$. Alice starts by using Lemma \ref{lem:binary-search} to approximate $m_B$ to within $\varepsilon$, then sets $x_0 = m_B-\eta$ and $x_1 = m_B$. She then maps out the cake to the right of $x_1$ by comparing the value of intervals to $[x_0, x_1]$. She can do so by making pieces of the form $[0, x_0] \cup [x_i, z]$, since if we assume $V_B([0, x_1])=V_B([x_1, 1])=1/2$ we get:
    \begin{align}
        & V_B([0, x_0] \cup [x_i, z]) - V_B([x_0, x_i] \cup [z, 1]) \\ =& \bigl(V_B([0, x_1]) - V_B([x_0, x_1]) + V_B([x_i, z])\bigr) - \bigl(V_B([x_1, 1]) - V_B([x_i, z]) + V_B([x_0, x_1])\bigr) \\
        =& 2\bigl(V_B([x_i, z]) - V_B([x_0, x_1])\bigr) \,.
    \end{align}
    So Bob prefers $[0, x_0] \cup [x_i, z]$ over its complement if and only if he prefers $[x_i, z]$ to $[x_0, x_1]$. Alice can construct $x_2, x_3, \ldots$ in sequence, stopping at $x_r$ once the next point would go past $1$. She then uses the same strategy but mirrored left-right, constructing $x_{-1}, x_{-2}, \ldots, x_{-\ell}$ by comparing their intervals to $[x_1, x_2]$. This method ensures each interval is either directly compared to $[x_0, x_1]$ or indirectly compared to it through $[x_1, x_2]$, so the $\varepsilon$ error on each use of Lemma \ref{lem:binary-search} does not compound.

    Alice then computes an optimal partition assuming Bob's value for each interval $[x_i, x_{i+1}]$ is exactly equal and their valuations sum to $1$. The resulting approximation is within $O(k\eta + \varepsilon/\eta) = O(\sqrt{k/T})$ of truly optimal. Shifting the computed cut points by another $O(\sqrt{k/T})$ in Bob's favor ensures he will take the piece Alice likes less.

    Alice's regret is then $O(\log(1/\varepsilon)/\eta)=O(\sqrt{Tk} \log(Tk))$ from the rounds it took to find her cut points plus $O(\sqrt{Tk})$ from the remaining rounds, for a total of $O(\sqrt{Tk} \log(Tk))$.
\end{proof}

\begin{restatable}[Lower bound for $k \geq 2$ cuts; Myopic Bob]{prop}{kMyopicLowerBound}
    \label{prop:k-myopic-lower-bound}
 Let $k \geq 2$. Suppose Bob is myopic. For every Alice valuation and deterministic Alice strategy, there exists a Bob valuation against which her regret is $\Omega\bigl(\frac{\sqrt{T}}{k^{3/2}}\bigr)$.
\end{restatable}
\begin{proof}[Proof sketch]
    Our construction for Bob's value density function starts from an \emph{unspiked} function $\sigma_0^k$. This function is piecewise constant and takes on two different values: $3/2$ in so-called \emph{high-value intervals} and some constant $c<1$ in \emph{low-value intervals}. The high- and low-value intervals alternate from left to right, starting from a high-value interval. There are $k+1$ total intervals. Each type of interval shares a common length and these lengths are chosen so that it takes nearly all the high-value intervals together to be worth $1/2$ to Bob. Against such a Bob, the only way for Alice to get close to her Stackelberg value is to put her $k$ cuts approximately on the $k$ interval boundaries. An example of such a function is plotted in red and purple in Figure \ref{fig:spiked-valuation-plot}.

\begin{figure}
    \centering
    \includegraphics[width=0.7\textwidth]{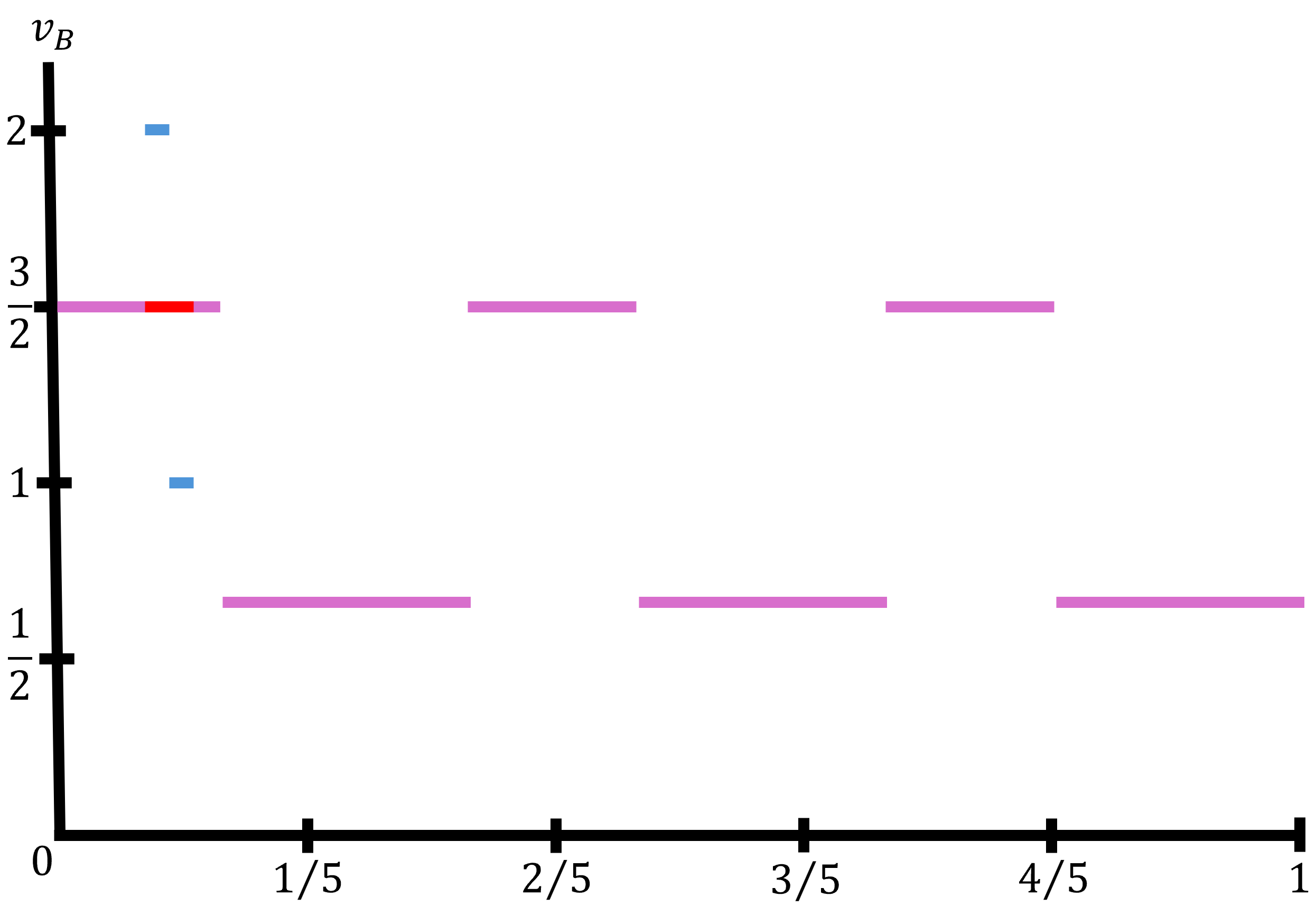}
    \caption{Superimposed plots of the unspiked value density function $\sigma_0^k$ and an example spiked value density $\sigma_{w;z}^k$ for $k=5$. The segments in purple are common to $\sigma_0^k$ and $\sigma_{w;z}^k$. The functions differ only in a small interval centered at $z$ called the spike, where $\sigma_0^k$ maintains its value in red and $\sigma_{w;z}^k$ diverges in blue.}\label{fig:spiked-valuation-plot}
\end{figure}

    We then modify $\sigma_0^k$ by including a ``spike" of nonuniform density in the first high-value interval. Let $\ell^k$ be the length of a high-value interval. The \emph{spiked} value density function $\sigma_{w;z}^k$ is  defined as:
    \begin{align}
        \sigma_{w;z}^k(x) &= \begin{cases}
            2 & \text{if $x/\ell^k \in (z-w, z]$} \\
            1 & \text{if $x/\ell^k \in (z, z+w)$} \\
            \sigma_0^k(x) & \text{otherwise}
        \end{cases}
    \end{align}
    An example is plotted in blue and purple over $\sigma_0^k$ in Figure \ref{fig:spiked-valuation-plot}. Alice's Stackelberg partitions are more restricted against such a Bob: she must put the first cut point at $z\ell^k$ and the other $k-1$ cut points approximately at the last $k-1$ interval boundaries. Her Stackelberg value is, however, greater than that against $\sigma_0^k$ by $\Theta(w/k)$. Therefore, if she wants to get better than $\Theta(Tw/k)$ regret, she must find the spike.

    Finding the spike is difficult mainly because the total value of $((z-w)\ell^k, (z+w)\ell^k)$ is the same between $\sigma_{w;z}^k$ and $\sigma_0^k$, so Alice must place a cut point in the spike in order to detect it. This would still be fine if Alice gets only a little regret with each incorrect guess, but she cannot because of two aspects of the construction:
    \begin{itemize}
        \item The base unspiked $\sigma_0^k$ requires Alice to put her $k$ cut points approximately at the $k$ interval boundaries to get near her Stackelberg value.
        \item The higher-value portion of the spike is on the left side, so if Bobs with value density $\sigma_0^k$ and $\sigma_{w;z}^k$ prefer different pieces of a partition it is because the one with value density $\sigma_{w;z}^k$ prefers the piece containing the left side of the spike.
    \end{itemize}
    Together these conditions show that under any near-optimal partition that would give information about $z$, a Bob with value density $\sigma_0^k$ prefers the piece containing mostly low-value intervals. Alice therefore incurs $\Omega(1/k)$ regret with each such incorrect guess.

    We now set $w=\frac{1}{4\sqrt{Tk}}$ and give Bob a value density function $\sigma_{w;z}^k$ that Alice's strategy never distinguishes from $\sigma_0^k$, or one that she takes the longest to distinguish if she eventually covers all $z$. In each case, Alice's regret is $\Omega(\sqrt{T}/k^{3/2})$:
    \begin{itemize}
        \item If Alice never identifies $z$, she gets $\Omega(w/k)$ regret in all $T$ rounds for a total of $\Omega(\sqrt{T}/k^{3/2})$.
        \item If Alice eventually identifies $z$, she had to cut in every other possible spike first. There are $\Theta(1/w)$ such spikes and she can only cut in $k$ of them with each partition. These $\Theta(1/(wk))$ incorrect guesses each gave her $\Omega(1/k)$ regret for a total of $\Omega(\sqrt{T}/k^{3/2})$.
    \end{itemize}
    This completes the proof sketch.
\end{proof}

\begin{proof}[Proof of  Theorem \ref{thm:myopic_Bob}] The  upper bound for $k=2$ cuts is shown in Proposition \ref{prop:two_cuts_myopic_Bob}. The upper bound for $k\geq 3$ cuts is shown in Proposition \ref{prop:k-cut-myopic-upper-bound}. The  lower bound for $k \geq 2$ cuts is  in Proposition \ref{prop:k-myopic-lower-bound}.
\end{proof}

\subsection{Non-Myopic Bob with Public Learning Rate}
Having investigated the game when Bob is myopic, we now build on it to analyze the setting where Bob uses an algorithm with sub-linear regret.  We start with the scenario where  Bob's regret budget is known.

\paragraph{Two Cuts.} 
For $k=2$ cuts, we completely characterize Alice's optimal regret as $\widetilde{O}(T^{\frac{2+\alpha}{3}})$. To obtain the characterization (Theorem~\ref{thm:2-cut-non-myopic}), we prove separately the upper bound and the lower bounds.

\medskip 

\begin{restatable}[Upper bound for 2 cuts; Non-Myopic Bob with Public Learning Rate]{prop}{twoCutKnownAlphaUpper}
    \label{prop:2-cut-non-myopic-upper-bound-known-alpha}
     Let $k=2$. Suppose Bob's strategy guarantees his regret over time horizon $T$ is at most $c  T^{\alpha}$, for $\alpha \in (-1/2, 1)$ and  $c > 0$. There exists a deterministic strategy for Alice, parameterized by $\alpha$ but independent of $c$, that guarantees her regret is $O(T^{\frac{2+\alpha}{3}} \log^{2/3} T)$.
\end{restatable}

\begin{proof}[Proof sketch]
    Alice's strategy will be nearly identical to the myopic case in Proposition \ref{prop:two_cuts_myopic_Bob}. The main change is that she can no longer use Lemma \ref{lem:binary-search} because Bob is not myopic. She can still use the binary search idea, though: she just has to repeat each partition enough times that Bob cannot mislead her while staying within his regret bound. Specifically, she repeats each partition $4T^{\alpha} \log T / (\delta \varepsilon)$ times and takes a majority vote. As long as $T^{\alpha} \log T$ exceeds Bob's regret bound and the partition is at least $\varepsilon/2$ away from being even to Bob, this process finds Bob's true preference. The first condition holds for sufficiently large $T$ and the second just means Alice cannot get results more accurate than $\varepsilon/2$ from her binary searches.

    There is also a new source of regret for Alice: Bob can use his regret budget even after Alice has chosen her cut points to lower her utility. Alice can mitigate this regret by shifting her cut points by an additional $\Theta(\varepsilon)$ to ensure Bob not only prefers the piece Alice likes less, but he prefers it by $\Omega(\varepsilon)$. He can therefore only lower Alice's utility by $O(T^{\alpha}/\varepsilon)$ in this way.

    Setting $\varepsilon=T^{\frac{\alpha-1}{3}} \log^{2/3} T$ and $\eta = \varepsilon \Delta/\delta$ equalizes Alice's main regret terms for a bound of $O(T^{\frac{2+\alpha}{3}} \log^{2/3} T)$. This completes the proof sketch.
\end{proof}

We show a lower bound for general $k \geq 2$, which implies the required $\Omega\left(T^{\frac{2+\alpha}{3}}\right)$ when $k=2$.

\begin{restatable}[Lower bound for $k\geq 2$ cuts; Non-Myopic Bob with Public Learning Rate]{prop}{knownAlphaLowerBound}
    \label{prop:k-non-myopic-lower-bound}
    Let $k \geq 2$. Suppose Bob's strategy guarantees his regret over time horizon $T$ is at most $c  T^{\alpha}$, for $\alpha \in (-1/2, 1)$ and $c > 0$. 
    For every Alice valuation and strategy, there exists a Bob valuation and an $O(T^{\alpha})$-regret strategy for Bob so that Alice incurs  $\Omega\left(\frac{T^{\frac{2+\alpha}{3}}}{k}\right)$ regret.
\end{restatable}
\begin{proof}[Proof sketch]
    We use the spiked value density functions of Proposition \ref{prop:k-myopic-lower-bound} as in the myopic case, but we now also have the ability to set Bob's strategy. Let $w = T^{\frac{\alpha-1}{3}}$ and let $n = kT^{\frac{2\alpha+1}{3}}$. We say a partition \emph{distinguishes $z$} if Bobs with value density $\sigma_0^k$ and $\sigma_{w;z}^k$ prefer different pieces of it. We choose $z$ and Bob's strategy in two cases, depending on what Alice does against a myopic Bob with value density $\sigma_0^k$:
    \begin{itemize}
        \item If for all $z$ Alice eventually makes $n$ partitions that distinguish it, let Bob's value density be $\sigma_{w;z}^k$ for the last $z$ to receive an $n$th distinguishing partition. His strategy will be to pretend as if he were a myopic Bob with value density $\sigma_0^k$ until that $n$th distinguishing partition, then switch to being myopic with his true value density. Since distinguishing partitions are within $O(w/k)$ of being even for Bob, his regret is upper-bounded by $O(nw/k)=O(T^{\alpha})$. Alice, on the other hand, gets $\Omega(1/k)$ regret for each time she made a potentially distinguishing partition. There are $\Theta(1/w)$ possible disjoint spikes and each partition can only be distinguishing for $k$ of them, so Alice's regret is $\Omega(n/(wk))=\Omega(T^{\frac{2+\alpha}{3}}/k)$.
        \item If there is some $z$ that Alice distinguishes fewer than $n$ times, let Bob's value density be $\sigma_{w;z}^k$ and have him play myopically as if his value density were $\sigma_0^k$. His only regret is $O(w/k)$ for each of the at most $n$ partitions that distinguish $z$, which is $O(T^{\alpha})$ in total. Alice then only sees choices consistent with $\sigma_0^k$, so in all $T$ rounds she gets at least the $\Theta(w/k)$ difference in Stackelberg value between $\sigma_0^k$ and $\sigma_{w;z}^k$ as regret. Her regret is therefore $\Omega(Tw/k) = \Omega(T^{\frac{2+\alpha}{3}}/k)$.
    \end{itemize}
    This completes the proof sketch.
\end{proof}

\begin{proof}[Proof of Theorem~\ref{thm:2-cut-non-myopic}]
    The upper bound is proved in Proposition~\ref{prop:2-cut-non-myopic-upper-bound-known-alpha} and the lower bound in Proposition~\ref{prop:k-non-myopic-lower-bound}. 
\end{proof}

\paragraph{Three or More Cuts.}
We now analyze Theorem~\ref{thm:k-cut-non-myopic-bounds}, which provides bounds for $k \geq 3$ cuts when Bob's regret budget is known. 
The lower bound is already covered by Proposition \ref{prop:k-non-myopic-lower-bound}. The following proposition shows the upper bound.

\begin{restatable}[Upper bound for $k \geq 3$ cuts; Non-Myopic Bob with Public Learning Rate]{prop}{kCutKnownAlphaUpper} \label{prop:k-cut-known-alpha-upper-bound}
   Let $k \geq 3$. Suppose Bob's strategy guarantees his regret over  horizon $T$ is at most $c  T^{\alpha}$, for $\alpha \in (-1, 1)$ and $c > 0$. There exists a deterministic strategy for Alice, parameterized by $\alpha$ but independent of $c$, that guarantees her regret is $O(T^{\frac{3+\alpha}{4}}k^{3/4}(\log T)^{1/2})$.
\end{restatable} 

\begin{proof}[Proof sketch]
    Similarly to Proposition \ref{prop:2-cut-non-myopic-upper-bound-known-alpha}, we modify the myopic strategy from Proposition \ref{prop:k-cut-myopic-upper-bound} by having Alice make each of her binary search partitions $4T^{\alpha} \log T / (\delta \varepsilon)$ times and taking a majority vote. This change allows Alice to still accurately learn Bob's value density.

    We also have the same new regret source, which we handle by shifting the final cut points an extra $\Theta(\varepsilon/\eta + \eta)$ in Bob's favor to ensure he significantly prefers the piece Alice likes less. Setting $\eta = T^{\frac{\alpha-1}{4}}k^{-1/4} \log^{1/2} T$ and $\varepsilon = \delta^2 \eta^2/2$ then equalizes all of Alice's regret terms for a bound of $O(T^{\frac{3+\alpha}{4}}k^{3/4}(\log T)^{1/2})$. This completes the proof sketch.
\end{proof}

\begin{proof}[Proof of Theorem~\ref{thm:k-cut-non-myopic-bounds}]
    The upper bound follows by Proposition~\ref{prop:k-cut-known-alpha-upper-bound}. The lower bound follows by Proposition~\ref{prop:k-non-myopic-lower-bound}.
\end{proof}

\subsection{Non-Myopic Bob with Private Learning Rate}

We now prove Theorem \ref{thm:unknown-alpha-bounds}, which covers the setting where Bob's regret bound is unknown to Alice. 
Towards this end, we analyze separately the upper bound when $k=2$, the upper bound when $k\geq 3$, and  the lower bound for all $k\geq 2$.

\medskip 

\begin{restatable}[Upper bound for $2$ cuts; Non-Myopic Bob with Private Learning Rate]
{prop}{twoCutUnknownAlphaUpper}
    \label{prop:2-cut-non-myopic-upper-bound-unknown-alpha}
    Let $k = 2$. Suppose Bob's strategy guarantees his regret over  horizon $T$ is at most $c T^{\alpha}$, for $\alpha \in (-\infty, 1)$ and $c > 0$. There exists a deterministic strategy for Alice, independent of both $\alpha$ and $c$, that guarantees her regret is $O(T/\log T)$.
\end{restatable}

\begin{proof}[Proof sketch]
    We use the strategy from Proposition \ref{prop:2-cut-non-myopic-upper-bound-known-alpha} (for when $\alpha$ is known) but make two changes:
    \begin{itemize}
        \item Increase the number of times Alice makes each binary search partition to $4T/(\delta \varepsilon \log^4 T)$. No matter what $\alpha$ is, for sufficiently large $T$ this will be enough to force truthful responses from Bob.
        \item Set the parameters $\varepsilon=1/\log T$ and $\eta = \varepsilon \Delta/\delta$. These balance the terms in Alice's regret to $O(T/\log T)$ as desired.
    \end{itemize}
    This completes the proof sketch.
\end{proof}

\begin{restatable}[Upper bound for $k \geq 3$ cuts; Non-Myopic Bob with Private Learning Rate]{prop}{kCutNonMyopicUpperBoundUnknownAlpha}
    \label{prop:k-cut-non-myopic-upper-bound-unknown-alpha}
  Let $k \geq 3$. Suppose Bob's strategy guarantees his regret over  horizon $T$ is at most $c T^{\alpha}$, for $\alpha \in (-\infty, 1)$ and $c > 0$. There exists a deterministic strategy for Alice, independent of both $\alpha$ and $c$, that guarantees her regret is $O(Tk^{3/4}/\log T)$.
\end{restatable}

\begin{proof}[Proof sketch]
    Similarly to Proposition \ref{prop:2-cut-non-myopic-upper-bound-unknown-alpha}, we make two changes to the strategy from Proposition \ref{prop:k-cut-known-alpha-upper-bound} from the known-$\alpha$ setting:
    \begin{itemize}
        \item Increase the number of times Alice makes each binary search partition to $4T/(\delta \varepsilon \log^6 T)$. No matter what $\alpha$ is, for sufficiently large $T$ this will be enough to force truthful responses from Bob.
        \item Set the parameters $\eta = 1/\log T$ and $\varepsilon=\delta^2 \eta^2/2$. These balance the terms in Alice's regret to $O(T/\log T)$ as desired.
    \end{itemize}
    This completes the proof sketch.
\end{proof}

\begin{restatable}[Lower bound for $k \geq 2$ cuts; Non-Myopic Bob with Private Learning Rate]{prop}{unknownAlphaLowerBound}
    \label{prop:unknown-alpha-lower-bound}
    Let $k \geq 2$. Suppose Bob's strategy guarantees his regret over  horizon $T$ is at most $c T^{\alpha}$, for $\alpha \in (-\infty, 1)$ and $c > 0$. Suppose that for some constant $\beta < 1$, Alice uses a strategy $S_A$, independent of both $\alpha$ and $c$, that guarantees her  regret is $O(T^{\beta})$ against such a Bob. Then there exists a Bob valuation and strategy that guarantees Bob's regret is $O(T^{\beta})$  but ensures Alice's regret is $\Omega(T)$ when using $S_A$.
\end{restatable}
\begin{proof}[Proof sketch]
    The Bob we construct will have one value density but pretend to play myopically under another. He will guarantee his regret is $O(T^{\beta})$ by reverting to playing myopically if he reaches that limit, but the condition on $S_A$ ensures that will never happen.

    Specifically, by Lemma \ref{lem:wlog-alice-valuation-warping} we assume Alice's valuation is the uniform $v_A(x)=1 \,\forall x \in [0, 1]$. Let $v_B^1$ and $v_B^2$ be the following value densities:
    \begin{align}
        v_B^1(x) = \begin{cases}
            1/2 & \text{if $x \in [0, 1/2]$} \\
            3/2 & \text{if $x \in (1/2, 1]$}
        \end{cases} \qquad \mbox{ and } \qquad 
        v_B^2(x) = \begin{cases}
            1/4 & \text{if $x \in [0, 1/2]$} \\
            7/4 & \text{if $x \in (1/2, 1]$}
        \end{cases}
    \end{align}
    Value density $v_B^2$ is a more extreme version of $v_B^1$, so any partitions that Alice could use to distinguish the two will give Alice $\Omega(1)$ regret if Bob takes the piece preferred by $v_B^1$. In particular, against a truly myopic Bob with value density $v_B^1$, the condition on $S_A$ forces Alice to only make $O(T^{\beta})$ partitions that could distinguish the Bobs. Therefore, if Bob's true value density is $v_B^2$, he can successfully pretend to have value density $v_B^1$ against $S_A$ without ever exceeding $O(T^{\beta})$ regret himself. But Alice's Stackelberg value is $1/21$ higher against a Bob with value density $v_B^2$ than $v_B^1$, so across $T$ rounds her regret is $\Omega(T)$.
\end{proof}

\begin{proof}[Proof of Theorem \ref{thm:unknown-alpha-bounds}]
The upper bound for $k=2$ follows by Proposition~\ref{prop:2-cut-non-myopic-upper-bound-unknown-alpha}, the upper bound for $k \geq 3$ follows by Proposition~\ref{prop:k-cut-non-myopic-upper-bound-unknown-alpha}, and the lower bound for $k \geq 3$ follows by Proposition~\ref{prop:unknown-alpha-lower-bound}. 
\end{proof}

\section{The RW Query Model} \label{sec:RW_corollary}

 The analysis for the repeated game allows us to obtain the following bounds for finding Stackelberg allocations in the RW query model. For constant $k$, the randomized query complexity is $\Theta(1/\varepsilon)$. We provide a proof sketch here, while the full proof can be found in Appendix \ref{app:rw-query-model}.

\RWQueryModel*

\begin{proof}[Proof sketch]
    For the upper bound, let $n = \lceil k/\varepsilon \rceil$. We can use $n$ cut queries to partition the cake into $n$ intervals of value $1/n$ to Alice, then use $n$ eval queries to get Bob's value of each interval. We can then enumerate every $k$-cut with cut points on the boundary of these intervals and compute how good they are for Alice. The best of these $k$-cuts is $(k/n)$-Stackelberg and therefore $\varepsilon$-Stackelberg.

    For the lower bound, rather than consider a worst-case input for an arbitrary randomized algorithm we apply Yao's minimax principle \cite{yao_minimax} to consider the best deterministic algorithm for a hard input distribution. We will use the spiked value density functions from Proposition \ref{prop:k-myopic-lower-bound}. Specifically, let $w=14\varepsilon$. Let $\mathcal{D}$ be the distribution of value density functions for Alice and Bob where:
    \begin{itemize}
        \item Alice has the uniform value density $v_A(x)=1 \, \forall x \in [0, 1]$, and
        \item Bob has the value density $\sigma_{w;z}^2$ for uniformly random $z$.
    \end{itemize}
    The only unknown parameter of an input from $\mathcal{D}$ is Bob's spike $z$. Any correct algorithm must find the spike, since the Stackelberg value when Bob has value density $\sigma_0^2$ is $\varepsilon$ lower than when he has value density $\sigma_{w;z}^2$. However, the only way for a cut or eval query to gain information about $z$ is to query a point in the spike. There are $\Theta(1/\varepsilon)$ disjoint spikes possible, so in order to succeed with high probability against inputs from $\mathcal{D}$ any deterministic algorithm must make $\Omega(1/\varepsilon)$ queries with constant probability. By Yao's minimax principle, the same $\Omega(1/\varepsilon)$ bound holds for the randomized query complexity.
\end{proof}

\section{Discussion}

In the measurable-cut game, we leave open whether the true optimal value of Alice's regret is $\Theta(T)$ or $o(T)$ against a myopic Bob. If it is $o(T)$, can a Bob with sub-linear regret force Alice to have $\Omega(T)$ regret?

The exponents in the upper and lower bounds of Theorem \ref{thm:k-cut-non-myopic-bounds} do not match. Closing this gap remains an open question. More broadly, our upper bounds for $k=2$ and $k \geq 3$ use different techniques. Is this split necessary?

\bibliographystyle{ACM-Reference-Format}
\bibliography{ref}

\appendix

\section{Preliminaries} \label{app:prelim}

In this section we include the proofs omitted from Section~\ref{sec:prelim}.

\aliceValuationWarping*

\begin{proof}
    The idea is to stretch different areas of the cake so that $v_A^1$ turns into $v_A^2$; everything else will simply be carried along by that stretching. Specifically, for $i \in \{1, 2\}$ define the function:
    \begin{align}
        V_i(x) &= \int_{0}^x v_A^i(y) \,\mathrm{d}y \,.
    \end{align}
    This is well-defined since both $v_A^1$ and $v_A^2$ are integrable. Then define  $f: [0, 1] \to [0, 1]$ as:
    \begin{align}
        f(x) &= V_2^{-1}(V_1(x)) \,. 
    \end{align}
    The inverse function $V_2^{-1}$ is well-defined because $v_A^2$ is positive and bounded, so its integral is strictly increasing and is Lipschitz.
    
    More strongly, $V_1$ and $V_2$ are integrals of functions with values in $[\delta, \Delta]$, so they are bi-Lipschitz. Therefore, their inverses $V_1^{-1}$ and $V_2^{-1}$ are also bi-Lipshitz, so their composition $f$ is bi-Lipschitz. Because $f$ is bi-Lipschitz, it can be used as a transformation in integration by substitution. This will be used to show that valuations are preserved across the map, if cut points are adjusted accordingly.

    Specifically, let $C$ be one piece of an arbitrary measurable cut. Then:
    \begin{align}
        \int_{C} v_A^1(x) \,\mathrm{d}x &= \int_{f(C)} v_A^1(f^{-1}(x)) (f^{-1})'(x) \,\mathrm{d}x = \int_{f(C)} v_A^1(V_1^{-1}(V_2(x))) \cdot \frac{1}{V_1'(V_1^{-1}(V_2(x)))} \cdot V_2'(x) \,\mathrm{d}x \\
        &= \int_{f(C)} \frac{v_A^1(V_1^{-1}(V_2(x)))}{v_A^1(V_1^{-1}(V_2(x)))} \cdot v_A^2(x) \,\mathrm{d}x 
        = \int_{f(C)} v_A^2(x) \,\mathrm{d}x \,.
    \end{align}
    Therefore, the value of $C$ to an Alice with valuation $v_A^1$ is the same as the value of $f(C)$ to an Alice with valuation $v_A^2$. Accordingly, Alice's strategy $S_A^2$ will be to run $S_A^1$ but apply $f$ to all the cuts that strategy chooses.

    We now need a value density $v_B^2$ and strategy $S_B^2$ for Bob that is equivalent to $v_B^1$ and $S_B^1$ after applying $f$. Bob's strategy $S_B^2$ will be to apply $f^{-1}$ to the cuts he sees and then follow $S_B^1$. Bob's value density $v_B^2$ will be:
    \begin{align}
        v_B^2(x) &= \frac{v_B^1(f^{-1}(x))}{v_A^1(f^{-1}(x))} \cdot v_A^2(x) \; \; \forall x \in [0,1] \,. 
    \end{align}
    The  value density $v_B^2$ is bounded to $[\delta^2/\Delta, \Delta^2/\delta]$ because each of $v_B^1$, $v_A^1$, and $v_A^2$ are bounded to $[\delta, \Delta]$. It is equivalent to $v_B^1$ after applying $f$ because, by integration by substitution, for any piece of cake $C$:
    \begin{align}
        \int_C v_B^1(x) \,\mathrm{d}x &= \int_{f(C)} v_B^1(f^{-1}(x)) (f^{-1})'(x) \,\mathrm{d}x 
        = \int_{f(C)} v_B^1(f^{-1}(x)) \cdot \frac{1}{V_1'(V_1^{-1}(V_2(x)))} \cdot V_2'(x) \,\mathrm{d}x \notag \\
        & = \int_{f(C)} \frac{v_B^1(f^{-1}(x))}{v_A^1(f^{-1}(x))} \cdot v_A^2(x) \,\mathrm{d}x = \int_{f(C)} v_B^2(x) \,\mathrm{d}x \,. 
    \end{align}
    This completes the proof of equivalence.
\end{proof}

\section{The Measurable-Cut Game} \label{app:measurable}

In this section we give the full proof of Theorem \ref{thm:measurable-cut-myopic-lower-bound}, which we restate here.

\measurableMyopicLower*

\begin{proof}[Proof of Theorem~\ref{thm:measurable-cut-myopic-lower-bound}]
Briefly, the Bob valuation we construct will partition $[0, 1/2)$ into infinitely many intervals, concentrating his value within each interval in either the left or right half. By choosing these halves adversarially against the cuts Alice's strategy makes, we can ensure none of them get her too close to her Stackelberg value.

  The partition is defined using a function $f:[0,\infty) \to [0,1/2]$. For the moment, we keep $f$ abstract and assume only that it is twice-differentiable and satisfies:  
        (i) $f(0) = \frac{1}{2}$,
       (ii)  $f'(x) < 0$ and $f''(x) > 0$ for all $x \geq 0$, and 
       (iii) $\lim_{x \to \infty} f(x) = 0$.
    
    Let $v_A$ and $S_A$ be an arbitrary value density and deterministic strategy, respectively, for Alice.
By Lemma \ref{lem:wlog-alice-valuation-warping}, we assume without loss of generality that Alice's valuation is  uniform: $v_A(x) = 1$ for all $x \in [0,1]$.
    
   For each infinite bit vector $s \in \{0, 1\}^{\mathbb{N}^*}$, we construct a Bob valuation parameterized by $s$. 
    Towards this end, let $g: \mathbb{N}^* \to [0, 1/2]$ be defined as $g(x) = f(\log x)$. Then 
    \begin{align} 
    g(1) = f(0) = 1/2 \qquad \mbox{ and } \qquad \lim_{x \to \infty} g(x) = \lim_{x \to \infty} f(\log x) = 0 \,. \label{eq:g-function-start-and-limit}
    \end{align}
    For each infinite bit vector $s \in \{0,1\}^{\mathbb{N}^*}$, we  define a  value density $v_B^s: [0,1] \to \mathbb{R}$ as the following piecewise function:
    \begin{itemize}
        \item Let $v_B^s(x) = 2/3$ for all $x \geq 1/2$ and $x=0$.
        \item For each $i \in \mathbb{N}^*$, let 
        \begin{align}  \label{eq:def_L_i_R_i}
        L_i = \Bigl[g(i+1), \frac{g(i+1)+g(i)}{2}\Bigr) \qquad and \qquad R_i = \Bigl[\frac{g(i+1)+g(i)}{2}, g(i)\Bigr) \,.
        \end{align} 
        Define 
        \begin{align}  
        v_B^s(x) = 2-\frac{4}{3}s_i \; \forall x \in L_i    
        \; \; \mbox{ and } \; \; v_B^s(x) = \frac{2}{3} + \frac{4}{3}s_i \; \forall x \in R_i \,.
        \end{align}
    \end{itemize}
    By \eqref{eq:g-function-start-and-limit}, we have:
    \begin{align}
        \bigcup_{i=1}^{\infty} (L_i \cup R_i) = \bigl(\lim_{i \to \infty} g(i), g(1)\bigr) = (0, 1/2) \,. \label{eq:union_Li_R_i_overall_partition_zero_one_half}
    \end{align}
    Therefore $v_B^s$ is well-defined on $[0, 1]$.
    
    The intuition is that  half of each interval $[g(i+1), g(i))$ (as chosen by $s_i$) has density $2$ and the rest of the cake has density $2/3$. 
    Then Alice's Stackelberg value is $3/4$ regardless of $s$: the intervals on which Bob's value density is $2$ make up $1/4$ of the cake in length, so he would be happy receiving those while Alice gets the remaining $3/4$ of the cake.

Our goal is to show there exists an infinite bit vector $s^*$ so that, against a Bob with value density $v_B^{s^*}$, no partition Alice could make using $S_A$ is $\varepsilon$-Stackelberg for $\varepsilon = 0.01 \cdot |f'(\log(10T))|$.

\paragraph{Enumerating the partitions that Alice could make under $S_A$. Constructing the sequence $(B_1^j)_{j\geq 1}$.}
For each horizon \(h\), let \(P^{(1)}_h,\ldots,P^{(m)}_h\) be the partitions Alice could make using $S_A$ in an $h$-round game. We have  \(m\le 2^h\) since by the end of each round $t$ she can observe at most $2^{t}$ different histories.
For each \(j\in[m]\), let \(\textsc{Less}^{(j)}_h\) denote the piece of \(P^{(j)}_h\) that Alice values less (breaking ties arbitrarily).
 
To collate these pieces across all $h=1,2,\ldots$, we relabel them as
\begin{align} \label{eq:B_1^j_sequence_def}
B^{2^h}_1,\,B^{2^h+1}_1,\,\ldots,\,B^{2^h+m-1}_1,
\end{align}
so indices in \(\{2^h,\ldots,2\cdot 2^h-1\}\) that do not occur correspond to no realized partitions.
We further restrict attention to those indices \(j\) for which the allocation \((\overline{B_1^j}, B^j_1)\) is \(1/12\)-Stackelberg for Alice with uniform density when she faces some Bob valuation \(v_B^s\). This is without loss of generality: if \((\overline{B_1^j}, B^j_1)\) is not even \(1/12\)-Stackelberg for Alice with any \(v_B^s\), then it cannot be \(\varepsilon\)-Stackelberg, so discarding it cannot make Alice's task easier. Formally, define 
\begin{align}
    \mathcal{J} & = \left\{j \in \mathbb{N} \mid \text{$B^j_1$ is defined in \eqref{eq:B_1^j_sequence_def} and } \right. \notag \\
    & \qquad  \qquad \; \; \; \left. \text{$\exists s \in \{0, 1\}^{\mathbb{N}^*}$ s.t. $(\overline{B^j_1}, B^j_1)$ is $\frac{1}{12}$-Stackelberg for Alice when Bob has  density $v_B^s$}\right\}
    \label{eq:cal-J-set-definition}
\end{align}

\paragraph{Converting the sequence $(B_1^j)_{j \geq 1}$ into the sequence $(B_2^j)_{j \geq 1}$.} For each index $j \in \mathcal{J}$, consider the corresponding piece \(B^j_1\).
We transform \(B^j_1\) into a more structured piece without improving Alice's position by much; the structured form will be easier to analyze. Our overall process is illustrated in Figure \ref{fig:B-set-transform}.
In the first stage we define a modified piece \(B^j_2\) as follows:
    \begin{itemize}
        \item For each $i \in \mathbb{N}^*$, let $m^j_i = \min \left\{\mu\left(B^j_1 \cap L_i\right), \mu\left(B^j_1 \cap R_i\right)\right\}$, where $\mu$ denotes the Lebesgue measure.
        \item We create $B_2^j$ from $B_1^j$ as follows. For each $i \in \mathbb{N}^*$, remove $m^j_i$ mass from each of $L_i$ and $R_i$. Make up this loss with $4m^j_i$ additional mass placed in $[1/2, 1]$. We also take this opportunity to arrange all the mass in convenient intervals. Formally, 
 define $z^j_1  := \mu(B^j_1 \cap [1/2, 1])$. For each $i \in \mathbb{N}^*$, define 
\begin{align} \label{eq:cases_E_F_G_j_defs}
\begin{cases}
\; E^j_i := \bigl[g(i+1), g(i+1) + \mu(B^j_1 \cap L_i) - m_i^j \bigr); \\ 
\; F^j_i  := \bigl[g(i) - \mu(B^j_1 \cap R_i) + m_i^j, g(i) \bigr); \\ 
\; G^j := \bigl[\frac{1}{2}, \frac{1}{2} + z^j_1 + 4\sum_{i=1}^{\infty} m_i^j\bigr] \,.
\end{cases}
\end{align} 
Now let  
\begin{align}
    B^j_2 := G^j \cup \left( \bigcup_{i=1}^{\infty} (E^j_i \cup F^j_i) \right) \,.
\end{align}
\end{itemize}

\begin{figure}
    \centering
    \includegraphics[width=0.95\textwidth]{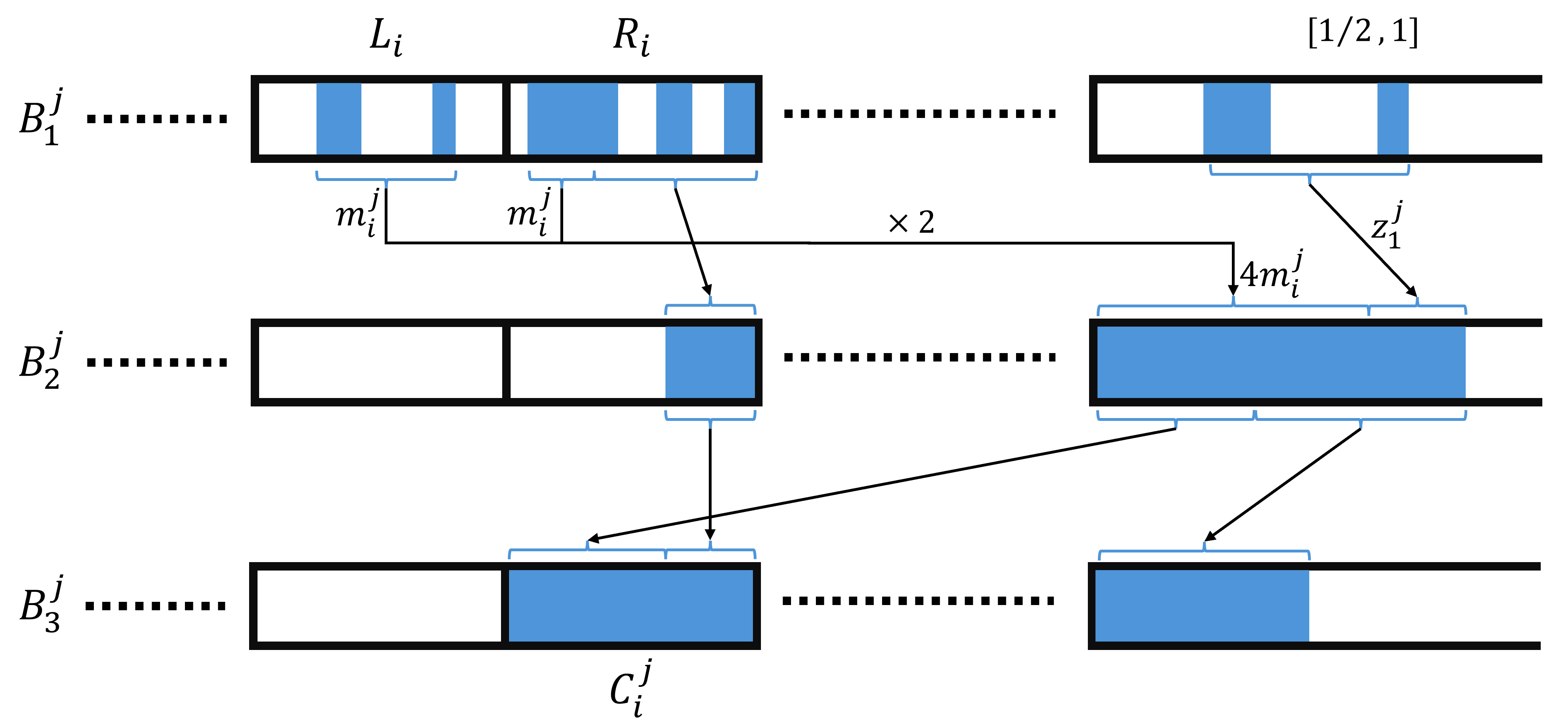}
    \caption{An illustration of the process to convert $B_1^j$ into $B_3^j$, focusing on a particular $L_i$ and $R_i$. The blue areas are the elements of their respective sets. Arrows indicate how much mass is moved and where. Mass from the other intervals would also be moved into $[1/2, 1]$ in $B^j_2$, but it is omitted for readability.}\label{fig:B-set-transform}
\end{figure}

\paragraph{Analyzing the properties of $(B_2^j)_{j \geq 1}$.}    Next we show that $B^j_2$ fits in $[0, 1]$ and analyze how the transformation affects $V_A(B^j_2)$. Since $j \in \mathcal{J}$, there exists an  infinite bit vector  $s^j \in \{0, 1\}^{\mathbb{N}^*}$ with the property  that $(\overline{B^j_1}, B^j_1)$ is $1/12$-Stackelberg for Alice with uniform valuation when Bob has value density $v_B^{s^j}$.

    Consider an arbitrary such $s^j$. By definition of $B^j_1$,  Alice prefers $\overline{B_1^j}$ to $B^j_1$, so we must have $V_B^{s^j}(B^j_1) \geq 1/2$.
    We use the notation $W_{i}^{j} := L_i$ if $s_i^j=1$ and $W_{i}^{j} := R_i$ otherwise. Then  
    \begin{align}
        \frac{1}{2} &\leq V_B^{s^j}(B^j_1) \label{eq:V-b-s-j-greater-than-one-half} \\
        &= 2\mu(B^j_1) - \frac{4}{3} \mu(B^j_1 \cap [1/2, 1]) - \frac{4}{3} \sum_{i=1}^{\infty} \mu(B^j_1 \cap W^j_i) \explain{By definition of $v_{B}^{s^j}.$}\\
        &\leq 2\mu(B^j_1) - \frac{4}{3} z^j_1 - \sum_{i=1}^{\infty} m^j_i \,. \explain{By definition of $z_1^j$ and $m_i^j$.}
    \end{align}
    Solving for $\mu(B^j_1)$ gives
    \begin{align}
        \mu(B^j_1) &\geq \frac{1}{4} + \frac{2}{3}\left(z^j_1 + \sum_{i=1}^{\infty} m_i^j\right) \,. \label{eq:B-j-1-measure-greater-than-1-4}
    \end{align}
    By definition of $\mu$ and $V_A$, we have that (i) $\mu(B^j_1) = V_A(B^j_1)$ and (ii)  Alice's Stackelberg value is $u_A^*(\infty)=3/4$. Then we can lower-bound Alice's regret in terms of $z^j_1$ and the $m^j_i$ as follows:
    \begin{align}
        u_A^*(\infty) - V_A\left(\overline{B^j_1}\right) = \frac{3}{4} - 1 + \mu(B^j_1) \geq \frac{2}{3}\left(z^j_1 + \sum_{i=1}^{\infty} m_i^j\right) \,. \label{eq:alice-stackelberg-from-B-1-j}
    \end{align}
    Since $(\overline{B^j_1}, B^j_1)$ is $1/12$-Stackelberg  for Alice when Bob has value density  $v_B^{s^j}$, we have that $u_A^*(\infty) - V_A\left(\overline{B^j_1}\right) \leq 1/12$, so 
    \begin{align}
        z^j_1 + \sum_{i=1}^{\infty} m_i^j \leq \frac{3}{2}\left(u_A^*(\infty) - V_A\left(\overline{B^j_1}\right)\right) \leq \frac{1}{8} \,. 
    \end{align}
    
    In particular, we have $\mu(G^j) = z^j_1 + 4\sum_{i=1}^{\infty} m_i^j \leq 4 \cdot 1/8 = 1/2$, which implies $B_2^j \subseteq [0, 1]$. Furthermore, Alice's valuation for $B_2^j$ satisfies:
    \begin{align}
        V_A(B_2^j) &= \mu(G^j) + \sum_{i=1}^{\infty} \left(\mu(E^j_i) + \mu(F^j_i)\right) \\
        &= z^j_1 + 4\sum_{i=1}^{\infty}m^j_i + \sum_{i=1}^{\infty} \left(\mu(B^j_1 \cap L_i) - m^j_i + \mu(B^j_1 \cap R_i) - m^j_i\right)\\
        &= \left[z^j_1 + \sum_{i=1}^{\infty} \left(\mu(B^j_1 \cap L_i) + \mu(B^j_1 \cap R_i)\right)\right] + 2\sum_{i=1}^{\infty} m^j_i \,. \label{eq:all_parts_of_B1_j}
    \end{align}
Since the expression in square brackets in \eqref{eq:all_parts_of_B1_j} represents $\mu(B^j_1)=V_A(B^j_1)$, by using \eqref{eq:alice-stackelberg-from-B-1-j} we obtain:
    \begin{align} 
        V_A(B^j_2) &= V_A(B^j_1) + 2\sum_{i=1}^{\infty} m^j_i \label{eq:V-A-B2j-greater-than-V-A-B1j} 
        \leq V_A(B^j_1) + 3\left(u_A^*(\infty) - V_A\left(\overline{B^j_1}\right)\right) \,. 
    \end{align}
    So if $S_A$ used the partition $(B^j_2, \overline{B^j_2})$ instead of $(B^j_1, \overline{B^j_1})$, it would at most quadruple Alice's Stackelberg regret. We can therefore analyze the game when Alice's strategy $S_A$ produces the partition $(B^j_2, \overline{B^j_2})$  instead of $(B^j_1, \overline{B^j_1})$ and divide our regret bound by $4$ at the end.

    This transformation also preserves Bob's valuation of the piece $B_2^j$ compared to $B_1^j$. Specifically, since for each $i \in \mathbb{N}^*$ we remove $2m_i^j$ mass of average density $4/3$ from $L_i \cup R_i$ and add $4m_i^j$ mass of density $2/3$ to $[1/2, 1]$, we have:
    \begin{align}
        V_B^s(B^j_1) = V_B^s(B^j_2) \qquad \text{for all } s \in \{0, 1\}^{\mathbb{N}^*} \,. \label{eq:B-j-1-value-equals-B-j-2-value}
    \end{align}

    \paragraph{Converting the sequence $(B_2^j)_{j \geq 1}$ into the sequence $(B_3^j)_{j \geq 1}$.} We further transform $B^j_2$ into $B^j_3$ so that, for every $i \in \mathbb{N}^*$, either $B^j_3 \cap L_i = L_i$ and $B^j_3 \cap R_i = \emptyset$ or vice versa. This is conceptually simpler to achieve than the previous change: to get $B^j_3$, just transfer mass from $[1/2, 1]$ into the appropriate intervals. Formally,  
    \begin{itemize} 
    \item Let $z^j_2 := \mu(B^j_2 \cap [1/2, 1])$. 
    \item For each $i \in \mathbb{N}^*$, define $C^j_i := L_i$ if $\mu(B^j_2 \cap L_i) > 0$ and $C^j_i := R_i$ otherwise.
    \item Let $\widetilde{G}^j := \left[\frac{1}{2}, \frac{1}{2} + z^j_2 - \sum_{i=1}^{\infty} \mu\left(C^j_i \setminus (C^j_i \cap B^j_2)\right)\right]$.
    \end{itemize}
    Then define:
    \begin{align}
        B^j_3 := \bigl(\bigcup_{i=1}^{\infty} C^j_i\bigr) \cup \widetilde{G}^j \,.
    \end{align}
    We show the interval $\widetilde{G}^j$ is well-defined. By \eqref{eq:union_Li_R_i_overall_partition_zero_one_half} we have $\sum_{i=1}^{\infty} \mu(L_i \cup R_i) = \mu((0, 1/2)) = 1/2$. Moreover, by the definition in \eqref{eq:def_L_i_R_i}, we have $\mu(L_i) = \mu(R_i)$, so $\sum_{i=1}^{\infty} \mu(C^j_i) = 1/4$. Then:
    \begin{align}
        \mu(\widetilde{G}^j) & = z^j_2 - \sum_{i=1}^{\infty} \mu\left(C^j_i \setminus (C^j_i \cap B^j_2)\right) = \left[z^j_2 + \sum_{i=1}^{\infty} \mu(C^j_i \cap B^j_2)\right] - \sum_{i=1}^{\infty} \mu(C^j_i) = \mu(B^j_2) - \frac{1}{4}  \notag \\
        & = V_A(B_2^j) - \frac{1}{4} \explain{By definition of $V_A$} \notag \\
        & \geq V_A(B_1^j) - \frac{1}{4} \explain{By \eqref{eq:V-A-B2j-greater-than-V-A-B1j}} \\
        & \geq 0 \,. \explain{Since $V_A(B_1^j) = \mu(B^j_1) \geq 1/4$ by \eqref{eq:B-j-1-measure-greater-than-1-4}}
    \end{align}
    So $\widetilde{G}^j$ is well-defined. Furthermore, we get $\mu(B^j_3) = \mu(B^j_2)$, so $V_A(B^j_3) = V_A(B^j_2)$.
    
    We construct $B^j_3$ from $B^j_2$ by moving mass from $[1/2, 1]$ to $[0, 1/2]$. Since, for all $s \in \{0, 1\}^{\mathbb{N}^*}$, we have $v_B^{s}(x) \geq 2/3$ for all $x \in [0, 1/2]$ and $v_B^s(x)=2/3$ for all $x \in [1/2, 1]$, we get:
    \begin{align}
        V^s_B(B^j_3) \geq V^s_B(B^j_2) \qquad \text{for all } s \in \{0, 1\}^{\mathbb{N}^*} \,. \label{eq:B-j-3-value-greater-than-B-j-2-value}
    \end{align}
    Combining \eqref{eq:B-j-1-value-equals-B-j-2-value} with \eqref{eq:B-j-3-value-greater-than-B-j-2-value} yields:
    \begin{align}
        V^s_B(B^j_3) \geq V^s_B(B^j_1) \qquad \text{for all } s \in \{0, 1\}^{\mathbb{N}^*} \,. \label{eq:B-j-3-value-greater-than-B-j-1-value}
    \end{align}
    \paragraph{Estimating the value of $B_3^{j}$ for Alice and Bob.}
 We now have pieces in a convenient form. Define the following measures of $B^j_3$:
    \begin{itemize}
        \item For each $i \in \mathbb{N}^*$, let $\ell_i^j := \mu(B^j_3 \cap L_i)$ 
            and  $r_i^j := \mu(B^j_3 \cap R_i)$.
        \item Let $z^j_3 := \mu(B^j_3 \cap [1/2, 1])$.
    \end{itemize}
    By \eqref{eq:B-j-3-value-greater-than-B-j-1-value} and \eqref{eq:V-b-s-j-greater-than-one-half}, we have 
    \begin{align}  \label{eq:Bobs_value_for_B_3_j_according_to_s_j_is_at_least_one_half}
    V_B^{s^j}(B^j_3) \geq V_B^{s^j}(B^j_1) \geq 1/2 \,.
    \end{align}

Directly computing $V_B^{s^j}(B^j_3)$, we decompose the integral over the regions where Bob's density is constant. Recall that Bob's density $v_B^{s^j}$ is $2/3$ on $[1/2, 1]$, $2 - \frac{4}{3} s^j_i$ on $L_i$, and $\frac{2}{3} + \frac{4}{3} s^j_i$ on $R_i$. Substituting the measures $z_3^j$, $\ell_i^j$, and $r_i^j$, we have:
\begin{align}
    V_B^{s^j}(B^j_3) &= \frac{2}{3}z^j_3 + \sum_{i=1}^{\infty} \left[ \ell_i^j \left(2 - \frac{4}{3} s^j_i\right) + r_i^j \left(\frac{2}{3} + \frac{4}{3} s^j_i\right) \right]  \,. \label{eq:vb-b3-explicit-sum}
\end{align}
We regroup the terms in the sum of \eqref{eq:vb-b3-explicit-sum} to separate the parts that depend on $s_i^j$ from those that do not and then factor out $(\ell_i^j - r_i^j)$ to isolate the dependence on $s_i^j$, obtaining:
\begin{align}
   V_B^{s^j}(B^j_3) &=  \frac{2}{3}z^j_3 + \sum_{i=1}^{\infty} \left[ \frac{4}{3}(\ell_i^j + r_i^j) + \frac{2}{3}(\ell_i^j - r_i^j) - \frac{4}{3} s^j_i (\ell_i^j - r_i^j) \right] \notag \\
   & =  \frac{2}{3}z^j_3 + \frac{4}{3}\sum_{i=1}^{\infty} (\ell_i^j + r_i^j) + \frac{2}{3}\sum_{i=1}^{\infty} (\ell_i^j - r_i^j)(1 - 2s^j_i) \,. \label{eq:Bob_value_for_B_3_j_according_to_s_j}  
\end{align}
We relate this to Alice's valuation for the complement piece $\overline{B_3^j}$, which is:
\begin{align} \label{eq:Alice_value_for_complement_of_B_3_j_according_to_s_j}
    V_A(\overline{B_3^j}) &= 1 - z^j_3 - \sum_{i=1}^{\infty} (\ell_i^j + r_i^j)\,.
\end{align}
Re-arranging Bob's value in \eqref{eq:Bob_value_for_B_3_j_according_to_s_j} to factor out Alice's value for the complement in \eqref{eq:Alice_value_for_complement_of_B_3_j_according_to_s_j}, we obtain
\begin{align}
    V_B^{s^j}(B^j_3) 
    &= \frac{2}{3} - \frac{2}{3}V_A(\overline{B^j_3}) + \frac{2}{3} \sum_{i=1}^{\infty} (\ell_i^j + r_i^j) + \frac{2}{3}\sum_{i=1}^{\infty} (\ell_i^j - r_i^j)(1 - 2s^j_i) \,.
\end{align}
By \eqref{eq:Bobs_value_for_B_3_j_according_to_s_j_is_at_least_one_half}, we have $ V_B^{s^j}(B^j_3) \geq 1/2$, so solving for $V_A(\overline{B_3^j})$ gives
\begin{align} \label{eq:V-A-B-j-3-complement-less-than-sums-of-l-and-r}
    V_A(\overline{B^j_3}) &\leq \frac{1}{4} + \left(\sum_{i=1}^{\infty} \ell_i^j + r_i^j\right) + \left(\sum_{i=1}^{\infty} (\ell_i^j-r_i^j)(1-2s^j_i)\right) \,. 
\end{align}
    For each $i \in \mathbb{N}^*$, either $\ell_i^j=\mu(L_i)=\mu(R_i)$ and $r_i^j=0$ or $\ell_i^j=0$ and $r_i^j=\mu(R_i)=\mu(L_i)$. 
    Then  
    \begin{align} 
        \sum_{i=1}^{\infty} \ell_i^j + r_i^j = \sum_{i=1}^{\infty} \mu(L_i \cup R_i)/2 = \mu((0, 1/2))/2 = 1/4 \,. \label{eq:sum_l_i_j_and_r_i_j}
    \end{align}
    
   Moreover, we can rewrite $\ell_i^j$ and $r_i^j$ in terms of variables $a_i^j \in \{0, 1\}$ where $a_i^j = 0$ if $\ell_i^j = \mu(L_i)$ and $a_i^j = 1$ otherwise. Then  
   \begin{align}
       \sum_{i=1}^{\infty} (\ell_i^j-r_i^j)(1-2s^j_i) &= \sum_{i=1}^{\infty} \mu(L_i) \cdot (1-2a_i^j)(1-2s^j_i) \,. \label{eq:sum-l-minus-r-in-terms-of-a}
   \end{align}
   Rewriting \eqref{eq:V-A-B-j-3-complement-less-than-sums-of-l-and-r} using \eqref{eq:sum_l_i_j_and_r_i_j} and \eqref{eq:sum-l-minus-r-in-terms-of-a}, we get:
    \begin{align}
        V_A(\overline{B^j_3}) &\leq \frac{1}{4} + \frac{1}{4} + \sum_{i=1}^{\infty} \mu(L_i) \cdot (1-2a_i^j)(1-2s^j_i) \\
        &= \frac{1}{2} + \left(\sum_{i=1}^{\infty} \mu(L_i) \right) - 2\sum_{i=1}^{\infty} \mu(L_i) \cdot \left(\frac{1 - (1-2a_i^j)(1-2s^j_i)}{2}\right) \\
        &= \frac{3}{4} - 2\sum_{i=1}^{\infty} \frac{g(i)-g(i+1)}{2} \left(\frac{1 - (1-2a_i^j)(1-2s^j_i)}{2}\right) \explain{By definition of $L_i$ and $g$} \\
        &= \frac{3}{4} - \sum_{i=1}^{\infty} \left(g(i)-g(i+1)\right) \left(a_i^j \oplus s^j_i\right) \label{eq:measurable-va-b3-bounded-by-sj}
    \end{align}
    \paragraph{Choosing $s^*$.} 
    Define a sequence of infinite bit vectors $\widetilde{\vec{a}} = \widetilde{a}^1, \widetilde{a}^2, \ldots,$, such that $\widetilde{a}^j = a^j$ for $j \in \mathcal{J}$ and $\widetilde{a}^j $ is all zeroes for $j \not \in \mathcal{J}$.
By Lemma \ref{lem:far-bitstring} invoked with the sequence $\widetilde{\vec{a}}$, we obtain an infinite vector $s^*$ that has the properties guaranteed by the lemma.

If $(\overline{B^j_1}, B^j_1)$ is not $1/12$-Stackelberg for Alice when  Bob has  density $v_B^{s^*}$, then it is not $\varepsilon$-Stackelberg either. If it is $1/12$-Stackelberg, then $s^*$ is a valid choice for $s^j$, so the conclusions of \eqref{eq:measurable-va-b3-bounded-by-sj} hold with $s^j=s^*$. In other words, we have:
\begin{align}  
V_A(\overline{B^j_3}) &\leq \frac{3}{4} - \sum_{i=1}^{\infty} (g(i) - g(i+1))(a_i^j \oplus s^*_i) \,.\label{eq:V-A-B-3-j-in-terms-of-s-star}
\end{align}

    \paragraph{Bounding Alice's regret in \eqref{eq:measurable-va-b3-bounded-by-sj} using the properties of $s^*$.} Let $N$ a power of $2$ be  large enough so that Lemma \ref{lem:far-bitstring} applies. Since we give asymptotic bounds, we can analyze the regret when the horizon is $T \geq N$. 
    
    Let $n$ be the largest power of $2$ that is still at most $T$. We therefore have $T \leq 2n$. Moreover, using \eqref{eq:V-A-B-3-j-in-terms-of-s-star}, we can upper bound Alice's value for $\overline{B^j_3}$ as follows:
    \begin{align}
        V_A(\overline{B^j_3}) &\leq \frac{3}{4} - \sum_{i=1}^{\infty} (g(i) - g(i+1))(a_i^j \oplus s^*_i) \notag \\
        & \leq \frac{3}{4} - \sum_{i=1}^{10n} (g(i) - g(i+1))(a_i^j \oplus s^*_i) \,. \label{eq:V-A-B-3-j-complement-sum-up-to-10n}
    \end{align}
    The sequence $(B^i_1, \overline{B^i_1})_{i \geq 1}$ consists of all the partitions Alice could make across all possible time horizons. The partitions that could arise for horizon $T$ appear in this sequence at indices $2^T, \ldots, 2 \cdot 2^T - 1$, so we restrict our attention to $j \in \mathcal{J} \cap [2 \cdot 2^T]$. We then have $j \leq 2 \cdot 2^{2n}$. 
    
    Because $s^*$ was chosen by Lemma \ref{lem:far-bitstring}, at least $n/2$ terms of the sum in \eqref{eq:V-A-B-3-j-complement-sum-up-to-10n} are non-zero. Since $g$ is concave up, these non-zero terms are decreasing, so we can upper-bound \eqref{eq:V-A-B-3-j-complement-sum-up-to-10n} by assuming the non-zero terms  are exactly the last $n/2$ terms of the sum:
    \begin{align}
        V_A(\overline{B^j_3}) &\leq \frac{3}{4} - \sum_{i=10n-n/2}^{10n} (g(i) - g(i+1)) \leq \frac{3}{4} - (g(19n/2)-g(10n)) \,. 
    \end{align}
    Because Alice's Stackelberg value is $3/4$ when Bob has density $v_B^{s^*}$, the loss $g(19n/2) - g(10n)$ is a lower bound on Alice's Stackelberg regret in each round. To estimate it, we use Taylor's theorem to express $g(19n/2)$ in terms of $g(10n)$:
    \begin{align}
        g(19n/2) &= f(\log(19n/2)) 
        = f(\log(10n) + \log(19/20))  \notag \\ 
        &= f(\log(10n)) + f'(\log(10n))\log(19/20) + \frac{1}{2}f''(c) \log(19/20)^2 \,,
    \end{align}
    for some $c \in [\log(19n/2), \log(10n)]$. Then, since $f''(x) > 0$ for all $x$:
    \begin{align}
        g(19n/2) - g(10n) &= f'(\log(10n))\log(19/20) + \frac{1}{2}f''(c) \log(19/20)^2 \\
        &> f'(\log(10n))\log(19/20) = |f'(\log(10n))| \cdot \log(20/19) \\
        &\geq |f'(\log(10T))| \cdot \log(20/19) \,.
    \end{align}
    Where the last line follows from $n \leq T$ and $|f'|$ being decreasing. Therefore, reintroducing the factor of $4$ from the $B^j_1$-$B^j_2$ transformation, Alice's Stackelberg regret is at least $|f'(\log(10T))| \cdot \log(20/19)/4$ per round. Setting $f(x) = 1/(x+2)$ and summing across $T$ rounds, her total regret is $\Omega(T / (\log T)^2)$, which completes the proof.
\end{proof}

\begin{lemma}
    \label{lem:far-bitstring}
    Consider an arbitrary sequence  $\vec{a} = a^1, a^2, \ldots \in \{0, 1\}^{\mathbb{N}^*}$ of infinite bit vectors. There exist a constant $n_0 > 0$ independent of $\vec{a}$ and an infinite bit vector $s = s(\vec{a})\in \{0, 1\}^{\mathbb{N}^*}$ such that:
    \begin{itemize}
        \item for every $n \geq n_0$ that is a power of $2$ and all $i \in [2 \cdot 2^{2n}]$, the Hamming distance between the first $10n$ bits of $s$ and the first $10n$ bits of $a^i$ is at least $n/2$.
    \end{itemize}
\end{lemma}
\begin{proof}
    We  choose the bits of $s$ in dyadic intervals. Let $n$ be a large arbitrary power of $2$; we will later choose how large $n$ needs to be.
    Let $I_n = \{5n, 5n+1, \ldots, 10n-1\}$. 
For a vector $v$, let $v|_{I_n}$ denote the restriction of $v$ to the indices in $I_n$.
    Let $|\cdot|$ denote the Hamming weight of a bit vector. We will choose $s|_{I_n}$ with the stronger guarantee that 
    \begin{align} 
    \Bigl|s|_{I_n} \oplus a^i|_{I_n}\Bigr| \geq n/2 \; \; \forall i \in [2 \cdot 2^{2n}]\,.
    \end{align} 
    The choices of $s|_{I_n}$ which fail this condition for any particular $a^i$ are those where fewer than $n/2$ of the bits are flipped between them. These can be counted as follows:
    \begin{align}
        \Bigl|\left\{v \in \{0, 1\}^{5n} \mid \bigl|a^i|_{I_n} \oplus v\bigr| < n/2\right\}\Bigr| &\leq \sum_{j=0}^{\lfloor n/2 \rfloor} \binom{5n}{j} \notag \\
        & 
        \leq \sum_{j=0}^{\lfloor n/2\rfloor} \binom{5n}{\lfloor n/2 \rfloor} \notag \\
        & \leq \left(\frac{n}{2}+1\right) \cdot (10e)^{n/2} \explain{Since $\lfloor n/2 \rfloor \leq n/2$ and $\binom{n}{k} \leq (ne/k)^k$.}
    \end{align}
   Then, by a union bound:
    \begin{align} \label{eq:result_of_union_bound_on_v_and_i}
        \Bigl|\left\{v \in \{0, 1\}^{5n} \mid \exists i \in [2 \cdot 2^{2n}] \mbox{ s.t. } \bigl|a^i|_{I_n} \oplus v\bigr| < n/2 \right\}\Bigr| &\leq 2 \cdot 2^{2n} \left(\frac{n}{2}+1\right) \cdot (10e)^{n/2} \,. 
    \end{align}
    Consider the right-hand side of \eqref{eq:result_of_union_bound_on_v_and_i}. Taking the $n$-th root and letting $n$ approach infinity, we get:
    \begin{align}
        \lim_{n \to \infty} \left(2 \cdot 2^{2n} \left(\frac{n}{2}+1\right) \cdot (10e)^{n/2}\right)^{1/n} &= \lim_{n \to \infty} (n+2)^{1/n} \cdot 4\sqrt{10e} = 4\sqrt{10e} < 2^5 \,. 
    \end{align}
    In particular, there exists a constant $n_0 > 0$ such that for $n \geq n_0$, we have:
    \begin{align}
        \Bigl|\left\{v \in \{0, 1\}^{5n} \mid \exists i \in [2 \cdot 2^n] \mbox{ s.t. } \bigl|a^i|_{I_n} \oplus v\bigr| < n/2 \right\}\Bigr| &< 2^{5n}
    \end{align}
    But $|\{0, 1\}^{5n}| = 2^{5n}$, so there must exist $v^* \in \{0, 1\}^{5n}$ such that  $\bigl|a^i|_{I_n} \oplus v^*\bigr| \geq n/2$ for all $i \in [2 \cdot 2^n]$. 
    
    Let $s|_{I_n} = v^*$. This assignment does not overlap with the blocks of $s$ reserved for other values of $n$, so we can consistently choose the bits of $s$ to the right of the location $5n_0$. The bits in locations situated left of $5n_0$ can be chosen arbitrarily, as they can only increase Hamming distances from $s$. This completes the proof.
\end{proof}

\section{The $k$-Cut Game} \label{app:k-cut}

In the following subsections we give full proofs of the propositions required for Theorems \ref{thm:myopic_Bob}, \ref{thm:2-cut-non-myopic}, \ref{thm:k-cut-non-myopic-bounds}, and \ref{thm:unknown-alpha-bounds}.

\subsection{Upper Bounds for Myopic Bob}

We start with the following helper lemma.

\binarySearchLemma*

\begin{proof}
    We can assume that $x_{k-1} < 1$; otherwise, then trivially no such $x_k^*$ can exist.
    
    For a fixed  $x_k$, let $Z = (Z_1, Z_2)$ denote the alternating partition induced by cut points $x_1, \ldots, x_k$. For convenience, we let $Y_1 \in \{Z_1, Z_2\}$ denote the piece that grows when $x_k$ moves right.
    
    To find if the required $x_k^*$ exists, Alice first cuts at points 
    $x_1, \ldots, x_{k-1}, 1$. This partition maximizes the length of $Y_1$. If Bob nevertheless prefers $Y_2$, no such $x_k^*$ exists. Similarly, Alice can make a second cut at $x_1, \ldots, x_{k-1}, x_k=x_{k-1}$ to maximize the size of $Y_2$. If Bob nevertheless prefers $Y_1$, no such $x_k^*$ exists.

    With those edge cases handled, Alice can now perform binary search to identify the point where Bob switches from preferring $Y_1$ to preferring $Y_2$. The initial feasible interval is $[x_{k-1}, 1]$. In each of the next rounds, she sets $x_k$ to the midpoint of the feasible interval and makes the alternating partition. If Bob prefers $Y_1$, she can update the right endpoint. If Bob prefers $Y_2$, she can update the left endpoint. She continues this process until the feasible interval has length at most $2\varepsilon$, and then returns its midpoint $\widetilde{x}_k$. 
    Since $x_k^*$ is located in the last feasible interval found, we have $|x_k^* - \widetilde{x}_k| \leq \varepsilon$ as required and 
    the process takes $O(\log 1/\varepsilon)$ rounds. 
\end{proof}

\twoCutMyopicAlgorithm*

\begin{proof}
Alice's strategy can be seen as the discretization of a moving-knife procedure that we describe next; afterwards, we explain how to discretize it. 
The moving knife procedure is essentially Austin's procedure and works as follows:
\begin{itemize}
\item Bob is asked to cut the cake in half according to his valuation using  the query  $Cut_{Bob}(1/2)$. Bob answers by reporting his midpoint $m_B$.
\item A referee slides a knife across the cake, from $0$ to $m_B$. For each position $x$ of the referee knife, Bob places his own knife at a point $y(x) > x$ such that $V_B([x, y(x)]) = 1/2$. Then Alice records her value for the better of the two pieces $[x, y(x)]$ and $[0,1] \setminus [x, y(x)]$, defined as 
\begin{align} \alpha(x) = \max\Bigl(V_A([x, y(x)]), V_A([0,1] \setminus [x, y(x)])\Bigr)\,.
\end{align} 
\item Alice cuts at points $x^*$ and $y(x^*)$ such that $x^*$  maximizes $\alpha(x^*)$.
\end{itemize}

This procedure finds a Stackelberg cut because the pieces  $[x, y(x)]$ represent all the intervals of value $1/2$ to Bob. Therefore, each Stackelberg cut gives either Alice or Bob one of these pieces, and Alice chooses whichever one is best for her.

Now we give a discretization of the moving knife procedure.
    Let $ \varepsilon >0 $ and let $\eta = \Delta \varepsilon$. Alice's strategy is the following:
    \begin{enumerate}[1.]
        \item For $i=0, 1,  \ldots$, let $x_i = \eta \cdot i$ and use Lemma \ref{lem:binary-search} to locate $y_i \in [0,1]$ for which Bob  values the pieces $[x_i, y_i]$ and $[0, x_i] \cup [y_i, 1]$ equally up to $\varepsilon$ error. Lemma \ref{lem:binary-search} will succeed if and only if $V_B([0, x_i]) < 1/2$, so once it fails we stop. Let $\widetilde{N}$ be the highest index that succeeded.
         
        \item Let $i^*$ be the index such that the cut at $x_{i^*}$ and $y_{i^*}$ contains the side of greatest value to Alice, i.e. 
        \begin{align} i^* = \argmax_{i \in \{0, \ldots, \widetilde{N}\}} \max \bigl(V_A([x_i, y_i]), V_A([0, x_i] \cup [y_i, 1]) \bigr) \,.
        \end{align}
        \item Set $\overline{x}$ and $\overline{y}$ to the values of $x_{i^*}$ and $y_{i^*}$,
        but shifted  by a total of $\varepsilon \Delta/\delta$ in the directions that increase the size of the side Alice values less. Formally: 
        \begin{itemize}
            \item If Alice prefers $[x_{i^*}, y_{i^*}]$ to $[0, x_{i^*}] \cup [y_{i^*}, 1]$, then  choose $\overline{x} \geq x_{i^*}$ and $\overline{y} \leq y_{i^*}$ such that $(\overline{x} - x_{i^*}) + (y_{i^*} - \overline{y}) = \varepsilon \Delta/\delta$.
            \item Else, choose $\overline{x} \leq x_{i^*}$ and $\overline{y} \geq y_{i^*}$ such that $(x_{i^*} - \overline{x}) + (\overline{y} - y_{i^*}) = \varepsilon \Delta/\delta$.
        \end{itemize}
        \item For all remaining rounds, cut at $\overline{x}$ and $\overline{y}$.
    \end{enumerate}
    
    Alice's overall regret will come from three parts: (i) the loss incurred in step 1, (ii) the loss from $\overline{x}$ and $\overline{y}$ not quite being a Stackelberg cut, and (iii) the potential loss from Bob taking the  side preferred by Alice in the remaining rounds. We estimate each part.
    \begin{enumerate}[(i)]
        \item Each search following Lemma \ref{lem:binary-search} can be done in $O(\log 1/\varepsilon)$ rounds. Finding all the $y_i$ therefore takes $O\left(\frac{\log 1/\varepsilon}{\eta}\right)$ rounds.
        \item Let $x^*$ and $y^*$ be cut points that achieve Alice's Stackelberg value. Thus Bob equally values the pieces $[x^*, y^*]$ and $[0, x^*] \cup [y^*, 1]$. Then there exists $i \in \{0, \ldots, \widetilde{N}\}$ such that $|x_i - x^*| \leq \eta$, so $|y_i - y^*| \leq \eta/\delta + \varepsilon$. Cutting at $x_i$ and $y_i$ then gets Alice close to her Stackelberg value, so cutting at $x_{i^*}$ and $y_{i^*}$ is at least that good. More precisely:
        \begin{align}
            u_A^*(2)&\leq  \left(\eta + \frac{\eta}{\delta} + \varepsilon\right) \Delta  + \max \bigl(V_A([x_i, y_i]), V_A([0, x_i] \cup [y_i, 1]) \bigr) \notag \\
            &\leq \left(\eta + \frac{\eta}{\delta} + \varepsilon\right) \Delta  + \max \bigl(V_A([x_{i^*}, y_{i^*}]), V_A([0, x_{i^*}] \cup [y_{i^*}, 1]) \bigr) \,.
        \end{align}
        The shifting of $\varepsilon \Delta/\delta$ in Bob's favor adds $O(\varepsilon \Delta^2 / \delta)$ loss to Alice, so Alice's loss compared to her Stackelberg value over the $T$ rounds is $O(T \eta \Delta / \delta)$.
        \item Shifting the cut points in Bob's favor ensures Bob will always take the side Alice wants him to. Bob values each side of the cut at $x_{i^*}$ and $y_{i^*}$ as within $\Delta \varepsilon$ of $1/2$. Therefore, shifting the cut points to $\overline{x}$ and $\overline{y}$ to make one side bigger ensures he will take it, so Alice experiences no regret from this source.
    \end{enumerate}
    Adding the nonzero terms from (i) and (ii) together, we get a regret bound of 
        $O\left(\frac{\log 1/\varepsilon}{\eta} + \frac{T\eta \Delta}{\delta}\right) \,.$
    Setting  $\varepsilon = \sqrt{(\log T)/T}$ yields  a regret bound of $O\left(\sqrt{T \log T}\right)$.
\end{proof}

\begin{restatable}{lemma}{knownBobIntervalLemma}\label{lem:known-bob-intervals-gives-stackelberg-cut}
    Let $n \in \mathbb{N}$ and let $I_1, \ldots, I_n$ be disjoint sub-intervals of $[0, 1]$. Let $U = [0, 1] \setminus \bigcup_{i=0}^n I_i$ be the cake uncovered by the sub-intervals. Suppose there exist $\eta \in (0, 1)$,  $\varepsilon \in (0, \eta^2/2)$, and $r \in \mathbb{N}^*$ such that $|V_B(I_i) - \eta| \leq \varepsilon$ for all $i \in [n]$ and $V_B(U) \leq r(\eta + \varepsilon)$. Then there exists a $(k+2r)(\eta + \varepsilon)\Delta/\delta$-Stackelberg $k$-cut for Alice that does not cut in the interior of any $I_i$ and where Bob prefers the piece containing more of the sub-intervals.
\end{restatable}

\begin{proof}
    Let $y_1 \leq \ldots \leq y_k$ be cut points that achieve Alice's Stackelberg value. Let $Z_A$ be the piece of the alternating partition at $y_1, \ldots, y_k$ that Alice prefers and let $Z_B$ be the other piece. We will move the cut points to grow $Z_B$ and ensure the resulting partition satisfies the lemma conditions.
    
    We first must ensure no cut point is in the interior of a sub-interval. Let $C \subseteq [0, 1]$ consist of the endpoints of all the sub-intervals along with all of $U$. Let $y'_1, \ldots, y'_k$ be the $y_i$, but each rounded to the nearest point in $C$ in the direction that weakly grows $Z_B$. Let $Z'_B$ and $Z'_A$ be the resulting pieces of the alternating partition at $y'_1, \ldots, y'_k$. 

    We now must ensure that $Z'_B$ contains more of the sub-intervals. Let $n_B$ be the number of sub-intervals in $Z'_B$ and let $n_A$ be the number of sub-intervals in $Z'_A$. In order to satisfy Bob's preferences while each sub-interval has value $\eta \pm \varepsilon$, we must have:
\begin{itemize}
\item $n_B \geq \frac{\frac{1}{2} - V_B(U)}{\eta + \varepsilon} \geq \frac{\frac{1}{2} - r(\eta + \varepsilon)}{\eta + \varepsilon} = \frac{1}{2(\eta + \varepsilon)} - r$;  
\item  $n_A \leq \frac{\frac{1}{2}}{\eta - \varepsilon} = \frac{1}{2(\eta - \varepsilon)}$.
    \end{itemize}
    Since $\max(\varepsilon, \varepsilon^2) < \eta^2/2$, we can bound their difference as follows:
    \begin{align}
        n_A - n_B &\leq \frac{1}{2(\eta - \varepsilon)} - \frac{1}{2(\eta + \varepsilon)} + r = \frac{\varepsilon}{\eta^2 - \varepsilon^2} + r < 1 + r \,. 
    \end{align}
    Since $r \geq 1$, shifting an extra $r$ sub-intervals from $Z'_A$ to $Z'_B$ is enough to make $Z'_B$ have more sub-intervals. Let $Z^*_A$ and $Z^*_B$ be the resulting pieces after this shift. Compared to the original $Z_B$ and $Z_A$, we have now given at most $k+r$ sub-intervals and all of $U$ to Bob:
    \begin{align}
        V_B(Z^*_B \setminus Z_B) &= V_B(Z^*_B \setminus Z'_B) + V_B(Z'_B \setminus Z_B) \\
        &\leq r(\eta + \varepsilon) + V_B((Z^*_B \setminus Z'_B) \cap U) + k(\eta + \varepsilon) + V_B((Z'_B \setminus Z_B) \cap U) \\
        &\leq (k+r)(\eta + \varepsilon) + V_B(U) \\
        &\leq (k+2r)(\eta + \varepsilon) \,.
    \end{align}
    Since Alice's value density is within a factor of $\Delta/\delta$ of Bob's, her loss compared to $Z_A$ satisfies:
    \begin{align}
        V_A(Z_A \setminus Z^*_A) &\leq \frac{\Delta}{\delta} V_B(Z^*_B \setminus Z_B) \leq (k+2r)(\eta + \varepsilon)\frac{\Delta}{\delta}\,.
    \end{align}
Thus the partition $(Z^*_A, Z^*_B)$ is $(k+2r)(\eta + \varepsilon)\Delta/\delta$-Stackelberg.
\end{proof}

\kCutMyopicUpper*

\begin{proof}
    Alice's strategy has three parts:
    \begin{itemize}
        \item Find Bob's midpoint $m_B$.
        \item Use the knowledge of $m_B$ to construct cuts that allow her to find many small intervals of (nearly) equal value to Bob, and map out Bob's value density with them.
        \item Compute the optimal cut given this approximate information about Bob's value density and cut there for the rest of the rounds.
    \end{itemize}

    Specifically, given parameters $\eta = (Tk)^{-1/2}$ and $\varepsilon = \delta^2 \eta^2 / 2$, Alice's strategy is as follows.

    \begin{enumerate}
        \item Use Lemma \ref{lem:binary-search} to get an approximation $\widetilde{m}_B$ of Bob's midpoint $m_B$ to within $\varepsilon$.
        \item Set $x_0 = \widetilde{m}_B - \eta$ and $x_1 = \widetilde{m}_B$. Set $i=1$. Repeat the following until the loop exits:
        \begin{enumerate}
            \item Use Lemma \ref{lem:binary-search} to find the $z$ for which $V_B((0, \widetilde{m}_B - \eta) \cup (x_i, z))=1/2$, to within $\varepsilon$. If the search fails, exit the loop.
            \item Set $x_{i+1} = z$ and increment $i$.
        \end{enumerate}
        \item Reset $i=0$. Repeat the following until the loop exits:
        \begin{enumerate}
            \item Use Lemma \ref{lem:binary-search} to find the $z$ for which $V_B((z, x_i) \cup (x_2, 1))=1/2$, to within $\varepsilon$. This can be accomplished by treating the cake as reflected left-right. If the search fails, exit the loop.
            \item Set $x_{i-1} = z$ and decrement $i$.
        \end{enumerate}
        \item Let $-\ell$ be the index of the last negative-index $x$ and let $r$ be the index of the last positive-index $x$. Let $a^*_1, \ldots, a^*_k$ be cut points chosen from the set $[0, x_{-\ell}] \cup \{x_{-\ell}, \ldots, x_r\} \cup [x_r, 1]$ in order to maximize Alice's value of the piece containing fewer of the intervals $(x_i, x_{i+1})$ for $i \in \{-\ell, \ldots, r-1\}$.
        \item Let $a_1, \ldots, a_k$ be cut points that deviate from $a_1^*, \ldots, a_k^*$ by exactly $\frac{3\Delta\varepsilon}{\delta^2 \eta} + \Delta \eta/\delta$ distance in ways that make Alice's less-preferred piece bigger. Cut there for the remaining rounds.
    \end{enumerate}
    We will eventually use Lemma \ref{lem:known-bob-intervals-gives-stackelberg-cut} on the intervals $(x_i, x_{i+1})$ to see how good Alice's cut is. To do so, we need to bound how similarly Bob values these intervals. Our reference point will be $(x_0, x_1) = (\widetilde{m}_B - \eta, \widetilde{m}_B)$. Each iteration of the loop in step 2 aims to find a new interval $(x_i, x_{i+1})$ of similar value. To see how well it does so, fix $i \in \{1, \ldots, r - 1\}$ and let $z^*$ be the exact solution Lemma \ref{lem:binary-search} aimed to find, i.e. such that $V_B((0, x_0) \cup (x_i, z^*))=1/2$. We then have $|x_{i+1} - z^*| \leq \varepsilon$. By the choice of $x_0$ and $x_1$, we also have:
    \begin{align}
        \frac{1}{2} &= V_B((0, x_0) \cup (x_i, z^*)) = V_B((0, x_1)) - V_B((x_0, x_1)) + V_B((x_i, z^*)) \\
        &= V_B((0, \widetilde{m}_B)) + \left(V_B((x_i, z^*)) - V_B((x_0, x_1))\right) \,.
    \end{align}
    By the construction of $\widetilde{m}_B$ in step 1, we also know that $|m_B - \widetilde{m}_B| \leq \varepsilon$, so $|V_B((0, \widetilde{m}_B)) - 1/2| \leq \Delta \varepsilon$. By the triangle inequality, we therefore have:
    \begin{align}
        \left|V_B((x_i, x_{i+1})) - V_B((x_0, x_1))\right| \leq& \left|V_B((x_i, x_{i+1})) - V_B((x_i, z^*))\right| + \left|V_B((x_i, z^*)) - V_B((x_0, x_1))\right| \notag \\
        \leq& \Delta \varepsilon + \Delta \varepsilon = 2\Delta \varepsilon \,.
    \end{align}
    We also know that the search failed when trying to find $x_{r+1}$, so $V_B((0, x_0) \cup (x_r, 1)) \leq 1/2$. Therefore:
    \begin{align}
        V_B((x_r, 1)) &\leq \frac{1}{2} - V_B((0, x_0)) = \frac{1}{2} - V_B((0, \widetilde{m}_B)) + V_B((x_0, x_1)) \leq V_B((x_0, x_1)) + \Delta \varepsilon \,. \label{eq:vb-xr-1-less-than-vb-x0-x1-plus-delta-epsilon}
    \end{align}
    Moving to to the negative $x$'s, fix $i \in \{-\ell+1, \ldots, 0\}$ and let $z^*$ be the exact solution Lemma \ref{lem:binary-search} aimed to find, i.e. such that $V_B((z^*, x_i) \cup (x_2, 1)) = 1/2$. We then have $|x_{i-1} - z^*| \leq \varepsilon$. By the choice of $x_1$ and $x_2$, we also have:
    \begin{align}
        \frac{1}{2} &= V_B((z^*, x_i) \cup (x_2, 1)) = V_B((x_1, 1)) - V_B((x_1, x_2)) + V_B((z^*, x_i)) \notag \\
        &= V_B((0, \widetilde{m}_B)) + \left(V_B((z^*, x_i)) - V_B((x_1, x_2))\right) \,. 
    \end{align}
   By the construction of $\widetilde{m}_B$, we have $|V_B((\widetilde{m}_B, 1)) - 1/2| \leq \Delta \varepsilon$. By the analysis of the positive $x$'s, we have $|V_B((x_0, x_1)) - V_B((x_1, x_2))| \leq 2\Delta \varepsilon$. Then by the triangle inequality we have:
    \begin{align}
        \left|V_B((x_{i-1}, x_i)) - V_B((x_0, x_1))\right| &\leq \Delta \varepsilon + 2\Delta \varepsilon = 3\Delta \varepsilon \,.
    \end{align}
    Similarly to $(x_r, 1)$, since the search failed when trying to find $x_{-\ell-1}$, it must be the case that $V_B((0, x_{-\ell}) \cup (x_2, 1)) \leq 1/2$. Therefore:
    \begin{align}
        V_B((0, x_{-\ell})) &\leq \frac{1}{2} - V_B((x_2, 1)) = \frac{1}{2} - V_B((\widetilde{m}_B, 1)) + V_B((x_1, x_2)) \leq V_B((x_0, x_1)) + 3\Delta \varepsilon \,. \label{eq:vb-0-xminusl-less-than-3-delta-epsilon}
    \end{align}
    Combining \eqref{eq:vb-xr-1-less-than-vb-x0-x1-plus-delta-epsilon} and \eqref{eq:vb-0-xminusl-less-than-3-delta-epsilon}, the cake uncovered by any of the intervals has value at most $2V_B((x_0, x_1)) + 4\Delta \varepsilon$ to Bob. Invoking Lemma \ref{lem:known-bob-intervals-gives-stackelberg-cut} with interval values of $V_B((x_0, x_1)) \pm 3\Delta\varepsilon$ and uncovered cake of value at most $2(V_B((x_0, x_1)) + 3\Delta\varepsilon)$, we get that there exists a $(k+4)(V_B((x_0, x_1))+3\Delta\varepsilon)\Delta/\delta$-Stackelberg cut among those that $a_1^*, \ldots, a_k^*$ was chosen from. Therefore, the cut $a_1^*, \ldots, a_k^*$ is at least that good. Factoring in the extra transfer to Bob in step 5 and using $V_B((x_0, x_1)) \leq \Delta \eta$, we get that the cut $a_1, \ldots, a_k$ is $O(k\eta + \varepsilon/\eta) = O(\sqrt{k/T})$-Stackelberg if Bob takes the piece Alice expects him to.

    To show Bob will take his expected piece, let $n = r+\ell$ be the number of sub-intervals. Using the upper bounds for Bob's value of each part of the cake, we have $(n+2)(V_B((x_0, x_1))+3\Delta\varepsilon) \geq 1$, so $n \geq 1/(V_B((x_0, x_1)) + 3\Delta\varepsilon) - 2$. Under the cut $a_1^*, \ldots, a_k^*$, Alice expects Bob to take the piece containing at least $n/2$ sub-intervals, each of which has value at least $V((x_0, x_1)) - 3\Delta\varepsilon$ to Bob. Then:
    \begin{align}
        \frac{1}{2} - \frac{n}{2}(V_B((x_0, x_1)) - 3\Delta\varepsilon) &\leq \frac{1}{2} - \frac{1}{2} \left( \frac{V_B((x_0, x_1)) - 3\Delta\varepsilon}{V_B((x_0, x_1)) + 3\Delta\varepsilon} + 2(V_B((x_0, x_1)) - 3\Delta\varepsilon)\right) \\
        &= \frac{3\Delta\varepsilon}{V_B((x_0, x_1)) + 3\Delta\varepsilon} + V_B((x_0, x_1)) - 3\Delta\varepsilon \\
        &\leq \frac{3\Delta\varepsilon}{\delta \eta} + \Delta \eta \,. \explain{Since $v_B$ has bounded density and  $|[x_0, x_1]|=\eta$.}
    \end{align}
    Therefore, the adjustment in step 5 gives Bob enough extra cake to guarantee he prefers the piece with more of the sub-intervals.

    Alice therefore accumulates $O(\sqrt{Tk})$ regret after finding her cut points $a_1, \ldots, a_k$. Finding those cut points took $O(n)$ uses of Lemma \ref{lem:binary-search}. Each use of Lemma \ref{lem:binary-search} takes $O(\log(1/\varepsilon)) = O(\log(Tk))$ rounds. We can upper-bound $n$ by using the fact that $n$ intervals of value at least $V_B((x_0, x_1)) - 3\Delta\varepsilon$ fit into $[0, 1]$, so:
    \begin{align}
        n &\leq \frac{1}{V_B((x_0, x_1)) - 3\Delta\varepsilon} \leq \frac{1}{\delta \eta - 3\Delta\varepsilon} \in O(1/\eta) = O(\sqrt{Tk}) \,.
    \end{align}
    Therefore, Alice's regret from this strategy is $O(\sqrt{Tk} + \sqrt{Tk} \log(Tk)) = O(\sqrt{Tk} \log(Tk))$.
\end{proof}

\subsection{Upper Bounds for Non-Myopic Bob}\label{sec:general-bob-upper-bounds}

In this section we give proofs of the upper bounds of Theorems \ref{thm:2-cut-non-myopic}, \ref{thm:k-cut-non-myopic-bounds}, and \ref{thm:unknown-alpha-bounds}. Alice's strategies are similar to the myopic case, but using the following lemma instead of Lemma \ref{lem:binary-search} to learn Bob's value density.

\begin{restatable}{lemma}{binarySearchWithRegret} \label{lem:binary-search-with-regret}
Let $T$ be the time horizon.   Suppose Bob's strategy guarantees  regret  at most $c \cdot f(T)$, where $f: \mathbb{N} \to \mathbb{R}$  is a public function and $c > 0$ is a constant private to Bob. Let $\varepsilon > 0$ and $x_1, \ldots, x_{k-1} \in [0, 1]$ such that $x_1 \leq x_2 \leq \ldots x_{k-1}$. Suppose $\varepsilon \leq f(T)$ and let $g: \mathbb{N} \to \mathbb{R} \in \omega(1)$. 
Then there exists $T_0 = T_0(g, c)$ such that if $T \geq T_0$, then within $O\bigl(\frac{f(T)}{\varepsilon} \log(1/\varepsilon) g(T)\bigr)$ rounds Alice can do one of the following:
    \begin{itemize}
        \item find $\widetilde{x}_k \in [x_{k-1}, 1]$ such that 
        Bob equally values both pieces of the alternating partition induced by cuts $x_1, \ldots, x_{k-1}, x_k^*$ for some $x_k^* \in [\widetilde{x}_k-\varepsilon, \widetilde{x}_k +\varepsilon] $.
        \item determine that no such $x_k^*$ exists in $(x_{k-1} + \varepsilon, 1 - \varepsilon)$.
    \end{itemize}
\end{restatable}

\begin{proof}
    We can assume that $1 - x_{k-1} \geq 2\varepsilon$; otherwise, trivially no such $x_k^*$ can exist.
    
    Alice's algorithm is similar to that of Lemma \ref{lem:binary-search}, but where she makes each cut enough times to make it impossible for Bob to fool her and stay within his regret bound. Specifically, she starts by instantiating $i=0$, then setting $\ell_0=x_{k-1}$ and $r_0=1$. Then, as long as $|r_i - \ell_i| > \varepsilon/2$, she does the following:
    \begin{enumerate}
        \item Set $z_i = (\ell_i + r_i) / 2$ as the midpoint of the current interval.
        \item For each of the next $\lceil 4f(T) g(T) / (\delta \varepsilon)\rceil$ rounds, cut at $[x_1, \ldots, x_{k-1}, z_i]$. Let $P_i$ and $Q_i$ be the pieces of the alternating partition with these cut points that Bob chose most and least often, respectively, in this set of rounds. 
        \item  If $[\ell_i, z_i] \subseteq P_i$, then set $\ell_{i+1} = \ell_i$ and $r_{i+1} = z_i$. Otherwise, set $\ell_{i+1} = z_i$ and $r_{i+1} = r_i$.
        \item Increment $i$.
    \end{enumerate}
    Let $n$ be the final index $i$ of the algorithm. If $\ell_n = x_{k-1}$ or $r_n = 1$, Alice returns that no such $x_k^*$ exists in $(x_{k-1} + \varepsilon, 1 - \varepsilon)$. Otherwise, she returns $\widetilde{x}_k = (r_n + \ell_n)/2$.

    \medskip 
    
Let $T_0$ be such that $g(T) > c$ for all $T \geq T_0$. It exists because $g \in \omega(1)$. To analyze Alice's algorithm, we first claim the following:
 \begin{description}
     \item[(a)] If there exists $x_k^* \in [x_{k-1}, 1]$ such that Bob equally values both sides of the alternating partition induced by cuts $x_1, \ldots, x_{k-1}, x_k^*$, then $x_k^* \in [\ell_i - \varepsilon/4, r_i + \varepsilon/4]$.
 \end{description}

    We prove  claim (a) by induction on $i$. As a base case, when $i=0$ the interval $[\ell_0 - \varepsilon / 4, r_0 + \varepsilon / 4]$ contains the entire search area $[x_{k-1}, 1]$, so if $x_k^*$ exists it will be in there.

    Now assume that claim (a) holds for an arbitrary $i$. If the required $x_k^*$ does not exist, the inductive hypothesis is vacuously true.  Otherwise, we consider two cases:
    \begin{itemize}
        \item Case 1: Bob strictly prefers $P_i$. Then $P_i$ contains $x_k^*$. Alice's setting of $\ell_{i+1}$ and $r_{i+1}$ ensures the interval $[\ell_{i+1} - \varepsilon / 4, r_{i+1} + \varepsilon / 4]$ contains $P_i \cap [\ell_{i} - \varepsilon / 4, r_{i} + \varepsilon / 4]$. The induction hypothesis gives $x_k^* \in [\ell_{i} - \varepsilon / 4, r_{i} + \varepsilon / 4]$, so  
        $x_k^* \in [\ell_{i+1} - \varepsilon / 4, r_{i+1} + \varepsilon / 4]$ as required.
        \item Case 2: Bob weakly prefers $Q_i$. Then Bob chose the piece he prefers less at least $\frac{2f(T)g(T)}{\delta \varepsilon}$ times. Denote $V_B(Q_i) = {1}/{2}-\beta$. Then Bob's utility loss is  at least 
\begin{align} \label{eq:Bob_utility_loss_Pi_minus_Qi}
(V_B(P_i) - V_B(Q_i)) \frac{2f(T)g(T)}{\delta \varepsilon}  =   \frac{4 \beta f(T)g(T)}{\delta \varepsilon}\,. 
\end{align}
        
      Bob's loss cannot exceed his regret bound, so continuing from \eqref{eq:Bob_utility_loss_Pi_minus_Qi} gives:
      \begin{align} \label{eq:c_FT_at_least_stuff}
          \frac{4 \beta f(T)g(T)}{\delta \varepsilon} \leq c \cdot f(T)\,.
     \end{align} 
    Applying $g(T)>c$ yields  $\beta < \delta \varepsilon/4$. 

    We know $x_k^* \in P_i$. Let $I$ be the interval with endpoints $z_i$ and $x_k^*$. By definition of $\beta$, we have $V_B(I) = \beta$. 
Since Bob's density is lower bounded by $\delta$, we have 
\begin{align}
    \frac{\delta \varepsilon}{4} > \beta = V_B(I) \geq \delta |I| = \delta  |z_i - x_k^*| \,. 
\end{align}
Thus $|z_i - x_k^*| < \varepsilon/4$. Since $z_i \in [\ell_{i+1}, r_{i+1}]$, we have $x_k^* \in [\ell_{i+1} - \varepsilon/4, r_{i+1} + \varepsilon/4]$ as required.
    \end{itemize}
    In both cases the inductive step holds, which completes the proof of claim (a).

We  now  claim the following:
\begin{description}
    \item[(b)] If there is no $x_k^* \in [x_{k-1}, 1]$ such that Bob equally values both sides of the alternating partition induced by cuts $x_1, \ldots, x_{k-1}, x_k^*$, then $\ell_n= x_{k-1}$ or $r_n = 1$.
\end{description}

We show that if the requirements of claim (b) are met, then for each  $i \in \{0, \ldots, n-1\}$  Bob always prefers $P_i$, recalling that $P_i$ is the piece that Bob chose most often for index $i$. 

Assume towards a contradiction that Bob prefers $Q_i$ for some index $i$. Then  Bob chose $P_i$ (the piece he prefers less) at least $\frac{2f(T)g(T)}{\delta \varepsilon}$ times. Denote $V_B(P_i) = {1}/{2}-\beta$. Then Bob's utility loss is  at least 
\begin{align} \label{eq:Bob_utility_loss_Qi_minus_Pi}
(V_B(Q_i) - V_B(P_i)) \frac{2f(T)g(T)}{\delta \varepsilon}  =   \frac{4 \beta f(T)g(T)}{\delta \varepsilon}\,. 
\end{align}
        
      Since Bob's loss is at  most $c f(T)$, we obtain 
      \begin{align} \label{eq:c_FT_at_least_stuff_contradiction}
          \frac{4 \beta f(T)g(T)}{\delta \varepsilon} \leq c \cdot f(T)\,.
     \end{align} 
    Applying $g(T)>c$ yields  $\beta < \delta \varepsilon/4$.
    Let $I = Q_i \cap [\ell_i, r_i]$. Since $z_i$ is the midpoint of $[\ell_i, r_i]$, interval $I$ is the left or right half of $[\ell_i, r_i]$. Because no suitable $x_k^*$ exists, even adding $I$ to $P_i$ would not make him prefer it to $Q_i$, so $V_B(I) \leq \beta$. Since Bob's density is lower bounded by $\delta$, we get    
    \begin{align}
     \delta \cdot \frac{r_i - \ell_i}{2} = \delta \cdot |I|  \leq V_B(I) \leq  \beta < \delta \varepsilon/4,
    \end{align}
    so $r_i - \ell_i < \varepsilon/2$. But this was precisely Alice's stopping condition, so the iteration for index $i$ could not have happened. Thus for each index $i$ Bob prefers $P_i$.

    Since no suitable $x_k^*$ exists, there is no threshold where Bob switches which piece he prefers. Then choosing $P_i$ either consistently tells Alice to recurse left or consistently tells Alice to recurse right. Therefore, we either have $\ell_n = x_{k-1}$ or $r_n = 1$. This proves claim (b).
    
    We can now show the correctness of Alice's algorithm. 
\begin{itemize} 
\item If $\ell_n=x_{k-1}$ or $r_n=1$, then Alice returns that no such $x_k^*$ exists in $(x_{k-1} + \varepsilon, 1 - \varepsilon)$. She cannot be wrong: if she were, then by claim (a) and the stopping condition $r_n - \ell_n \leq \varepsilon/2$, we would have $\ell_n > x_{k-1} + \varepsilon/4$ and $r_n < 1 - \varepsilon/4$, which is a contradiction.
\item Else, Alice returns $\widetilde{x}_k = (r_n +\ell_n)/2$. In this case  $\ell_n \neq x_{k-1}$ and $r_n \neq 1$. By claim (b) there  exists a suitable $x_k^* \in [x_{k-1},1]$. By claim (a) we have  $x_k^* \in [\ell_n - \varepsilon/4, r_n + \varepsilon/4]$. Thus $\widetilde{x}_k \in [x_{k-1},1]$ and $| \widetilde{x}_k - x_k^*| \leq \varepsilon$ as required.
\end{itemize}
  Since Alice's algorithm has  $O\bigl({f(T)g(T)}/{\varepsilon}\bigr)$ rounds for each  $i \in \{0, \ldots, n-1\}$ and $n \leq \log_2 (2/\varepsilon)$, the total number of rounds is  $O\bigl(\frac{f(T)}{\varepsilon} \log(1/\varepsilon) g(T)\bigr)$  as desired.
\end{proof}

\begin{restatable}{lemma}{twoCutNonMyopicGeneralUpperBound}
    \label{lem:2-cut-non-myopic-general-upper-bound}
    Let $k=2$. Let $f : \mathbb{N} \to \mathbb{R}$ be a function known to both players such that $f(x) \in o(x/\log^2 x)$ and $f(x) \in \Omega(x^{\alpha})$ for some $\alpha > -1/2$. Suppose Bob's strategy guarantees his regret is at most $c \cdot f(T)$, where $c > 0$ is a constant private to Bob. Then Alice has a deterministic strategy which guarantees her regret is $O\bigl((f(T))^{1/3}T^{2/3}\log^{2/3} T\bigr)$.
\end{restatable}

\begin{proof}
Alice's strategy is similar to that of Proposition \ref{prop:two_cuts_myopic_Bob} and it is based on the moving knife procedure described there.
Let $\varepsilon = (f(T))^{1/3}T^{-1/3}\log^{2/3}T$ and $\eta = \varepsilon \Delta/\delta$. Let $g: \mathbb{N} \to \mathbb{R}$ be $g(x) = \log{x}$ and assume $T$ is sufficiently large that Lemma \ref{lem:binary-search-with-regret} can be used.
Alice's strategy is as follows:
    \begin{enumerate}[1.]

    \item Initialize $i=0$. Repeat the next step until it terminates:
    \begin{itemize}
        \item Invoke Lemma \ref{lem:binary-search-with-regret} with $k=2$, the function $g$, and cut point  $z_i =  \eta \cdot i$. We satisfy $\varepsilon \leq f(T)$ because of the choice of $\varepsilon$ and the restriction $f(T) \in \Omega(T^{\alpha})$ for some $\alpha > -1/2$. If the lemma returns  $y_i \in [z_i, 1]$ such that 
        Bob equally values pieces $[z_i, y_i^*]$ and $[0, z_i] \cup [y_i^*, 1]$  for some $y_i^* \in [y_i-\varepsilon, y_i +\varepsilon]$, then increment $i$. Else, set ${N}=i-1$ and stop.
    \end{itemize}
        \item Let $\ell$ be the index such that the cut at $z_{\ell}$ and $y_{\ell}$ contains the piece of greatest value to Alice, i.e.
           $ \ell = \argmax_{i \in \{0, \ldots, {N}\}} \max \bigl(V_A([z_i, y_i]), V_A([0, z_i] \cup [z_i, 1])\bigr) \,.$
        \item Set $\overline{z}$ and $\overline{y}$ to the values of $z_{\ell}$ and $y_{\ell}$,
        but shifted  by a total of $2\eta$ in the directions that increase the size of the piece  Alice values less. Formally:
        \begin{itemize}
            \item If Alice prefers $[z_{\ell}, y_{\ell}]$ to $[0, z_{\ell}] \cup [y_{\ell}, 1]$, then choose $\overline{z} \geq z_{\ell}$ and $\overline{y} \leq y_{\ell}$ such that $(\overline{z} - z_{\ell}) + (y_{\ell} - \overline{y}) = 2\eta$.
            \item Else, choose $\overline{z} \leq z_{\ell}$ and $\overline{y} \geq y_{\ell}$ such that $(z_{\ell} - \overline{z}) + (\overline{y} - y_{\ell}) = 2\eta$.
        \end{itemize}
        \item For all remaining rounds, cut at $\overline{z}$ and $\overline{y}$.
    \end{enumerate}

\paragraph{Estimating $z_N$.}  Before analyzing Alice's regret, we estimate $z_{{N}}$. To truly discretize the moving-knife procedure, we need $z_{{N}} \approx m_B$, where $m_B$ is Bob's midpoint. We claim that $z_{{N}} \geq m_B - 2\eta$. Suppose for contradiction that $z_{{N}} < m_B - 2\eta$ $(\dag)$. Then the execution of Lemma \ref{lem:binary-search-with-regret} that failed was at index $N+1$ with cut point $z_{N+1} = (N+1) \eta$. Adding $\eta$ to both sides of $(\dag)$ we get $z_N + \eta < m_B - \eta$. Since $\eta = \varepsilon \Delta/\delta$, we get $z_{N+1} < m_B - \varepsilon \Delta/\delta$. Since Bob's density is lower bounded by $\delta$ we have  $V_B([z_{N+1}, m_B]) > \varepsilon \Delta$.  Bob's density is also at most $\Delta$, so     \begin{align}
        V_B([0, z_{N+1}] \cup [1-\varepsilon, 1]) &= V_B([0, m_B]) - V_B([z_{N+1}, m_B]) + V_B([1-\varepsilon, 1]) 
        < \frac{1}{2} - \varepsilon \Delta + \varepsilon \Delta = \frac{1}{2} \,. \notag 
    \end{align} 
    Since $V_B([0, z_{N+1}] \cup [m_B, 1]) > 1/2$, there exists $y_{N+1}^* \in (z_{N+1} + \varepsilon, 1 - \varepsilon)$ by the intermediate value theorem such that $V_B([0, z_{N+1}] \cup [y_{N+1}^*, 1]) = 1/2$. But then 
    Alice's strategy should have succeeded when invoking Lemma  \ref{lem:binary-search-with-regret} in iteration $N+1$. However, this contradicts the definition of $N$, so the assumption must have been false and  $z_{{N}} \geq m_B - 2\eta$.
        
\paragraph{Estimating Alice's regret.}  We now estimate Alice's overall regret in three parts: (i) the loss incurred in step 1; (ii) the loss from $\overline{z}$ and $\overline{y}$ not quite being a Stackelberg cut; and (iii) the loss from Bob choosing Alice's more preferred piece once Alice has started cutting at $\overline{z}$ and $\overline{y}$.
    \begin{enumerate}[(i)]
        \item Each invocation of Lemma \ref{lem:binary-search-with-regret} in step 1 takes $O\left(\frac{f(T)}{\varepsilon} \log(1/\varepsilon) \log T\right)$ rounds.
        
        Since $\eta \in o(1)$ by the choice of $\varepsilon$ and the restriction $f(T) \in o(T/\log^2 T)$, we have $N = O(1/\eta)$, so finding all the points $y_0, \ldots, y_N$  takes $O\left(\frac{f(T)}{\varepsilon \eta} \log(1/\varepsilon) \log T\right)$ rounds.
        \item Let $\widehat{z}$ and $\widehat{y}$ be cut points where Alice achieves her  Stackelberg value. Thus Bob equally values the pieces $[\widehat{z}, \widehat{y}]$ and $[0, \widehat{z}] \cup [\widehat{y}, 1]$. Since $\widehat{z} \in [0, m_B]$ and  $z_{{N}} \geq m_B - 2\eta$, there exists $i \in \{0, \ldots, {N}\}$ such that $|z_i - \widehat{z}| \leq 2\eta$. Then since $V_B([z_i, y_i^*]) = V_B([\widehat{z}, \widehat{y}]) = 1/2$, we have $|y_i^* - \widehat{y}| \leq |z_i - \widehat{z}|\Delta/\delta \leq 2\eta \Delta/\delta$. We also have $|y_i - y_i^*| \leq \varepsilon$, so $|y_i - \widehat{y}| \leq 2\eta\Delta/\delta + \varepsilon$. Cutting at $z_i$ and $y_i$ then gets Alice close to her Stackelberg value, so cutting at $z_{\ell}$ and $z_{\ell}$ is at least that good. More precisely:
        \begin{align}
            u_A^*(2) &\leq \left(2\eta + \frac{2\eta\Delta}{\delta} + \varepsilon\right)\Delta + \max \bigl(V_A([z_i, y_i]), V_A([0, z_i] \cup [y_i, 1])\bigr) \\
            &\leq \left(2\eta + \frac{2\eta\Delta}{\delta} + \varepsilon\right)\Delta + \max \bigl(V_A([z_{\ell}, y_{\ell}]), V_A([0, z_{\ell}] \cup [y_{\ell}, 1])\bigr)
        \end{align}
        The shifting of $2\eta$ in Bob's favor only adds another $O(\eta)$ loss to Alice, so Alice's loss compared to her Stackelberg value over the $T$ rounds is $O(T \eta)$.
        \item By the construction of $y_{\ell}$, we have $|V_B([z_{\ell}, y_{\ell}])-1/2| \leq \Delta \varepsilon$. Therefore, shifting the cut points $(z_{\ell}, y_{\ell})$ by $2\eta$ to $(\overline{z}, \overline{y})$ to make one piece bigger ensures he prefers that piece by at least $2\eta \delta - \Delta \varepsilon = \eta \delta$. He can then only afford to take the piece Alice likes less $O(f(T) / (\eta\delta))$ times. This term is dominated by the loss from (i), so we can ignore it.
    \end{enumerate}
    Adding all these together, we get a regret bound of:
        $O\left(\frac{f(T)}{\eta \varepsilon} \log(1/\varepsilon) \log T + T\eta\right)$.
    Our choices of $\varepsilon = (f(T))^{1/3}T^{-1/3}\log^{2/3}T$ and $\eta = \varepsilon \Delta/\delta$ now yield the regret bound:
    \begin{align}
        O\left((f(T))^{1/3}T^{2/3} \left(\log\left(\frac{T}{f(T) \log^2 T}\right)\log^{-1/3}T + \log^{2/3} T\right)\right)\,.
    \end{align}
    But since $f(T)$ is at least inverse polynomial in $T$, both terms in the sum are $O(\log^{2/3} T)$, yielding the desired regret bound.
\end{proof}

\twoCutKnownAlphaUpper*

\begin{proof}
    Invoking Lemma \ref{lem:2-cut-non-myopic-general-upper-bound} with $f(T) = T^{\alpha}$ immediately gives the desired regret bound.
\end{proof}

\twoCutUnknownAlphaUpper*

\begin{proof}
    Invoking Lemma \ref{lem:2-cut-non-myopic-general-upper-bound} with $f(T) = T/\log^5 T$ immediately gives the desired regret bound.
\end{proof}

\begin{restatable}{lemma}{kCutNonMyopicGeneralUpper}
    \label{lem:k-cut-non-myopic-general-upper-bound}
    Let $k \geq 3$. Let $f : \mathbb{N} \to \mathbb{R}$ be a function known to both players such that $f(x) \in o(x/\log^2 x)$ and $f(x) \in \Omega(x^{\alpha})$ for some $\alpha > -1$. Suppose Bob's strategy guarantees his regret is at most $c \cdot f(T)$, where $c > 0$ is a constant private to Bob. Then Alice has a deterministic strategy which guarantees her regret is $O((f(T))^{1/4} T^{3/4} k^{3/4} \log^{1/2} T)$.
\end{restatable}

\begin{proof}
    Alice's strategy uses the same idea as Proposition \ref{prop:k-cut-myopic-upper-bound} but with Lemma \ref{lem:binary-search-with-regret} using $g(T)=\log T$ instead of Lemma \ref{lem:binary-search}.
    Specifically, given parameters $\eta = (f(T))^{1/4} T^{-1/4} k^{-1/4}(\log T)^{1/2}$ and $\varepsilon = \delta^2 \eta^2 / 2$, Alice's strategy is as follows.

    \begin{enumerate}
        \item Use Lemma \ref{lem:binary-search-with-regret} to get an approximation $\widetilde{m}_B$ of Bob's midpoint $m_B$ to within $\varepsilon$. Assuming $T$ is sufficiently large, we can use Lemma \ref{lem:binary-search-with-regret} with these parameters since $\varepsilon = \delta^2 \sqrt{f(T)/(Tk)} (\log T)/2 <f(T)$ for $f(T) \in \Omega(T^{\alpha})$ for some $\alpha > -1$.
        \item Set $x_0 = \widetilde{m}_B - \eta$ and $x_1 = \widetilde{m}_B$. Set $i=1$. Repeat the following until the loop exits:
        \begin{enumerate}
            \item Use Lemma \ref{lem:binary-search-with-regret} to find the $z$ for which $V_B((0, \widetilde{m}_B - \eta) \cup (x_i, z))=1/2$, to within $\varepsilon$. If the search fails, exit the loop.
            \item Set $x_{i+1} = z$ and increment $i$.
        \end{enumerate}
        \item Reset $i=0$. Repeat the following until the loop exits:
        \begin{enumerate}
            \item Use Lemma \ref{lem:binary-search-with-regret} to find the $z$ for which $V_B((z, x_i) \cup (x_2, 1))=1/2$, to within $\varepsilon$. This can be accomplished by treating the cake as reflected left-right. If the search fails, exit the loop.
            \item Set $x_{i-1} = z$ and decrement $i$.
        \end{enumerate}
        \item Let $-\ell$ be the index of the last negative-index $x$ and let $r$ be the index of the last positive-index $x$. Let $a^*_1, \ldots, a^*_k$ be cut points chosen from the set $[0, x_{-\ell}] \cup \{x_{-\ell}, \ldots, x_r\} \cup [x_r, 1]$ in order to maximize Alice's value of the piece containing fewer of the intervals $(x_i, x_{i+1})$ for $i \in \{-\ell, \ldots, r-1\}$.
        \item Let $a_1, \ldots, a_k$ be cut points that deviate from $a_1^*, \ldots, a_k^*$ by exactly $3\Delta\varepsilon/(\delta^2 \eta) + 2\Delta \eta/\delta$ distance in ways that make Alice's less-preferred piece bigger. Cut there for the remaining rounds.
    \end{enumerate}
    The analysis is almost identical to Proposition \ref{prop:k-cut-myopic-upper-bound}. We need to bound how similarly Bob values the intervals $(x_i, x_{i+1})$. Our reference point will be $(x_0, x_1) = (\widetilde{m}_B - \eta, \widetilde{m}_B)$. Each iteration of the loop in step 2 aims to find a new interval $(x_i, x_{i+1})$ of similar value. To see how well it does so, fix $i \in \{1, \ldots, r - 1\}$ and let $z^*$ be the exact solution Lemma \ref{lem:binary-search-with-regret} aimed to find, i.e. such that $V_B((0, x_0) \cup (x_i, z^*))=1/2$. We then have $|x_{i+1} - z^*| \leq \varepsilon$. By the choice of $x_0$ and $x_1$, we also have:
    \begin{align}
        \frac{1}{2} &= V_B((0, x_0) \cup (x_i, z^*)) = V_B((0, x_1)) - V_B((x_0, x_1)) + V_B((x_i, z^*)) \\
        &= V_B((0, \widetilde{m}_B)) + \left(V_B((x_i, z^*)) - V_B((x_0, x_1))\right)
    \end{align}
    By the construction of $\widetilde{m}_B$ in step 1, we also know that $|m_B - \widetilde{m}_B| \leq \varepsilon$, so $|V_B((0, \widetilde{m}_B)) - 1/2| \leq \Delta \varepsilon$. By the triangle inequality, we therefore have:
    \begin{align}
        \left|V_B((x_i, x_{i+1})) - V_B((x_0, x_1))\right| \leq& \left|V_B((x_i, x_{i+1})) - V_B((x_i, z^*))\right| + \left|V_B((x_i, z^*)) - V_B((x_0, x_1))\right| \notag \\
        \leq& \Delta \varepsilon + \Delta \varepsilon = 2\Delta \varepsilon \,.
    \end{align}
    We also know that the search failed when trying to find $x_{r+1}$, so $V_B((0, x_0) \cup (x_r, 1-\varepsilon)) \leq 1/2$. Therefore:
    \begin{align}
        V_B((x_r, 1)) &\leq \frac{1}{2} - V_B((0, x_0)) + V_B((1-\varepsilon, 1)) = \frac{1}{2} - V_B((0, \widetilde{m}_B)) + V_B((x_0, x_1)) + V_B((1-\varepsilon, 1)) \notag \\
        &\leq V_B((x_0, x_1)) + 2\Delta \varepsilon \,. \label{eq:vb-xr-1-less-than-vb-x0-x1-plus-2-delta-epsilon}
    \end{align}
    Moving to to the negative $x$'s, fix $i \in \{-\ell+1, \ldots, 0\}$ and let $z^*$ be the exact solution Lemma \ref{lem:binary-search} aimed to find, i.e. such that $V_B((z^*, x_i) \cup (x_2, 1)) = 1/2$. We then have $|x_{i-1} - z^*| \leq \varepsilon$. By the choice of $x_1$ and $x_2$, we also have:
    \begin{align}
        \frac{1}{2} &= V_B((z^*, x_i) \cup (x_2, 1)) = V_B((x_1, 1)) - V_B((x_1, x_2)) + V_B((z^*, x_i)) \notag  \\
        &= V_B((0, \widetilde{m}_B)) + \left(V_B((z^*, x_i)) - V_B((x_1, x_2))\right) \,.
    \end{align}
    We have $|V_B((\widetilde{m}_B, 1)) - 1/2| \leq \Delta \varepsilon$ by the construction of $\widetilde{m}_B$ and $|V_B((x_0, x_1)) - V_B((x_1, x_2))| \leq 2\Delta \varepsilon$ by the analysis of the positive $x$'s, so by the triangle inequality we have:
    \begin{align}
        \left|V_B((x_{i-1}, x_i)) - V_B((x_0, x_1))\right| &\leq \Delta \varepsilon + 2\Delta \varepsilon = 3\Delta \varepsilon
    \end{align}
    Similarly to $(x_r, 1)$, since the search failed when trying to find $x_{-\ell-1}$, we have $V_B((\varepsilon, x_{-\ell}) \cup (x_2, 1)) \leq 1/2$. Therefore:
    \begin{align}
        V_B((0, x_{-\ell})) &\leq \frac{1}{2} - V_B((x_2, 1)) + V_B((0, \varepsilon)) = \frac{1}{2} - V_B((\widetilde{m}_B, 1)) + V_B((x_1, x_2)) + V_B((0, \varepsilon)) \\
        &\leq V_B((x_0, x_1)) + 4\Delta \varepsilon \label{eq:vb-0-xminusl-less-than-4-delta-epsilon}
    \end{align}
    Combining \eqref{eq:vb-xr-1-less-than-vb-x0-x1-plus-2-delta-epsilon} and \eqref{eq:vb-0-xminusl-less-than-4-delta-epsilon}, the cake uncovered by any of the intervals has value at most $2V_B((x_0, x_1)) + 6\Delta \varepsilon$ to Bob. Invoking Lemma \ref{lem:known-bob-intervals-gives-stackelberg-cut} with interval values of $V_B((x_0, x_1)) \pm 3\Delta\varepsilon$ and uncovered cake of value at most $2(V_B((x_0, x_1)) + 3\Delta\varepsilon)$, we get that there exists a $(k+4)(V_B((x_0, x_1))+3\Delta\varepsilon)\Delta/\delta$-Stackelberg cut among those that $a_1^*, \ldots, a_k^*$ was chosen from. Therefore, the cut $a_1^*, \ldots, a_k^*$ is at least that good. Factoring in the extra transfer to Bob in step 5 and using $V_B((x_0, x_1)) \leq \Delta \eta$, we get that the cut $a_1, \ldots, a_k$ is $O(k \eta + \varepsilon/\eta)$-Stackelberg if Bob prefers the piece Alice expects him to.

    We can now analyze the three sources of Alice's regret: (i) loss from $a_1, \ldots, a_k$ not quite achieving her Stackelberg value, (ii) loss from Bob taking the ``wrong" piece once Alice starts cutting there, and (iii) the rounds needed for steps 1-4.

    \begin{itemize}
        \item[(i)] The cut $a_1, \ldots, a_k$ is $O(k\eta + \varepsilon/\eta)$-Stackelberg. Since $\varepsilon = O(\eta^2)$, this reduces to $O(k\eta)$ regret per round for a total of $O(Tk\eta)$.
        \item[(ii)] There are two possible issues here: Alice could be wrong about which piece Bob prefers, and Bob could use his regret bound to take the other piece anyway. Let $S$ be the piece Alice expects Bob to prefer and let $n = r+\ell$ be the number of sub-intervals. Using the upper bounds for Bob's value of each part of the cake, we have $(n+2)(V_B((x_0, x_1))+3\Delta\varepsilon) \geq 1$, so $n \geq 1/(V_B((x_0, x_1)) + 3\Delta\varepsilon) - 2$. Alice's expectation is based on $S$ containing at least $n/2$ sub-intervals, each of which has value at least $V_B((x_0, x_1)) - 3\Delta\varepsilon$ to Bob. Then:
        \begin{align}
            \frac{1}{2} - V_B(S) &\leq \frac{1}{2} - \frac{n}{2}(V_B((x_0, x_1)) - 3\Delta\varepsilon) \\
            &\leq \frac{1}{2} - \frac{1}{2} \cdot \left(\frac{V_B((x_0, x_1)) - 3\Delta\varepsilon}{V_B((x_0, x_1)) + 3\Delta\varepsilon} + 2(V_B((x_0, x_1)) - 3\Delta\varepsilon)\right) \\
            &= \frac{3\Delta\varepsilon}{V_B((x_0, x_1)) + 3\Delta\varepsilon} + V_B((x_0, x_1)) - 3\Delta\varepsilon \\
            &\leq \frac{3\Delta\varepsilon}{\delta \eta} + \Delta \eta \explain{Since $V_B((x_0, x_1)) \geq \delta(x_1-x_0)=\delta \eta$}
        \end{align}
        Therefore, the adjustment in step 5 gives Bob enough extra cake to guarantee he not only prefers $S$ to $\overline{S}$, but prefers it by at least $\Delta \eta$. He can therefore only take the wrong piece $O(f(T)/\eta)$ times.
        \item[(iii)] Alice uses Lemma \ref{lem:binary-search-with-regret} $O(n)$ times to find her cut points. Each use of Lemma \ref{lem:binary-search-with-regret} takes at most $O(f(T) \log(T) \log(1/\varepsilon)/\varepsilon)$ rounds. We can upper-bound $n$ by using the fact that $n$ intervals of value at least $V_B((x_0, x_1)) - 3\Delta\varepsilon$ fit into $[0, 1]$, so:
        \begin{align}
            n &\leq \frac{1}{V_B((x_0, x_1)) - 3\Delta\varepsilon} \leq \frac{1}{\delta \eta - 3\Delta\varepsilon} \in O(1/\eta) \,.
        \end{align}
        So it took Alice $O(f(T) \log(T) \log(1/\varepsilon)/(\eta \varepsilon))$ rounds to find her cut points.
    \end{itemize}
    Adding all three sources together, we get a regret bound of:
    \begin{align}
        O\left(Tk\eta + \frac{f(T)}{\eta} + \frac{f(T) \log(T) \log(1/\varepsilon)}{\eta \varepsilon}\right) \,.
    \end{align}
    Because $\varepsilon = O(\log(T) \sqrt{f(T)/(Tk)})$ and $f(T) \in o(T/\log^2 T)$, we have $\varepsilon \in o(1)$, so the term from (iii) dominates the term from (ii). Using the choice of $\eta = f(T)^{1/4} T^{-1/4} k^{-1/4} \log^{1/2} T$ gives a regret bound of:
    \begin{align}
        O\left(f(T)^{1/4} T^{3/4} k^{3/4} \log^{1/2} T + f(T)^{1/4} T^{3/4} \log\left(\frac{Tk^2}{f(T) \log^2 T}\right) k^{3/4} \log^{-1/2} T\right) \,.
    \end{align}
    Since $f(T)$ is at least inverse polynomial in $T$, the second term is dominated by the first. We therefore get the desired regret bound of $O(f(T)^{1/4} T^{3/4} k^{3/4} \log^{1/2} T)$.
\end{proof}

\kCutKnownAlphaUpper*

\begin{proof}
    Invoking Lemma \ref{lem:k-cut-non-myopic-general-upper-bound} with $f(T) = T^{\alpha}$ immediately gives the desired regret bound.
\end{proof}

\kCutNonMyopicUpperBoundUnknownAlpha*

\begin{proof}
    Invoking Lemma \ref{lem:k-cut-non-myopic-general-upper-bound} with $f(T) = T/\log^6 T$ immediately gives the desired regret bound.
\end{proof}

\subsection{Lower Bounds for Non-Myopic Bob with Public Learning Rate} \label{sec:known-alpha-lower-bounds}

All of the lower bounds in this section are based on the same construction for Bob's valuation. The idea is to hide a narrow ``spike" of nonuniform density in a larger region of uniform density. Because the spike's average density is the same as the surrounding region, Alice must cut inside the spike to find it. On the rest of the cake, the density alternates between high and low values in a way that forces Alice to ``waste" $k-1$ of her cut points if she wants a near-Stackelberg payoff.

The spiked valuation functions are defined in Definition \ref{def:generalized-spike}. Estimates for the lengths and values of various parts of the construction are in Lemma \ref{lem:generalized-spike-low-value} and are used throughout. Lemma \ref{lem:generalized-spiked-stackelberg} shows that Alice must find the spike to get close to her Stackelberg value, while Lemma \ref{lem:generalized-spiked-similarity} shows that the existence of a spike does not affect Bob's valuations very much. Lemma \ref{lem:alternating-region-too-few-cuts} shows that any partition that does not ``waste" $k-1$ of its cut points has a certain combinatorial property and Lemma \ref{lem:generalized-spiked-bad-cut-suboptimal} shows that any partition with that property cannot get Alice close to her Stackelberg value. Lemma \ref{lem:generalized-spiked-distinguishers} puts the final nail in the coffin by showing that even the partitions that do ``waste" $k-1$ of their cut points give Alice a lot of regret if she incorrectly guesses where the spike is.

\medskip 

We now construct Bob's valuation. The cake will be divided into alternating ``high-value" and ``low-value" regions in such a way that the high-value regions make up just over half the cake by Bob's valuation. The spike will be embedded in the first high-value region.

\begin{definition}
    \label{def:generalized-spike}
    Let $k \geq 2$. We partition the cake $[0, 1]$ into $k+1$ disjoint intervals that alternate between \emph{high-value} and \emph{low-value} types, starting with a high-value interval. Let the length of each of the $n_H^k=\lfloor k/2 \rfloor + 1$ high-value intervals be $\ell^k=\frac{1}{3(\lfloor k/2 \rfloor + 1/2)}$ and distribute the remaining length equally among the $n_L^k = \lceil k/2 \rceil$ low-value intervals. 
    
    Let $\mathcal{H}^k$ be the union of the high-value intervals and $\mathcal{L}^k$ be the union of the low-value intervals.
    Then the \emph{unspiked} value density function $\sigma_0^k$ is defined as:
    \begin{align}
        \sigma_0^k(x) &= \begin{cases}
            \frac{3}{2} & \text{if $x \in \mathcal{H}^k$} \\
            \frac{1-\frac{3}{2}\mu(\mathcal{H}^k)}{\mu(\mathcal{L}^k)} & \text{if $x \in \mathcal{L}^k$}
        \end{cases}
    \end{align}
    For $z \in [5/6, 11/12]$ and $w \in (0, 1/48]$, the \emph{spiked} value density function $\sigma_{w;z}^k$ is defined as:
    \begin{align}
        \sigma_{w;z}^k(x) &= \begin{cases}
            2 & \text{if $x/\ell^k \in (z-w, z]$} \\
            1 & \text{if $x/\ell^k \in (z, z+w)$} \\
            \sigma_0^k(x) & \text{otherwise}
        \end{cases}
    \end{align}
\end{definition}

We will typically not use the exact measurements of $\sigma_0^k$. The bounds of the following lemma often suffice.

\begin{restatable}{lemma}{generalizedSpikeLowValue}
    \label{lem:generalized-spike-low-value}
    For $k \geq 2$, the unspiked valuation function $\sigma_0^k$ has the following properties:
    \begin{itemize}
        \item The total value of the high-value intervals satisfies $\int_{\mathcal{H}^k} \sigma_0^k(x) \,\mathrm{d}x \in \left(\frac{1}{2}, \frac{2}{3}\right]$.
        \item For all $x \in \mathcal{L}^k$, we have $\sigma_0^k(x) \in \left(\frac{1}{2}, \frac{3}{4}\right)$.
        \item The length of each low-value interval is at least $\frac{5}{6k}$.
        \item The length $\ell^k$ of each high-value interval satisfies $\ell^k \in \left[\frac{4}{9k}, \frac{2}{3k}\right]$.
    \end{itemize}
\end{restatable}

\begin{proof}
    For the first part, since $k \geq 2$ the total length of the high-value intervals is:
    \begin{align}
        \mu(\mathcal{H}^k) = n_H^k \ell^k = \frac{1}{3} \cdot \frac{\left\lfloor \frac{k}{2} \right\rfloor + 1}{\left\lfloor \frac{k}{2} \right\rfloor + \frac{1}{2}} \in \left(\frac{1}{3}, \frac{4}{9}\right] \label{eq:high-value-region-length-exact-expression}
    \end{align} 
    So because $\sigma_0^k(x) = 3/2$ for $x \in \mathcal{H}^k$, we have:
    \begin{align}
        \int_{\mathcal{H}^k} \sigma_0^k(x) \,\mathrm{d}x = \frac{3}{2} \mu(\mathcal{H}^k) \in \left(\frac{1}{2}, \frac{2}{3}\right]\,.
    \end{align}
    This proves the first part of the lemma. The $n_L^k = \lceil k/2 \rceil$ low-value intervals share the remaining value, which then lies in $\left[\frac{1}{3}, \frac{1}{2}\right)$, equally. Since $k \geq 2$, their total length is:
    \begin{align}
        \mu(\mathcal{L}^k) = 1-n_H^k \ell^k &= 1 - \frac{1}{3} \cdot \frac{\left\lfloor \frac{k}{2} \right\rfloor + 1}{\left\lfloor \frac{k}{2} \right\rfloor + \frac{1}{2}} \in \left[\frac{5}{9}, \frac{2}{3}\right) \label{eq:low-region-length-in-interval-5/9-2/3}
    \end{align}
    Let $y \in \mathcal{L}^k$. Since $\mu(\mathcal{H}^k) \leq 4/9$ and $\mu(\mathcal{L}^k) < 2/3$, we have:
    \begin{align}
        \sigma_0^k(y) &= \frac{1-\frac{3}{2}\mu(\mathcal{H}^k)}{\mu(\mathcal{L}^k)} > \frac{1}{2}
    \end{align}
    To get the upper bound, we compute $\sigma_0^k(y)$ explicitly using the exact expressions from \eqref{eq:high-value-region-length-exact-expression} and \eqref{eq:low-region-length-in-interval-5/9-2/3}:
    \begin{align}
        \sigma_0^k(y) = \frac{1-\frac{3}{2}\mu(\mathcal{H}^k)}{\mu(\mathcal{L}^k)} = 1 - \frac{\frac{1}{6}}{\frac{\left\lfloor \frac{k}{2} \right\rfloor + \frac{1}{2}}{\left\lfloor \frac{k}{2} \right\rfloor + 1}-\frac{1}{3}} < \frac{3}{4}
    \end{align}
    Which completes the proof of the second part.
    To prove the third part, divide the lower bound on $\mu(\mathcal{L}^k)$ \eqref{eq:low-region-length-in-interval-5/9-2/3} by $n_L^k$:
    \begin{align}
        \frac{5/9}{\lceil k/2 \rceil} &\geq \frac{5/9}{\frac{2}{3}k} \explain{$\lceil k/2 \rceil \leq 2k/3$ since $k \geq 2$} \\
        &= \frac{5}{6k},
    \end{align}
    which completes the proof of the third part.

    For the high-value intervals, we can bound the exact value of $\ell^k$ from below:
    \begin{align}
        \frac{1}{3 \cdot \left(\left\lfloor \frac{k}{2} \right\rfloor + \frac{1}{2} \right)} &\geq \frac{1}{3 \cdot \frac{3}{4}k} \explain{Since $(\lfloor k/2 \rfloor + 1/2)/k$ is maximized at $k=2$} \\
        &= \frac{4}{9k} \,.
    \end{align}
    We can also bound it  from above as follows:
    \begin{align}
        \frac{1}{3 \cdot \left(\left\lfloor \frac{k}{2} \right\rfloor + \frac{1}{2} \right)} &\leq \frac{1}{3 \cdot \frac{1}{2}k} \explain{Since $(\lfloor k/2 \rfloor + 1/2)/k$ is minimized at $k=3$} \\
        &= \frac{2}{3k} \,.
    \end{align}
    This completes the proof.
\end{proof}

\begin{restatable}{lemma}{generalizedSpikedStackelberg}
    \label{lem:generalized-spiked-stackelberg}
    Suppose Alice's value density is  uniform. Her Stackelberg value in the $k$-cut game for $k \geq 2$ against a Bob with value density $\sigma_0^k$ is $\frac{2}{3}$, while against a Bob with valuation $\sigma_{w;z}^k$ for every $z \in [5/6, 11/12]$ and $w \in (0, 1/48]$ it is at least $\frac{2}{3} + \frac{4}{27k}w$.
\end{restatable}

\begin{proof}
    First consider $\sigma_0^k$. The ideal cut would be one where Bob receives only cake with his maximum value density. By the second part of Lemma \ref{lem:generalized-spike-low-value}, such a cut will give Bob only high-value intervals.
    Since there are $k+1$ total intervals, it is possible to make a piece consisting exactly of $\mathcal{H}^k$ with $k$ cuts. By the first part of Lemma \ref{lem:generalized-spike-low-value}, a Bob with value density $\sigma_0^k$ would be satisfied with such a piece. Moving the cut points inwards until the piece has value exactly $1/2$ to Bob results in a piece with measure $(1/2)/(3/2) = 1/3$, so under such a partition Alice receives the other $2/3$ as her Stackelberg value.

    Now consider $\sigma_{w;z}^k$. Here, Alice can start to improve her payoff by putting the first cut at a $z$ proportion through the first high-value interval and the other $k-1$ cuts at the last $k-1$ interval boundaries. Denote the resulting piece, consisting of all of $\mathcal{H}^k$ except for a $1-z$ proportion of the first high-value interval, by $S$. Since $z \geq 5/6$, Bob's value of $S$ is then at least:
    \begin{align}
        \int_S \sigma_{w;z}^k(x) \,\mathrm{d}x \geq \frac{3}{2}\left(z + \left\lfloor \frac{k}{2}\right\rfloor\right)\ell^k = \frac{3}{2} \cdot \frac{z + \left\lfloor \frac{k}{2}\right\rfloor}{3 \cdot \left(\left\lfloor \frac{k}{2}\right\rfloor + \frac{1}{2}\right)} &\geq \frac{1}{2} \cdot \frac{\left\lfloor \frac{k}{2}\right\rfloor + \frac{5}{6}}{\left\lfloor \frac{k}{2} \right\rfloor + \frac{1}{2}} > \frac{1}{2}
    \end{align}
    So again by moving the last $k-1$ cut points slightly inwards, Alice can give Bob exactly $1/2$ in value. The difference between this partition and the partition against $\sigma_0^k$ is that Bob receives $w\ell^k$ cake with density $2$. Compared to making up that much value with density-$3/2$ cake, Alice saves an extra $\frac{2-3/2}{3/2}w\ell^k = w\ell^k/3$ cake for herself. Therefore, her Stackelberg value is at least:
    \begin{align}
        u_A^*(k) &\geq \frac{2}{3} + \frac{w\ell^k}{3} \\
        &\geq \frac{2}{3} + \frac{w}{3} \cdot \frac{4}{9k} \explain{By the fourth part of Lemma \ref{lem:generalized-spike-low-value}} \\
        &= \frac{2}{3} + \frac{4}{27k}w,
    \end{align}
    which completes the proof.
\end{proof}

\begin{restatable}{lemma}{generalizedSpikedSimilarity}
    \label{lem:generalized-spiked-similarity}
    Let $S$ be an arbitrary measurable subset of $[0, 1]$. Let $z \in [5/6, 11/12]$ and $w \in (0, 1/48]$. Then:
    \begin{align*}
        \left|\int_S \sigma_0^k(x) \,\mathrm{d}x - \int_S \sigma_{w;z}^k(x) \,\mathrm{d}x \right| &\leq \frac{2}{3k}w
    \end{align*}
\end{restatable}

\begin{proof}
    The only difference between $\sigma_0^k$ and $\sigma_{w;z}^k$ is a small interval in the first high-value region. Specifically, we have $\sigma_0(x) \neq \sigma_w^z(x)$ only for $x \in ((z-w)\ell^k, (z+w)\ell^k)$. Accordingly:
    \begin{align}
        \left|\int_S \sigma_0^k(x) \,\mathrm{d}x - \int_S \sigma_{w;z}^k(x) \,\mathrm{d}x \right| &= \left| \int_{((z-w)\ell^k, (z+w)\ell^k) \cap S} \sigma_0^k(x) - \sigma_{w;z}^k(x) \,\mathrm{d}x \right| \\
        &\leq  \left| \int_{((z-w)\ell^k, z\ell^k) \cap S} -\frac{1}{2}\,\mathrm{d}x \right| + \left| \int_{(z\ell^k, (z+w)\ell^k) \cap S} \frac{1}{2} \,\mathrm{d}x\right| \\
        &\leq \int_{(z-w)\ell^k}^{z\ell^k} \frac{1}{2} \,\mathrm{d}x + \int_{z\ell^k}^{(z+w)\ell^k} \frac{1}{2} \,\mathrm{d}x \\
        &= w\ell^k
    \end{align}
    And by Lemma \ref{lem:generalized-spike-low-value} we have $\ell^k \leq 2/(3k)$, which completes the proof.
\end{proof}

\begin{restatable}{lemma}{alternating-region-too-few-cuts}
    \label{lem:alternating-region-too-few-cuts}
    Suppose an interval $[a, b]$ is partitioned into $n \geq 2$ connected regions, alternately colored red and blue. If $[a, b]$ is split by an alternating partition with fewer than $n-1$ cuts, then at least one of the following holds:
    \begin{itemize}
        \item One piece will contain an entire red region and an entire blue region, or
        \item Both pieces will contain an entire region of the same color.
    \end{itemize}
\end{restatable}

\begin{proof}
    Suppose for contradiction that the statement were false and let $m$ be the smallest number of cuts for which a counterexample exists. Let $n > m+1$ be the minimal number of regions among all $m$-cut counterexamples, and fix a counterexample with $m$ cuts and $n$ regions.

    We first claim that this counterexample cannot put a cut within the first or last region. Without loss of generality, we focus on the first region. If there is at least one cut in the first region, there are at most $m-1$ cuts in all the other regions. But then those $n-1$ regions have at most $m-1 < n-2$ cuts among them, so by the minimality of $m$ they must satisfy the lemma statement. Adding the first region back into that division does not change that, contradicting the choice of $m$.

    Given that there are no cuts in the first or last regions, we claim that we cannot have $m = n-2$. If it were, then the parity is wrong:
    \begin{itemize}
        \item If $n$ is even, then its first and last regions are different colors. But with an even number of cuts, they end up in the same piece.
        \item If $n$ is odd, then its first and last regions are the same color. But with an odd number of cuts, they end up in different pieces.
    \end{itemize}
    So we have $m < n-2$.

    We can also rule out some of the smallest values for $n$ and $m$. If $m=0$, then because $n \geq 2$ we would trivially have a red and blue region in the same piece. Therefore $m \geq 1$, so paired with $m < n-2$ we have $n \geq 4$.

    Now consider the leftmost cut $x$. As seen earlier, we cannot have $x$ in the first region. We also cannot have $x$ past the second region, as then the first two regions would be in the same piece. So $x$ must be (weakly) in the second region. But then there are at most $m-1 < n-3$ cuts in the last $n-2$ regions, so the last $n-2$ regions must be divided according to the lemma statement. Adding the first two regions back into the division does not change that, contradicting the choice of $m$.

    As every case leads to a contradiction, there must not be such an $m$, completing the proof.
\end{proof}

\begin{restatable}{lemma}{generalizedSpikedBadCutSuboptimal}
    \label{lem:generalized-spiked-bad-cut-suboptimal}
    Let $\eps > 0$. Let $S \subseteq [0, 1]$ be measurable such that $|\int_S \sigma_0^k(x) \,\mathrm{d}x - 1/2| \leq \eps$. Then:
    \begin{itemize}
        \item If $S$ contains an entire low-value interval, then $\mu(S) > \frac{1}{3} + \frac{5}{12k} - \frac{2}{3}\eps$.
        \item If $S$ contains at least $\frac{19}{24}$ of a high-value interval, then $\mu(S) \leq \frac{2}{3} - \frac{7}{36k} + \frac{4}{3}\eps$.
    \end{itemize}
\end{restatable}

\begin{proof}
    We first suppose that $S$ contains an entire low-value interval. Breaking up $\int_S \sigma_0(x) \,\mathrm{d}x$ into $S \cap \mathcal{L}^k$ and $S \cap \mathcal{H}^k$:
    \begin{align}
        \frac{1}{2} - \eps &\leq \int_S \sigma_0^k(x) \,\mathrm{d}x = \int_{S \cap \mathcal{L}^k} \sigma_0^k(x) \,\mathrm{d}x + \int_{S \cap \mathcal{H}^k} \sigma_0^k(x) \,\mathrm{d}x \\
        &< \frac{3}{4}\mu(S \cap \mathcal{L}^k) + \frac{3}{2} (\mu(S) - \mu(S \cap \mathcal{L}^k)) \explain{By the second part of Lemma \ref{lem:generalized-spike-low-value}}\\
        &= \frac{3}{2} \mu(S) - \frac{3}{4} \mu(S \cap \mathcal{L}^k) \\
        &\leq \frac{3}{2} \mu(S) - \frac{5}{8k} \explain{By the third part of Lemma \ref{lem:generalized-spike-low-value}, since $S$ contains a low-value region}
    \end{align}
    Solving for $\mu(S)$ gives the first part of the statement.

    On the other hand, suppose $S$ contains at least $19/24$ of a high-value interval. In order to maximize $\mu(S)$ while ensuring $\int_S \sigma_0^k(x) \,\mathrm{d}x \leq 1/2 + \eps$, the rest of $S$ should be taken from $\mathcal{L}^k$.
    
    We claim that containing all of $\mathcal{L}^k$ would push $\int_S \sigma_0(x) \,\mathrm{d}x$ too high. Let $X \subset \mathcal{H}$ such that $\mu(X) = (n_H^k - 1/2)\ell^k$. Then:
    \begin{align}
        \int_X \sigma_0^k(x) \,\mathrm{d}x = \frac{3}{2} \cdot \left(\left\lfloor \frac{k}{2}\right\rfloor + 1 - \frac{1}{2}\right) \cdot \frac{1}{3 \cdot \left(\left\lfloor \frac{k}{2} \right\rfloor + \frac{1}{2} \right)} = \frac{1}{2}
    \end{align}
    Then its complement $\overline{X}$, consisting of all of $\mathcal{L}^k$ plus $\ell^k/2$ measure of $\mathcal{H}^k$, satisfies $\int_{\overline{X}} \sigma_0(x) \,\mathrm{d}x=1/2$. Let $Y \subset [0, 1]$ consist of all of $\overline{X}$ and an additional $7\ell^k/24$ measure in $\mathcal{H}^k$, for a total of $19\ell^k/24$ measure of $\mathcal{H}^k$. Then:
    \begin{align}
        \int_Y \sigma_0^k(x) \,\mathrm{d}x &= \int_{\overline{X}} \sigma_0^k(x) \,\mathrm{d}x + \frac{7\ell}{24} \cdot \frac{3}{2} = \frac{1}{2} + \frac{7}{16}\ell
    \end{align}
    So compared to $Y$, the set $S$ must have $\frac{7}{16}\ell^k-\eps$ less value in $\mathcal{L}^k$. By the second part of Lemma \ref{lem:generalized-spike-low-value}, that represents at least $\frac{4}{3}\left(\frac{7}{16}\ell^k-\eps\right) = \frac{7}{12}\ell^k - \frac{4}{3}\eps$ in length, so we have:
    \begin{align}
        \mu(S) &\leq \mu(Y) - \frac{7}{12}\ell^k + \frac{4}{3}\eps \\
        &= \left(1-\left(\left\lfloor \frac{k}{2} \right\rfloor + 1\right)\ell^k\right) + \frac{19}{24}\ell^k - \frac{7}{12}\ell^k + \frac{4}{3}\eps \explain{Splitting $Y$ into low-value and high-value} \\
        &= 1 - \left(\left\lfloor \frac{k}{2}\right\rfloor + \frac{1}{2}\right)\ell^k - \frac{7\ell^k}{24} + \frac{4}{3}\varepsilon \label{eq:mu-S-less-than-term-with-unsimplified-ell}
    \end{align}
    Substituting in the exact value of $\ell^k=1/(3(\lfloor k/2\rfloor + 1/2))$ simplifies the second term of \eqref{eq:mu-S-less-than-term-with-unsimplified-ell}:
    \begin{align}
        \mu(S) &\leq 1 - \frac{1}{3} - \frac{7\ell^k}{24} + \frac{4}{3}\eps \\
        &\leq \frac{2}{3} - \frac{7}{36k} + \frac{4}{3}\eps, \explain{By the fourth part of Lemma \ref{lem:generalized-spike-low-value}}
    \end{align}
    which completes the proof.
\end{proof}

\begin{definition}[Distinguishing partitions]
    Let $z \in [5/6, 11/12]$ and $w \in (0, 1/48]$. We say a partition \emph{distinguishes $z$} if a Bob with value density $\sigma_0^k$ and a Bob with value density $\sigma_{w;z}^k$ prefer different pieces.
\end{definition}

\begin{restatable}{lemma}{generalizedSpikedDistinguishers}
    \label{lem:generalized-spiked-distinguishers}
    Let $z \in [5/6, 11/12]$ and $w \in (0, 1/48]$. Suppose Alice's value density is the uniform $v_A(x)=1$. If a $k$-cut distinguishes $z$, then Alice's payoff is at most $\frac{2}{3}-\frac{1}{6k}$ if Bob takes the piece preferred by $\sigma_0^k$.
\end{restatable}

\begin{proof}
    We will call the piece preferred by $\sigma_{w;z}^k$ ``Bob's piece" $B$ and the other one ``Alice's piece" $A$. At a high level, we will show that the partition must put nearly all of $\mathcal{H}^k$ in $B$ and nearly all of $\mathcal{L}^k$ in $A$. Doing so requires Alice to use $k-1$ of her cuts separating those regions, leaving only one cut to find $z$. The fact that the density-$2$ part of the spike is left of the density-$1$ part will then ensure $B$ has little value to Alice. A Bob with valuation $\sigma_0^k$ would take $A$, completing the proof.

    Any distinguishing partition must put at least one cut point in the spike, as Bob's average value across the spike is the same between $\sigma_0^k$ and $\sigma_{w;z}^k$. We claim that if Alice puts a second cut in the first interval, her payoff $V_A(B) = \mu(B)$ will be at most $\frac{2}{3}-\frac{1}{6k}$. If more than one cut point lies in the first interval, then at most $k-2$ cuts are left in the other $k$ intervals. By Lemma \ref{lem:alternating-region-too-few-cuts}, at least one of the following must then happen:
    \begin{itemize}
        \item[(i)] Both a low-value interval and a high-value interval are contained in $B$
        \item[(ii)] Both a low-value interval and a high-value interval are contained in $A$
        \item[(iii)] Both $A$ and $B$ contain a low-value interval
        \item[(iv)] Both $A$ and $B$ contain a high-value interval
    \end{itemize}
    By Lemma \ref{lem:generalized-spiked-similarity}, Bob's valuations of $B$ differ by at most $\frac{2}{3k}w$ between $\sigma_{w;z}^k$ and $\sigma_0^k$. Since this difference pushes the valuation from above $1/2$ to below it, we have $|\int_B \sigma_0^k(x) \,\mathrm{d}x - 1/2| \leq \frac{2}{3k}w$. Therefore, using Lemma \ref{lem:generalized-spiked-bad-cut-suboptimal} in cases (i)-(iv):
    \begin{itemize}
        \item (i) and (iv) guarantee a high-value interval is in $B$, so we have:
        \begin{align}
            \mu(B) \leq \frac{2}{3} - \frac{7}{36k} + \frac{4}{3} \cdot \frac{2}{3k}w = \frac{2}{3} - \frac{7 - 32w}{36k} \leq \frac{2}{3} - \frac{19}{108k}
        \end{align}
        \item (ii) and (iii) guarantee a low-value interval is in $A$, so we have:
        \begin{align}
            \mu(A) \geq \frac{1}{3} + \frac{5}{12k} - \frac{2}{3} \cdot \frac{2}{3k}w = \frac{1}{3} + \frac{15-12w}{36k}
        \end{align}
        In terms of $\mu(B)=1-\mu(A)$, we then have:
        \begin{align}
            \mu(B) \leq \frac{2}{3} - \frac{15-12w}{30k} \leq \frac{2}{3} - \frac{59}{120k}
        \end{align}
    \end{itemize}
    In all cases we have $\mu(B) \leq \frac{2}{3} - \frac{1}{6k}$ as desired.

    That leaves partitions where one cut point goes through the spike and the other $k-1$ cut points are in the last $k$ intervals. The piece containing $x=0$ is worth more under $\sigma_{w;z}^k$ than under $\sigma_0^k$, as the higher-value portion of the spike is on the left side. Therefore, that piece must be $B$. Since the first cut is at least $\frac{5}{6}-\frac{1}{48} > \frac{19}{24}$ of the way through the first high-value interval, piece $B$ must contain at least $\frac{19}{24}$ of a high-value interval. As in cases (i) and (iv) above, then, we have $\mu(B) \leq \frac{2}{3} - \frac{1}{6k}$, which completes the proof.
\end{proof}

\kMyopicLowerBound*

\begin{proof}
    By Lemma \ref{lem:wlog-alice-valuation-warping}, we can assume without loss of generality that Alice's valuation is uniform. The desired Bob will then have a spiked valuation.

    Specifically, let $w = \frac{1}{4\sqrt{Tk}}$. Divide the interval $\left[\frac{5}{6}, \frac{11}{12}\right]$ into $\lfloor \frac{1}{2w} \rfloor$ intervals of equal length. Consider what would happen if Alice's strategy were matched against a myopic Bob with value density $\sigma_0^k$. There are two cases:
    \begin{itemize}
        \item If Alice eventually distinguishes each of the centers of the $\lfloor \frac{1}{2w} \rfloor$ intervals, let $z$ be the last such center to be distinguished (breaking ties arbitrarily) and let Bob have valuation $\sigma_{w;z}^k$. Up until that point, all of Alice's potentially distinguishing partitions saw choices consistent with $\sigma_0^k$. Each partition Alice made could only have distinguished $k$ of the centers, since each one's $2w$-width spike was contained within their interval and all distinguishing partitions have a cut point in the spike. As $z$ is one of the last centers to be distinguished, there must have been at least $\lfloor \frac{1}{2w} \rfloor/k -1$ prior cuts that distinguished the other centers. As $\frac{1}{2wk} = 2\sqrt{T/k} \geq 2$, we can neglect the $-1$. Therefore, by Lemma \ref{lem:generalized-spiked-distinguishers}, Alice accumulates at least $\lfloor \frac{1}{2wk} \rfloor \cdot \frac{1}{6k} \in \Omega\left(\frac{\sqrt{T}}{k^{3/2}}\right)$ regret.
        \item If the center of one of the $\lfloor \frac{1}{2w} \rfloor$ intervals is never distinguished, let $z$ be that center and let Bob have valuation $\sigma_{w;z}^k$. As Alice never distinguishes $z$, she only sees results consistent with $\sigma_0^k$. Therefore, by Lemma \ref{lem:generalized-spiked-stackelberg}, she receives at least $\frac{4}{27k}w$ regret in each round, for a total of $\frac{4}{27k}wT = \Omega\left(\frac{\sqrt{T}}{k^{3/2}}\right)$.
    \end{itemize}
    In either case, Alice's regret is $\Omega\left(\frac{\sqrt{T}}{k^{3/2}}\right)$, completing the proof.
\end{proof}

\knownAlphaLowerBound*

\begin{proof}
    By Lemma \ref{lem:wlog-alice-valuation-warping}, we can assume without loss of generality that Alice's valuation is uniform. The desired Bob will then have a spiked valuation and his strategy will be to pretend to have an unspiked valuation.

    Specifically, let $w = T^{\frac{\alpha-1}{3}}$. Divide the interval $\left[ \frac{5}{6}, \frac{11}{12} \right]$ into $\lfloor \frac{1}{2w} \rfloor$ intervals of equal length. Consider what would happen if Alice's strategy were matched against a myopic Bob with value density $\sigma_0^k$. Let $n = \lfloor kT^{\frac{2\alpha+1}{3}} \rfloor$. Since $\alpha > -1/2$, we have that $n$ is positive. We consider two cases:
    \begin{itemize}
        \item If, for every center of the $\lfloor \frac{1}{2w} \rfloor$ intervals, Alice eventually makes $n$ partitions that distinguish it, let $z$ be the last such center to receive an $n$th distinguishing partition (breaking ties arbitrarily) and let Bob have valuation $\sigma_{w;z}^k$. His strategy will be to behave as if his valuation is $\sigma_0^k$ and he is playing myopically for the first $n$ partitions that distinguish $z$, then play myopically with his true valuation. By Lemma \ref{lem:generalized-spiked-similarity}, any piece $S$ of a distinguishing partition satisfies $|\int_S \sigma_{w;z}^k(x) \,\mathrm{d}x - 1/2| \leq \frac{2}{3k}w$, so his total regret compared to playing myopically is at most $n \cdot \frac{4}{3k}w \in O\left(T^{\alpha}\right)$. On the other hand, by Lemma \ref{lem:generalized-spiked-distinguishers}, Alice receives at least $\frac{1}{6k}$ regret for each one of her distinguishing partition. Each partition can only distinguish $k$ centers, since each center's spike is contained in its interval and every distinguishing partition must place a cut point in the spike. She therefore makes at least $\frac{n}{k} \lfloor \frac{1}{2w} \rfloor$ such cuts for a total of $\frac{n}{12wk^2} \in \Omega\left(\frac{T^{\frac{2+\alpha}{3}}}{k}\right)$ regret.
        \item If there exists an interval whose center $z$ is distinguished fewer than $n$ times, let Bob have valuation $\sigma_{w;z}^k$ but behave as if his valuation were $\sigma_0^k$ and he were playing myopically. As in the first case, Bob's regret comes from at most $n$ rounds with regret at most $\frac{4}{3k}w$ each, for a total of $n \cdot \frac{4}{3k}w \in O\left(T^{\alpha}\right)$. In this case, Alice only ever sees results consistent with $\sigma_0$, so by Lemma \ref{lem:generalized-spiked-stackelberg} she receives at least $\frac{4}{27k}w$ regret in all $T$ rounds. Therefore, her regret is at least $\frac{4wT}{27k} \in \Omega\left(\frac{T^{\frac{2+\alpha}{3}}}{k}\right)$.
    \end{itemize}
    In both cases Alice's regret is $\Omega\left(\frac{T^{\frac{2+\alpha}{3}}}{k}\right)$, which completes the proof.
\end{proof}

\subsection{Lower Bounds for Non-Myopic Bob with Private Learning Rate}\label{sec:unknown-alpha-lower-bounds}

There is actually only one lower bound in this section, as the same short argument works for all $k \geq 2$.

\unknownAlphaLowerBound*

\begin{proof}
    The Bob we construct will have one valuation but pretend to play myopically under another. He will stay within his $O(T^{\beta})$ regret bound by switching to playing myopically after accumulating enough regret. However, the conditions on $S_A$ make it impossible for Alice to reach that point.

    Specifically, by Lemma \ref{lem:wlog-alice-valuation-warping} we can assume Alice's valuation is the uniform $v_A(x)=1$ for all $x \in [0, 1]$. Let $v_B^1$ and $v_B^2$ be the following value densities:
    \begin{align}
        v_B^1(x) = \begin{cases}
            1/2 & \text{if } x \in [0, 1/2] \\
            3/2 & \text{if } x \in (1/2, 1]
        \end{cases}  \qquad \mbox{and} \qquad 
        v_B^2(x) = \begin{cases}
            1/4 & \text{if } x \in [0, 1/2] \\
            7/4 & \text{if } x \in (1/2, 1]
        \end{cases}
    \end{align}
    Let $c > 0$ be such that $S_A$ guarantees Alice at most $c \cdot T^{\beta}$ regret against a Bob with regret $O(T^{\alpha})$. Then we define the following two Bobs:
    \begin{itemize}
        \item $\textsc{Bob}_1$ has value density $v_B^1$ and plays myopically.
        \item $\textsc{Bob}_2$ has value density $v_B^2$. When choosing his piece for round $t$, he computes his cumulative regret across rounds $1, \ldots, t-1$. If this regret is at most $6cT^{\beta}$, choose the piece $\textsc{Bob}_1$ would choose. Otherwise, play myopically.
    \end{itemize}
    We first analyze what $S_A$ does against $\textsc{Bob}_1$. Call a partition \emph{distinguishing} if Bobs with value densities $v_B^1$ and $v_B^2$ prefer different pieces of it. We claim that against $\textsc{Bob}_1$, Alice makes at most $6cT^{\beta}$ distinguishing partitions. Since $\textsc{Bob}_1$ plays myopically, he never gets any regret, so $S_A$ must ensure Alice's regret is at most $c T^{\beta}$. Her Stackelberg value against $\textsc{Bob}_1$ is $2/3$ since he would be satisfied with just $1/3$ length in $(1/2, 1]$.
    
    Consider an arbitrary distinguishing partition $(S, \overline{S})$ and assume without loss of generality that $\overline{S}$ is the piece $\textsc{Bob}_1$ would take. Let $\ell = \mu(S \cap [0, 1/2])$ and $r = \mu(S \cap (1/2, 1])$. We then have:
    \begin{align}
        \frac{1}{2} &\geq V_B^1(S) = \frac{1}{2}\ell + \frac{3}{2}r \\
        \frac{1}{2} &\leq V_B^2(S) = \frac{1}{4}\ell + \frac{7}{4}r
    \end{align}
    Subtracting twice the second inequality from three times the first, we get:
    \begin{align}
        \frac{1}{2} &\geq \ell + r
    \end{align}
    But $\ell + r = V_A(S)$, so when Bob takes $\overline{S}$ Alice gets at least $1/6$ Stackelberg regret. Therefore, she cannot make more than $6cT^{\beta}$ distinguishing partitions against $\textsc{Bob}_1$.

    We now analyze what $S_A$ does against $\textsc{Bob}_2$. As long as his regret stays under $6cT^{\beta}$, he is indistinguishable from $\textsc{Bob}_1$. He only gets regret when Alice makes distinguishing partitions. But as long as he behaves like $\textsc{Bob}_1$, Alice cannot make more than $6cT^{\beta}$ distinguishing partitions, so she never reaches this regret threshold. Therefore, when using $S_A$ against $\textsc{Bob}_2$, Alice only sees choices consistent with $\textsc{Bob}_1$. She can then only get her Stackelberg value against $\textsc{Bob}_1$ of $2/3$ payoff per round, but the true Stackelberg value against $\textsc{Bob}_2$ is $5/7$. Therefore, Alice's overall regret against $\textsc{Bob}_2$ is at least $T/21 \in \Omega(T)$.
\end{proof}

\section{The RW Query Model} \label{app:rw-query-model}

In this section we prove Theorem \ref{thm:rw-query-model-bounds}. We start with the upper bound.

\begin{restatable}{lemma}{approximateStackelbergDiscretization}
    \label{lem:approximate-stackelberg-discretization}
    Consider the $k$-cut game for $k \geq 1$. Let $x_1, \ldots, x_n$ be an increasing sequence in $[0, 1]$ such that $x_1 = 0$ and $x_n = 1$, and let $v>0$ be such that $V_A([x_i, x_{i+1}]) \leq v$ for all $i \in [n-1]$. Then there exists a $(kv)$-Stackelberg allocation with cut points in $\{x_1, \ldots, x_n\}$.
\end{restatable}

\begin{proof}
    Let $y_1, \ldots, y_k$ be the cut points of a Stackelberg allocation. Let $z_1, \ldots, z_k$ be obtained from the $y$'s by rounding each to the nearest $x_i$ in the direction that increases Bob's piece in the Stackelberg allocation. We claim that cutting at the $z$'s is the desired $(kv)$-Stackelberg cut.
    
    As Bob's piece only gets bigger, he will still take the same piece as he did under the $y_i$. All that remains is to bound how much Alice loses. Each cut point was moved by at most $v$ by Alice's valuation, so in total she values her side of the cut as at most $kv$ less. Therefore, it is $(kv)$-Stackelberg as desired.
\end{proof}

\begin{restatable}{prop}{epsStackelbergUpper}
    \label{prop:eps-stackelberg-upper}
    The query complexity of finding an $\eps$-Stackelberg cut in the $k$-cut game is $O\left(\frac{k}{\eps}\right)$.
\end{restatable}

\begin{proof}
    The algorithm is as follows. Let $n = \left\lceil \frac{k}{\eps}\right\rceil$. Discretize the cake into $n$ pieces of value $1/n$ each to Alice and query Bob's value of each one. Return whichever $k$-cut maximizes Alice's payoff while still giving Bob at least half the cake by his valuation. This process $n$ cut queries on Alice's valuation and $n$ eval queries on Bob's for a total of $2n \in O(k/\varepsilon)$.

    We now show that this cut is $\eps$-Stackelberg. By Lemma \ref{lem:approximate-stackelberg-discretization}, there exists a $(k/n)$-Stackelberg cut among the ones the algorithm considers. The one the algorithm returns can only be better than that one, so it too is $(k/n)$-Stackelberg. By the choice of $n$, it is $\eps$-Stackelberg as desired.
\end{proof}

We now prove the lower bound using the same construction as section \ref{sec:known-alpha-lower-bounds}.

\begin{restatable}{prop}{epsStackelbergRandomizedLower}
    \label{prop:eps-stackelberg-randomized-lower}
    The randomized query complexity of finding an $\eps$-Stackelberg cut in the $k$-cut game for $k \geq 2$ is $\Omega(1/\eps)$.
\end{restatable}
\begin{proof}
    Rather than directly lower-bound the performance of an arbitrary randomized algorithm, we invoke Yao's minimax principle \cite{yao_minimax} to construct a worst-case input distribution for deterministic algorithms. Specifically, up to a constant factor, the randomized query complexity (with error probability $1/10$) can be lower-bounded by the best performance of any deterministic algorithm on a hard distribution of inputs (with error probability $1/5$).

    Our hard distribution will be made up of the spiked valuation functions of Definition \ref{def:generalized-spike}. Their key property is that any near-Stackelberg cut must include a cut point in a small unknown interval. If many such intervals are possible and the correct one can only be found by directly querying it, the query complexity must be high. Accordingly, let $w = 14\eps$ and let $\mathcal{D}$ be the distribution where:
    \begin{itemize}
        \item Alice's value density is the uniform $v_A(x)=1$ for all $x \in [0, 1]$, and
        \item Bob's value density is $\sigma_{w;z}^2$ for uniformly random $z \in [5/6, 11/12]$.
    \end{itemize}
    Against an input from $\mathcal{D}$, each query made does not narrow down the range of possible values of $z$ very much:
    \begin{itemize}
        \item An eval query for Bob's value of an interval $[0, x]$ will be indistinguishable from an unspiked valuation $\sigma_0^k$ unless $x/\ell^k$ is in the interval $(z-w, z+w)$. Any other query therefore eliminates $(x/\ell^k-w, x/\ell^k+w)$ from consideration, for a total of $2w$ of the interval $[5/6, 11/12]$.
        \item A cut query for the point $x$ such that Bob values $[0, x]$ as $\alpha$ will similarly be indistinguishable from an unspiked valuation $\sigma_0^k$ unless $x/\ell^k$ is in the interval $(z-w, z+w)$. As $x$ is proportional to $\alpha$, the algorithm gains nothing by querying by $\alpha$ instead of $x$ as in an eval query.
    \end{itemize}
    Therefore, each query only reduces the search space by $2w$. If a deterministic algorithm makes fewer than $\frac{1}{24 \cdot 2w}$ queries, then, its chances of seeing anything inconsistent with the unspiked $\sigma_0^k$ are at most $1/2$. If it fails to find any inconsistency, then its probability of returning a cut that goes through the spike is negligible for $w << 1$. If it does not return a cut that goes through the spike, then by Lemma \ref{lem:generalized-spiked-stackelberg} it falls at least $\frac{2w}{27} > \eps$ short of the Stackelberg value. Therefore, any deterministic algorithm with success probability greater than $4/5$ must make more than $\frac{1}{24 \cdot 2w} \in \Omega(1/\eps)$ queries with constant probability, so its expected runtime is $\Omega(1/\eps)$.
\end{proof}

\begin{proof}[Proof of Theorem \ref{thm:rw-query-model-bounds}]
    The upper bound follows by Proposition~\ref{prop:eps-stackelberg-upper}. The lower bound follows by Proposition~\ref{prop:eps-stackelberg-randomized-lower}.
\end{proof}

\end{document}